\documentclass[11pt]{article}
\usepackage{preamble}
\usetikzlibrary{positioning}
\usepackage{xcolor}
\usepackage{placeins}
\definecolor{gacEdge}{RGB}{0,150,170}
\definecolor{ghostEdge}{RGB}{200,200,200}
\definecolor{treeEdge}{RGB}{33,113,181}
\definecolor{pathEdge}{RGB}{228,26,28}
\definecolor{vertexA}{RGB}{235,250,252}
\definecolor{vertexB}{RGB}{230,244,255}
\definecolor{vertexC}{RGB}{255,235,235}  
\definecolor{diaggray}{RGB}{110,110,110}
\definecolor{xblue}{RGB}{38,92,155}
\definecolor{zzorange}{RGB}{210,120,40}

\definecolor{datablue}{HTML}{2166AC}
\definecolor{bandblue}{HTML}{9ECAE1}
\definecolor{refcoral}{HTML}{D6604D}
\definecolor{gridgray}{HTML}{CCCCCC}
\usepackage{pgfplots}
\pgfplotsset{compat=1.18}
\pgfplotsset{plot coordinates/math parser=false}
\usepackage{subcaption}

\usepackage[margin=1in]{geometry}
\usepgfplotslibrary{groupplots}
\usepackage{algorithm}
\algtext*{EndIf}
\algtext*{EndFor}
\algrenewcommand{\algorithmicrequire}{\textbf{Input:}}
\algrenewcommand{\algorithmicensure}{\textbf{Output:}}
\algrenewcommand{\algorithmicrequire}{\textbf{Input:}}
\algrenewcommand{\algorithmicensure}{\textbf{Output:}}
\usepackage{amsmath, amssymb}
\usepackage{enumitem}
\usepackage{placeins}
\usepackage[T1]{fontenc}

\numberwithin{equation}{section}

\pgfplotsset{
  paperplot/.style={
    width              = 8cm,
    height             = 6cm,
    axis lines         = left,
    axis line style    = {semithick, black},
    tick align         = outside,
    major tick length  = 4pt,
    minor tick num     = 3,
    xtick style        = {black, semithick},
    ytick style        = {black, semithick},
    tick label style   = {font=\small},
    label style        = {font=\small},
    title style        = {font=\small\bfseries},
    legend style       = {
      font        = \small,
      draw        = gray!50,
      fill        = white,
      fill opacity= 0.9,
      text opacity= 1,
      inner sep   = 5pt,
      row sep     = 1pt,
    },
    grid               = both,
    grid style         = {gridgray, very thin},
    minor grid style   = {gridgray!40, ultra thin},
    clip               = false,
  },
}

\allowdisplaybreaks

\usepackage{defs}
\usepackage{float}

\usepackage{textgreek}
\usepackage{geometry}

\vspace{1.5cm}

\newcommand{\opt}{\mathsf{OPT}}
\newcommand{\mps}{\mathsf{MPS}}
\newcommand{\mpo}{\mathsf{MPO}}
\newcommand{\sr}{\mathrm{SR}}
\newcommand{\ttn}{\mathsf{TTN}}

\usepackage[titles]{tocloft}

\usepackage{stackengine}
\stackMath
\newcommand\tsup[2][2]{%
 \def\useanchorwidth{T}%
  \ifnum#1>1%
    \stackon[-.5pt]{\tsup[\numexpr#1-1\relax]{#2}}{\scriptscriptstyle\sim}%
  \else%
    \stackon[.5pt]{#2}{\scriptscriptstyle\sim}%
  \fi%
}
\title{Proper Agnostic Learning of Matrix Product States and Tree Tensor Networks
}

\author{
Constantin Cedillo Vayson de Pradenne \\
Caltech, Harvard University \\
\href{mailto:ccedillo@caltech.edu}{\texttt{ccedillo@caltech.edu}}
\and
Jordan Cotler \\
Harvard University \\
\href{mailto:jcotler@fas.harvard.edu}{\texttt{jcotler@fas.harvard.edu}}
}

\title{Proper Agnostic Learning of Matrix Product States \\ and Tree Tensor Networks}

\date{September 24, 2026}
\begin{document}
%\captionsetup[algorithm]{hypcap=false}
\emergencystretch=3em
\hbadness=10000
\hfuzz=60pt

\pagestyle{empty}
{
  \renewcommand{\thispagestyle}[1]{}
  \maketitle
\begin{abstract}
We establish proper agnostic learning of matrix product states and tree tensor networks. Given copies of an arbitrary quantum state $\rho$, our algorithms return a state $|\psi\rangle$ of chosen bond dimension such that $\langle\psi| \rho |\psi\rangle$ is within $\varepsilon$ of the optimum over the model class, without assuming that $\rho$ itself belongs to or is well approximated by that class. The main idea is to use improper learning to compress the mixed-state objective, reducing proper agnostic learning to optimization against a finite collection of explicitly specified pure states. We then introduce a comparator-dual compression procedure that reduces the bond dimension of these targets while uniformly preserving their overlaps with all bounded-bond comparators, with an error independent of system size. For matrix product states, this gives polynomial copy complexity in the system size, local dimension, bond dimension, and $1/\varepsilon$, together with polynomial runtime in the system size for fixed local dimension, bond dimension, and accuracy. The same framework yields proper agnostic learning of tree tensor networks on bounded-degree trees, with both copy complexity and runtime polynomial in the system size when the remaining parameters are fixed.
\end{abstract}

}

\clearpage
\pagestyle{plain}
\pagenumbering{arabic}

\tableofcontents

\newpage

\section{Introduction}

Generic many-body quantum states require exponentially many parameters to describe. However, many states of physical interest admit accurate descriptions with only polynomially many parameters. Matrix product states (MPSs) provide an important example, since ground states of gapped one-dimensional local Hamiltonians can be efficiently approximated by MPSs. This raises a natural learning problem: given copies of an unknown many-body state, can one efficiently find a good approximation to it within a chosen low-complexity class, without first reconstructing the full state? This important problem has been studied extensively for tensor network states, including efficient tomography and reconstruction algorithms for MPSs and matrix product operators (MPOs) under the assumption that the unknown state itself belongs to the corresponding tensor network class~\cite{CramerEtAl2010,BaumgratzEtAl2013}.

The more physical problem is to remove the assumption that the unknown state itself belongs to the tensor network class, while still requiring the learned state to lie in the chosen class. One then asks for a state $|\psi\rangle$ in the class whose score $\langle\psi|\rho|\psi\rangle$ is nearly optimal over all states in the class.  We refer to the states in the model class as ``comparators''.  This is the proper agnostic tomography problem: it is ``agnostic'' because no structural assumption is made on $\rho$, and ``proper'' because the output is required to belong to the chosen model class. The latter requirement matters if the goal is to learn a genuinely low-complexity description. By contrast, an ``improper'' learner need only return some state $|b\rangle$, possibly an MPS of much larger bond dimension, whose score is nearly as large as the maximum score achievable by a bond-$D$ MPS. Such a guarantee does not imply that $|b\rangle$ is close to any bond-$D$ MPS, so truncating $|b\rangle$ to bond dimension $D$ need not preserve its score. Proper agnostic learning of product states and improper agnostic learning of MPSs were established in~\cite{BakshiEtAl2025}, while improper learning of tree tensor networks (TTNs) and more general tensor networks was obtained in~\cite{CaroMcHughStrelchuk2026}. The open problem was whether MPSs themselves can be learned properly in the agnostic setting. We solve this problem here, and also obtain proper agnostic learning of bounded-degree TTNs.

More precisely, given copies of an arbitrary density operator $\rho$ and a chosen bond dimension $D$, our algorithm returns a normalized state $\widehat\psi\in\mps_D$ such that, with probability at least $1-\delta$,
\begin{equation}
\langle\widehat\psi|\rho|\widehat\psi\rangle \geq \opt_D(\rho)-\varepsilon,
\end{equation}
where $\opt_D(\rho)$ is the maximum score over bond-$D$ MPSs. The copy complexity is polynomial in the system size, local dimension, bond dimension, and inverse accuracy, while the runtime is polynomial in the system size for fixed local dimension, bond dimension, and accuracy. We prove an analogous proper agnostic learning result for bond-$D$ TTNs on bounded-degree trees.

The basic difficulty is that an improper learner may achieve nearly optimal score while outputting an MPS whose bond dimension is much larger than the desired one. Simply truncating such an output to low bond dimension can substantially reduce its score. Our first step therefore does not compress the improper hypothesis itself. Instead, we use repeated calls to the improper learner to identify a low-dimensional subspace such that compressing the objective $\rho$ onto this subspace changes the score of every low-bond-dimension comparator by only a small amount. Equivalently, the construction prunes directions that are irrelevant for optimization over the model class, while retaining all directions that can carry appreciable score for some comparator.

For a target accuracy $\varepsilon$, we show that repeated calls to the improper learner construct a subspace of dimension $O(\varepsilon^{-2})$ such that compressing $\rho$ onto this subspace changes the score of every comparator by at most $O(\varepsilon)$. The $O(\varepsilon^{-2})$ dimension bound follows because each new orthogonal direction found by the learner carries $\Omega(\varepsilon^2)$ weight under $\rho$ while the total weight is at most $\tr(\rho)=1$.

The resulting low-dimensional mixed-state problem can then be reduced to a finite collection of optimization problems against explicit pure-state targets. These targets may still have bond dimension much larger than $D$, so we compress them while preserving their overlaps with every bond-$D$ MPS. To formalize this, define
\begin{equation}
\|x\|_{\mathrm{MPS}(D)} := \sup_{\psi\in\mps_D}
|\langle \psi|x\rangle|.
\end{equation}
This norm measures the largest overlap of $x$ with a bond-$D$ MPS. We show that every vector $u$ can be approximated by an MPS $\widetilde u_K$ of bond dimension at most $K$ satisfying
\begin{equation}
\|u-\widetilde u_K\|_{\mathrm{MPS}(D)}
\leq
\|u\|_2\sqrt{\frac{D}{K+1}}.
\end{equation}
For $\|u\|_2\leq1$, taking $K=O(D/\eta^2)$ therefore suffices to preserve the overlap of the target with every bond-$D$ MPS to accuracy $\eta$, independently of the length of the chain. After this compression, a classical optimization over bond-$D$ MPSs returns a proper hypothesis.

The construction is a modification of the usual left-to-right singular-value truncation. At each cut, in addition to discarding the singular-value tail, we uniformly decrease the retained squared singular values by the square of the first discarded singular value. The discarded component then has uniformly bounded singular values, and the corresponding decrease in the squared Euclidean norm of the retained state telescopes over the chain. Since a bond-$D$ MPS has Schmidt rank at most $D$ across each cut, this yields the bound above without a factor proportional to the number of sites. The same construction also gives an approximation guarantee in Euclidean norm: increasing the bond dimension from $D$ to $D+O(D/\alpha)$ suffices to obtain squared error within a factor $1+\alpha$ of the best bond-$D$ approximation, independently of the chain length.

After this compression, we optimize against the target by dynamic programming, building a bond-$D$ MPS from left to right. At each cut, the dependence of the final overlap on the prefix constructed so far is captured by a bounded-dimensional matrix obtained by contracting the candidate and target over the processed sites. The update of this matrix is contractive, so we can discretize both the local tensors and these matrices and retain only one partial MPS for each discretized matrix value. The output is proper by construction since every state considered in the search has bond dimension at most $D$. Combining this procedure with the compression result gives a proper optimizer for arbitrary explicit MPS targets and, together with the reduction above, yields the agnostic learning theorem.

For TTNs, we first choose an ordering of the vertices of a bounded-degree tree such that every rooted subtree forms a contiguous interval and, for any initial segment of this ordering, only $O(\Delta_T\log n)$ tree edges connect that segment to the remaining vertices. In this ordering, a bond-$D$ TTN can be represented as an MPS of bond dimension $D^{O(\Delta_T\log n)}$, which is polynomial in the system size for fixed $D$ and tree degree. This allows us to use the improper MPS learner in the first stage. The compression procedure extends by processing the tree from the leaves toward the root. The dynamic program is modified similarly: instead of keeping track of the part of the chain already processed, it keeps track of the information from each processed subtree that can affect the final overlap with the target.

The same measurements on $\rho$ can also be reused across bond dimensions. If the relevant subspace is constructed for a maximum bond dimension $D_{\max}$, then no additional copies of $\rho$ are needed to find nearly optimal MPSs or estimate the corresponding optimal scores for any $D\leq D_{\max}$. Thus a single set of measurements determines how the best achievable score varies with the allowed bond dimension. By vectorization, the MPS results also apply to Hilbert--Schmidt-normalized MPOs.

A quantum state may also admit a useful description as a coherent superposition of several simple states, sometimes called branches. Such branch structure can arise naturally in many-body systems through decoherence and the formation of redundant or effectively orthogonal records~\cite{Riedel2017,riedel2025wavefunction, taylor2025non, taylor2025wavefunction, pilatowsky2026emergent}. For example, a GHZ state is a superposition of two product branches whose coherence is invisible to observables supported on fewer than all sites, while an MPS representation of the full state need not reveal this decomposition. We therefore consider states $|\Phi\rangle=\sum_{a=1}^r c_a|\psi_a\rangle$ with normalized bond-$D$ MPS branches $|\psi_a\rangle$, subject to a local noninterference condition suppressing off-diagonal matrix elements between distinct branches for observables supported on at most $k$ sites. At zero interference, all such observables have the same expectation values in the superposition and the corresponding incoherent mixture. Our algorithm returns the branches $|\psi_a\rangle$ and coefficients $c_a$ with nearly optimal score, preserving the bond bounds while allowing the interference to increase by at most a prescribed amount.

Several ingredients of our approach are related to earlier work. The improper learner used here builds on agnostic learning algorithms for MPSs, with subsequent extensions and improvements for tensor networks on more general graphs~\cite{BakshiEtAl2025,LinChiaHung2025,CaroMcHughStrelchuk2026}. Our improper-to-proper reduction is also related to quantum agnostic boosting and learning through structured decompositions~\cite{ArunachalamEtAl2026Boosting,ArunachalamDutt2026Structure,ArunachalamDutt2026Extent}; here, however, we use improper learning to construct a low-rank operator that approximately preserves the score of every proper comparator. The pure-target optimizer is related to earlier dynamic-programming algorithms for optimization over fixed-bond-dimension MPSs~\cite{SchuchCirac2010,AharonovAradIrani2010}, while the singular-value shrinkage underlying comparator-dual compression is reminiscent of tensor-rounding methods and Frequent Directions~\cite{Oseledets2011,Grasedyck2010,Liberty2013,GhashamiEtAl2016}. The key difference is that our compression controls overlaps uniformly over all bounded-bond comparators, with an error independent of the system size. We discuss these connections further in Appendix~\ref{sec:related-work}.

We now turn to the precise statements of our results.

\section{Results}
\subsection{Problem formulation and main result}
\label{sec:results-setup}

Throughout this section, an MPS means an open-boundary matrix product state. Consider a chain of qudits with $\mathcal H_i\cong\mathbb C^d$ and $\mathcal H:=\bigotimes_{i=1}^n\mathcal H_i$, and let $\mps_D$ denote the normalized MPSs of bond dimension at most $D$. Equivalently, by the cut-rank characterization, $\psi\in\mps_D$ if and only if
\begin{equation*}
\operatorname{rank}_{\mathcal H_{1:i}\,|\,\mathcal H_{i+1:n}}(\psi)\leq D \qquad \text{for every }1\leq i\leq n-1.
\end{equation*}
For a positive semidefinite operator $X\succeq0$, write
\begin{equation*}
\opt_D(X):=\max_{\psi\in\mps_D}\langle\psi|X|\psi\rangle.
\end{equation*}

Given copies of an arbitrary density operator $\rho$, the proper agnostic learning problem is to output a state $\widehat\psi\in\mps_D$ whose score is nearly optimal within the class,
\begin{equation}
\langle\widehat\psi|\rho|\widehat\psi\rangle \geq \opt_D(\rho)-\varepsilon.
\label{eq:results-agnostic-objective}
\end{equation}
No assumption is made that $\rho$ is an MPS, or even that it is close to one.  The class $\mps_D$ serves only as a class of comparators.

An improper learner is allowed to return an efficiently represented state of much larger bond dimension, provided that its score is close to $\opt_D(\rho)$.  A proper learner must return a member of $\mps_D$.  This distinction cannot in general be removed by truncating the improper output: a near-optimal value of $\langle b|\rho|b\rangle$ does not imply that $b$ is close to any bond-$D$ MPS.

Our main result is that properness can nevertheless be enforced with polynomial overhead in the system size.

\begin{theorem}[Proper agnostic learning of MPSs]
\label{thm:proper-agnostic-mps-results}
For every $n,d,D\in\mathbb N$ and every $\varepsilon,\delta\in(0,1/4)$, there is a randomized quantum algorithm which, given copies of an arbitrary density operator $\rho$ on $(\mathbb C^d)^{\otimes n}$, outputs a classical description of a normalized state $\widehat\psi\in\mps_D$ such that, with probability at least $1-\delta$,
\begin{equation}
\langle\widehat\psi|\rho|\widehat\psi\rangle \geq \opt_D(\rho)-\varepsilon.
\label{eq:mps-proper-agnostic-guarantee-results}
\end{equation}
The copy complexity is bounded by $\poly\!\left(n,d,D,1/\varepsilon,\log(1/\delta)\right).$ The runtime is bounded by $\poly\!\left(n,d,D,1/\varepsilon,\log(1/\delta)\right)\left(\frac{Cn\sqrt D}{\varepsilon}\right)^{O(D^2/\varepsilon^2+dD^2)}.$
In particular, the runtime is polynomial in $n$ for fixed $d,D,\varepsilon$.
\end{theorem}

The proof has two parts. First, a general reduction uses an improper learner to identify a low-dimensional subspace such that compressing the objective to this subspace changes the score of every proper comparator by only a small amount. This reduces the mixed-state optimization problem to finitely many pure-target problems. Second, for MPSs, we solve these pure-target problems by compressing the target in a norm adapted to bond-$D$ comparators and then running a dynamic program that retains only a bounded-dimensional cross environment at each cut.

\subsection{From improper learning to a finite relevant subspace}
\label{sec:results-finite-subspace}

We first isolate the part of the argument that does not depend on MPSs. Let $\mathcal C$ be a nonempty compact class of normalized comparator states in a finite-dimensional Hilbert space $\mathcal H$.\footnote{Since $\mathcal H$ is finite-dimensional and we take states to be normalized, compactness here is equivalent to closedness. We assume compactness so that the relevant optimization problems attain their maxima.} Let $(\mathcal V_R)_{R\geq1}$ be a family of explicit vector classes indexed by a complexity parameter $R$, with
\begin{equation*}
\mathcal V_1\subseteq\mathcal V_2\subseteq\cdots\subseteq\mathcal H,
\end{equation*}
and define the corresponding normalized classes
\begin{equation*}
\mathcal A_R:=\{u\in\mathcal V_R:\|u\|_2=1\}.
\end{equation*}
Thus increasing $R$ enlarges the class of allowed explicit vectors. We assume that $\mathcal C\subseteq\mathcal A_{R'}$ for some fixed $R'$, and that finite linear combinations of vectors in the classes $\mathcal V_R$ remain in a class $\mathcal V_{\widetilde R}$, where $\widetilde R$ can be bounded computably in terms of the original complexity parameters. The two main algorithmic primitives are an improper learner competing with the fixed class $\mathcal A_{R'}$ and a proper optimizer against explicit pure targets. We further assume that the explicit vectors produced by these procedures can be manipulated efficiently: in particular, that one can compute overlaps, prepare the corresponding states, project onto the span of known orthonormal vectors, and estimate matrix elements of $\rho$. Definition~\ref{def:properization-data} formalizes these
assumptions.

The improper learner is used to discover directions that
matter to the proper class. Suppose that mutually orthonormal explicit states $s_1,\ldots,s_j$ have already been found, and set
\begin{equation*}
P_j:=\sum_{\ell=1}^j|s_\ell\rangle\langle s_\ell|, \qquad Q_j:=I-P_j.
\end{equation*}
When $\mu_j:=\tr(Q_j\rho)>0$, postselection prepares the
normalized residual state
$\overline\rho_j:=Q_j\rho Q_j/\mu_j$.
For every original comparator $\phi\in\mathcal C$,
\begin{equation}
\langle\phi|\overline\rho_j|\phi\rangle =\frac{\langle Q_j\phi|\rho|Q_j\phi\rangle}{\mu_j}.
\label{eq:results-fixed-comparator-identity}
\end{equation}
Thus a comparator with substantial residual score is detectable by the learner using the original comparator parameter $R'$. 

If the improper learner finds a state $b_j$ with sufficiently large residual score, we add the normalized direction
\begin{equation*}
s_{j+1}:=\frac{Q_jb_j}{\|Q_jb_j\|_2}
\end{equation*}
to the list. Every added direction is orthogonal to the previous ones and has a definite amount of $\rho$-weight. Since $\tr(\rho)=1$, at most $2/\theta$ such directions can be added. The resulting statement is the following.

\begin{theorem}[Finite relevant subspace]
\label{thm:general-finite-relevant-subspace-results}
Let $(\mathcal C,(\mathcal V_R)_{R\geq1})$ satisfy the assumptions above, let $\rho$ be an arbitrary density operator, and let $\theta\in(0,1)$.  On the event that all estimation and improper-learning calls succeed, there is an algorithm which constructs mutually orthonormal explicit states $s_1,\ldots,s_m$, with $m\leq2/\theta$, and the projectors
\begin{equation*}
P:=\sum_{a=1}^m|s_a\rangle\langle s_a|, \qquad Q:=I-P,
\end{equation*}
such that
\begin{equation}
\langle Q\phi|\rho|Q\phi\rangle \leq\theta\qquad \text{for every }\phi\in\mathcal C.
\label{eq:general-residual-smallness-results}
\end{equation}
Consequently, for $A:=P\rho P$,
\begin{equation}
\sup_{\phi\in\mathcal C}\left|\langle\phi|\rho|\phi\rangle-\langle\phi|A|\phi\rangle\right|\leq 2\sqrt\theta.
\label{eq:general-core-approximation-results}
\end{equation}
\end{theorem}

The theorem compresses the objective rather than the improper hypotheses.  For $\theta=\Theta(\varepsilon^2)$, the relevant subspace has dimension $m=O(\varepsilon^{-2})$, independent of the dimension of the original Hilbert space.  Equation~\eqref{eq:general-core-approximation-results} says that, as far as optimization over $\mathcal C$ is concerned, $\rho$ may be replaced by its compression to this subspace.

It remains to optimize over the compressed operator.  Let $S:\mathbb C^m\to\mathcal H$ be the isometry $Se_a=s_a$, and let $G:=S^\dagger\rho S.$ The entries $G_{ab}=\langle s_a|\rho|s_b\rangle$ can be estimated from copies of $\rho$.  From these estimates we construct a positive semidefinite matrix $\widehat G$ such that, for $\widehat A:=S\widehat G S^\dagger,$ one has $\|\widehat A-A\|_\infty\leq\xi$.  Diagonalize $\widehat A=\sum_{j=1}^m\lambda_j|u_j\rangle\langle u_j|$ where $\lambda_j\geq0$, and for every unit vector $c=(c_1,\ldots,c_m)\in\mathbb C^m$, define
\begin{equation}
v_c:=\sum_{j=1}^m\sqrt{\lambda_j}\,c_ju_j.
\label{eq:results-pure-target-family}
\end{equation}
Then, for every normalized $\phi$,
\begin{equation}
\langle\phi|\widehat A|\phi\rangle=\max_{\|c\|_2=1}|\langle v_c|\phi\rangle|^2.
\label{eq:results-mixed-to-pure}
\end{equation}
Thus optimization of the mixed-state score is reduced to optimization of overlap with explicit pure targets.

The identity above still involves a continuous maximization over the coefficient vector $c$. To reduce the problem to finitely many pure targets, we discretize the unit sphere of coefficient vectors by a finite net. Before doing so, we discard the small eigenvalues of $\widehat A$ so that this sphere has low dimension. Let $\widehat A_\tau$ retain only the eigenvalues greater than $\tau>\xi$, and write
\begin{equation*}
\widehat A_\tau=\sum_{j=1}^{\ell}\lambda_j|u_j\rangle\langle u_j|,
\qquad
\ell:=\operatorname{rank}(\widehat A_\tau).
\end{equation*}
Since $\tr(A)\leq1$, Lemma~\ref{lem:general-spectral-cutoff} gives
\begin{equation*}
\|\widehat A_\tau-A\|_\infty\leq\xi+\tau,\qquad
\ell\leq\frac{1}{\tau-\xi}.
\end{equation*}
For each unit vector $c=(c_1,\ldots,c_\ell)\in\mathbb C^\ell$, define
\begin{equation*}
v_c:=\sum_{j=1}^{\ell}\sqrt{\lambda_j}\,c_j u_j.
\end{equation*}
Then, for every normalized $\phi$,
\begin{equation*}
\langle\phi|\widehat A_\tau|\phi\rangle
=\max_{\|c\|_2=1}|\langle v_c|\phi\rangle|^2.
\end{equation*}
Thus optimizing the truncated mixed-state objective is equivalent to optimizing the overlap with the family of pure targets $v_c$, indexed by the unit sphere in $\mathbb C^\ell$. We can therefore approximate this continuous family by a finite net of coefficient vectors.

A Euclidean $\zeta/16$-net of the unit sphere in $\mathbb C^\ell$ has size at most $(80/\zeta)^{2\ell}$.
The corresponding pure-target problems yield a proper optimizer for $\widehat A_\tau$ to additive error $\zeta$.
Transferring back to $A$ adds at most $2(\xi+\tau)$. Choose
\begin{equation*}
\theta:=\left(\frac{\varepsilon}{16}\right)^2,\qquad\xi:=\frac{\varepsilon}{32},
\qquad\tau:=\frac{\varepsilon}{8},\qquad\zeta:=\frac{\varepsilon}{4}.
\end{equation*}
The total error is
\begin{equation*}
4\sqrt\theta+2(\xi+\tau)+\zeta=\frac{13\varepsilon}{16}<\varepsilon.
\end{equation*}
Although the relevant subspace has dimension $O(\varepsilon^{-2})$, the spectral cutoff reduces the coefficient search to dimension $O(\varepsilon^{-1})$. A finite net of this coefficient space therefore produces only $(C/\varepsilon)^{O(1/\varepsilon)}$ explicit pure targets, and we run the proper pure-target optimizer once for each of them.

We obtain the following general reduction.

\begin{theorem}[Improper-to-proper reduction]
\label{thm:general-improper-to-proper-results}
Let $(\mathcal C,(\mathcal V_R)_{R\geq1})$ be a properization pair. For every density operator $\rho$ and every $\varepsilon,\delta\in(0,1/4)$, there is a randomized algorithm using $O(\varepsilon^{-2})$ improper-learning calls and at most $(C/\varepsilon)^{O(1/\varepsilon)}$ proper pure-target calls which outputs a normalized state $\widehat\phi\in\mathcal C$ such that, with probability at least $1-\delta$,
\begin{equation}
\langle\widehat\phi|\rho|\widehat\phi\rangle\geq\max_{\phi\in\mathcal C}\langle\phi|\rho|\phi\rangle-\varepsilon.
\label{eq:general-proper-learning-guarantee-results}
\end{equation}
\end{theorem}

Algorithm~\ref{alg:proper-agnostic-learning} summarizes the reduction. Here \textsc{RelevantSubspace} denotes the relevant-subspace construction of Appendix~\ref{sec:improper-to-proper}, with the fixed comparator parameter $R'$, and \textsc{PureTargetOpt} denotes the squared-overlap optimizer in Definition~\ref{def:properization-data}.

\begin{algorithm}[htbp]
\caption{Proper agnostic learning from an improper learner}
\label{alg:proper-agnostic-learning}
\small
\begin{algorithmic}[1]
\Require Copies of $\rho$; properization data for $\mathcal C$; $\varepsilon,\delta\in(0,1/4)$.
\Ensure A normalized $\widehat\phi\in\mathcal C$ whose $\rho$-score
is within $\varepsilon$ of optimal, with probability at least $1-\delta$.

\Statex \textit{On any declared failure, return the fixed state $\phi_0\in\mathcal C$.}

\State $\theta\gets(\varepsilon/16)^2$,$\xi\gets\varepsilon/32$,$\tau\gets\varepsilon/8$, $\zeta\gets\varepsilon/4$.
\State $(s_1,\ldots,s_m)\gets
\Call{RelevantSubspace}{\rho,\theta,\delta/2}$; define $Se_a:=s_a$.
\State Estimate a Hermitian $G_0$ with $\|G_0-S^\dagger\rho S\|_\infty\leq\xi/2$, using failure probability $\delta/2$.
\State $\widehat G\gets(G_0)_+$.
\State Retain the eigenpairs $(\lambda_j,g_j)_{j=1}^{\ell}$ of $\widehat G$ with $\lambda_j>\tau$; set $u_j:=Sg_j$.
\If{$\ell=0$}
    \State \Return $\phi_0$.
\EndIf

\State Construct a Euclidean $\zeta/16$-net $\mathcal N$
       of the unit sphere in $\mathbb C^\ell$.

\ForAll{$c\in\mathcal N$}
    \State $v_c\gets \sum_{j=1}^{\ell}\sqrt{\lambda_j}\,c_j u_j$; $w_c\gets v_c/\|v_c\|_2$.
    \State $\phi_c\gets\Call{PureTargetOpt}{w_c,\zeta/4}$.
    \State $F_c\gets\sum_{j=1}^{\ell}\lambda_j |\langle u_j|\phi_c\rangle|^2$.
\EndFor

\State Choose $c_\star\in\operatorname*{arg\,max}_{c\in\mathcal N}F_c$.
\State \Return $\widehat\phi:=\phi_{c_\star}$.
\end{algorithmic}
\end{algorithm}

For the MPS application, the only remaining task is therefore to find a nearly optimal bond-$D$ MPS given an explicit target vector.  The next two subsections solve this problem.

\subsection{Comparator-dual compression}
\label{sec:results-compression}

The targets $v_c$ produced above are explicit linear combinations of the states found by the relevant-subspace procedure.  Their MPS bond dimensions may therefore be much larger than $D$.  A direct optimization algorithm with exponential dependence on the target bond dimension would therefore not give the desired complexity.

The target need not, however, be approximated well in Euclidean norm.  It is enough to preserve all overlaps with bond-$D$ MPSs.  This leads to the comparator-dual norm
\begin{equation}
\|x\|_{\mathrm{MPS}(D)}:=\sup_{\psi\in\mps_D}|\langle\psi|x\rangle|.
\label{eq:mps-comparator-dual-norm-results}
\end{equation}
Two targets that are close in this norm define nearly the same optimization problem over $\mps_D$, even if they are far apart in Euclidean norm.

\begin{theorem}[Comparator-dual compression for MPSs]
\label{thm:mps-comparator-dual-compression-results}
For every vector $u\in\mathcal H$, every $D\geq1$, and every $K\geq1$, one can construct an MPS $\widetilde u_K$ of bond dimension at most $K$ satisfying
\begin{equation}
\|u-\widetilde u_K\|_{\mathrm{MPS}(D)}\leq\|u\|_2\sqrt{\frac{D}{K+1}}.
\label{eq:mps-compression-guarantee-results}
\end{equation}
Moreover, $\|\widetilde u_K\|_2\leq\|u\|_2$.  If $u$ is an explicit MPS of bond dimension $B$, the construction uses $\poly(n,d,B,K)$ arithmetic operations.
\end{theorem}

The absence of any factor depending on $n$ is essential. For targets satisfying $\|u\|_2\leq1$, taking $K=O(D/\eta^2)$ gives comparator-dual error at most $\eta$, regardless of the length of the chain.

We describe the mechanism behind the bound.  The construction sweeps from left to right.  After sites $1,\ldots,i-1$ have been processed, their retained information is stored in an effective boundary space $\mathcal B_{i-1}$ of dimension at most $K$, and the unprocessed part is a vector
\begin{equation*}
T_i\in\mathcal B_{i-1}\otimes\mathcal H_{i:n}.
\end{equation*}
Regard $T_i$ as a matrix from $\mathcal H_{i+1:n}$ to $\mathcal B_{i-1}\otimes\mathcal H_i$, with singular values $\sigma_1\geq\cdots\geq\sigma_p>0.$
If $p\leq K$, the step is exact.  Otherwise, set $\Delta_i:=\sigma_{K+1}^2$
and propagate the modified singular weights
\begin{equation}
b_j:=\sqrt{\sigma_j^2-\Delta_i},\qquad 1\leq j\leq K,
\label{eq:results-propagated-singular-values}
\end{equation}
to the next remainder $T_{i+1}$.

This differs from ordinary rank-$K$ truncation.  Standard truncation would retain the leading singular values unchanged and discard only the tail.  Here we also remove the same amount $\Delta_i$ from the square of each retained singular weight.  This deliberately loses more Euclidean norm.  The advantage is that the one-step decomposition can be arranged in the form
\begin{equation}
T_i=C_i(T_{i+1}\oplus L_i),
\label{eq:results-one-step-compression}
\end{equation}
where $C_i$ is a contraction, every singular value of the loss tensor $L_i$ is at most $\sqrt{\Delta_i}$, and
\begin{equation}
(K+1)\Delta_i \leq\|T_i\|_2^2-\|T_{i+1}\|_2^2.
\label{eq:results-one-step-energy-drop}
\end{equation}
The first property controls the loss in the comparator-dual norm, and the second gives a telescoping bound on the total loss over all cuts.

Indeed, across the boundary cut a bond-$D$ comparator has Schmidt rank at most $D$.  If all singular values of $L_i$ are at most $\sqrt{\Delta_i}$, von Neumann's trace inequality gives
\begin{equation}
\|L_i\|_{\mathrm{MPS}(D)}^2\leq D\Delta_i.
\label{eq:results-flat-loss-bound}
\end{equation}
The comparator-dual norm is nonincreasing under the contraction $C_i$, and losses in orthogonal boundary sectors obey a squared-sum bound.  If $E_i$ denotes the accumulated error after cut $i$, this implies
\begin{equation}
\|E_i\|_{\mathrm{MPS}(D)}^2 \leq\|E_{i+1}\|_{\mathrm{MPS}(D)}^2+D\Delta_i.
\label{eq:results-compression-error-recursion}
\end{equation}
Consequently,
\begin{equation}
\|u-\widetilde u_K\|_{\mathrm{MPS}(D)}^2\leq D\sum_{i=1}^{n-1}\Delta_i.
\label{eq:results-compression-error-sum}
\end{equation}
On the other hand, summing \eqref{eq:results-one-step-energy-drop} telescopes:
\begin{equation}
(K+1)\sum_{i=1}^{n-1}\Delta_i\leq\|T_1\|_2^2-\|T_n\|_2^2\leq\|u\|_2^2.
\label{eq:results-compression-energy-budget}
\end{equation}
Combining \eqref{eq:results-compression-error-sum} and \eqref{eq:results-compression-energy-budget} proves Theorem~\ref{thm:mps-comparator-dual-compression-results}.

The extra shrinkage decreases the squared Euclidean norm by at least $(K+1)\Delta_i$ at each step, whereas the corresponding squared comparator-dual error is at most $D\Delta_i$.  Ordinary truncation gives the same pointwise bound on the largest discarded singular value, but it does not force the additional $K\Delta_i$ decrease needed to obtain the factor $1/(K+1)$ after summing over all cuts. 
The same construction also gives a relative-error approximation in Euclidean norm. If
\begin{equation*}
e_D(u):=\min_{\substack{a\in\mathbb C\\\psi\in\mps_D}}\|u-a\psi\|_2,
\end{equation*}
then, for $K\geq D$,
\begin{equation*}
\|u-\widetilde u_K\|_2^2\leq\frac{K+1}{K+1-D}\,e_D(u)^2.
\end{equation*}
Thus bond $K=O(D/\alpha)$ suffices for squared error within
a factor $1+\alpha$ of the best bond-$D$ approximation,
independently of the chain length. This uses the compression
sweep alone, without the optimization procedure below. Corollary~\ref{cor:mps-relative-error-compression} gives the
proof.

\subsection{Proper optimization against an explicit target}
\label{sec:results-explicit-target}

After comparator-dual compression, it remains to optimize against a target of bond dimension $K=O(D/\eta^2)$. As a candidate MPS is built from left to right, the effect of the processed prefix on the final overlap is completely captured by a bounded-dimensional cross environment.

Let $v$ be a target MPS of bond at most $K$ and norm at most one, and let $\psi\in\mps_D$ be a candidate.  Put both states in sequential isometric form.  Write $V_i^s$ and $A_i^s$ for their local tensors.  The contraction of the two prefixes through site $i$ is summarized by the cross environment $X_i$, initialized by $X_0=1$ and updated according to
\begin{equation}
X_i =\sum_{s=1}^dV_i^sX_{i-1}(A_i^s)^\dagger.
\label{eq:mps-cross-environment-update-results}
\end{equation}
At the final cut, $X_n$ is the scalar overlap $\langle\psi|v\rangle$.  Thus, once $X_i$ is known, the exponentially large prefix itself contains no further information relevant to the objective. For fixed local tensors $V$ and $A$, define
\begin{equation*}
\Phi^A(X):=\sum_{s=1}^dV^sX(A^s)^\dagger.
\end{equation*}
The isometry conditions imply the stability estimates
\begin{equation}
\|\Phi^A(X)-\Phi^A(Y)\|_1\leq\|X-Y\|_1
\label{eq:mps-environment-contraction-results}
\end{equation}
and
\begin{equation}
\|\Phi^A(X)-\Phi^{A'}(X)\|_1\leq\|A-A'\|_\infty\|X\|_1.
\label{eq:mps-environment-local-stability-results}
\end{equation}
This means that the errors in the stored environment are not amplified as the sweep proceeds.

The algorithm discretizes both the candidate isometries and the cross environments $X_i$ defined above. At cut $i$, each retained prefix is extended by every tensor in a sufficiently fine local net. The resulting cross environment is then assigned to a nearby point in a finite net of possible environments, and one exact prefix is retained for each occupied cell. If two prefixes are assigned to the same cell, their cross environments are close, so any common continuation produces final overlaps that differ by at most the propagated environment error.

Let $h$ be the operator-norm radius of the local-tensor net and $q$ the trace-norm radius of the environment net. These parameters control, respectively, the error from discretizing the next local tensor and from merging nearby cross environments. Following the discretized path of an optimal candidate through the dynamic program gives the recursion
\begin{equation*}
e_i\leq e_{i-1}+h+2q,
\end{equation*}
where $e_i$ is the trace-norm error of the surviving environment after cut $i$.  Taking $h:=\frac{\eta}{4n}$ and $q:=\frac{\eta}{8n}$ keeps the final overlap error below $\eta$.  The number of retained prefixes is bounded by the size of the environment net, rather than by the dimension of the many-body prefix.

For a target of bond at most $K$, this yields a deterministic algorithm which outputs a normalized state $\widehat\psi\in\mps_D$ satisfying
\begin{equation}
|\langle v|\widehat\psi\rangle|\geq\max_{\psi\in\mps_D}|\langle v|\psi\rangle| -\eta
\label{eq:results-bounded-target-guarantee}
\end{equation}
in time
\begin{equation}
n\left(\frac{Cn\sqrt D}{\eta}\right)^{O(KD+dD^2)}.
\label{eq:results-bounded-target-runtime}
\end{equation}
Every state considered by the search is assembled from isometric local tensors of bond dimension at most $D$.  The output is therefore normalized and proper by construction since there is no final rounding step.

Combining this dynamic program with Theorem~\ref{thm:mps-comparator-dual-compression-results} removes the target bond dimension from the exponential part of the runtime.

\begin{theorem}[Proper optimization against arbitrary explicit targets]
\label{thm:mps-arbitrary-target-optimizer-results}
Let $u$ be a normalized explicit MPS of arbitrary bond dimension $B$.  For every $D\geq1$ and every $\eta\in(0,1)$, there is a deterministic algorithm which outputs a normalized state $\widehat\psi\in\mps_D$ such that
\begin{equation}
|\langle u|\widehat\psi\rangle|^2\geq\max_{\psi\in\mps_D}|\langle u|\psi\rangle|^2-\eta.
\label{eq:mps-arbitrary-target-guarantee-results}
\end{equation}
Its runtime is bounded by
\begin{equation}
n\left(\frac{Cn\sqrt D}{\eta}\right)^{O(D^2/\eta^2+dD^2)}
\poly(d,B,D,1/\eta)
\label{eq:mps-arbitrary-target-runtime-results}
\end{equation}
for a universal constant $C$.
\end{theorem}

To see the reduction, choose $K+1\geq64D/\eta^2$.  The compression theorem produces a bond-$K$ target $\widetilde u_K$ satisfying
\begin{equation*}
\sup_{\psi\in\mps_D}|\langle\psi|u-\widetilde u_K\rangle|\leq\frac{\eta}{8}.
\end{equation*}
Running the bounded-target optimizer on $\widetilde u_K$ and transferring the guarantee back to $u$ by the triangle inequality gives an amplitude error at most $\eta/2$.  A difference-of-squares estimate then yields \eqref{eq:mps-arbitrary-target-guarantee-results}.
\subsection{Proper agnostic learning of MPSs}
\label{sec:results-proper-mps}

We now verify the hypotheses of the general reduction for MPSs.  Take $\mathcal C:=\mps_D$ and, for $R\geq1$, define
\begin{equation*}
\mathcal V_R:=\left\{u\in\mathcal H:\operatorname{rank}_{\mathcal H_{1:i}\,|\,\mathcal H_{i+1:n}}(u)\leq R\text{ for every }1\leq i\leq n-1\right\}.
\end{equation*}
Thus $\mathcal A_R=\mps_R$ and $\mathcal C=\mathcal A_D$.  Finite linear combinations satisfy
\begin{equation*}
\sum_{a=1}^qc_au_a\in\mathcal V_{\sum_{a=1}^qR_a} \qquad\text{whenever }u_a\in\mathcal V_{R_a}.
\end{equation*}
For MPSs, we use the improper learner of Bakshi et al.~\cite[Theorem~B.2]{BakshiEtAl2025}. The remaining steps of the reduction can be carried out with polynomial overhead for explicit MPSs of polynomial bond dimension, and Theorem~\ref{thm:mps-arbitrary-target-optimizer-results} handles the proper optimization against explicit pure targets.

The general reduction returns a state in $\mps_D$ with score at least $\opt_D(\rho)-\varepsilon$. Since every improper learning call uses the same comparator bond $D$, the bond dimensions of the learner outputs remain polynomially controlled. Its output bond is therefore $\poly(n,d,D,1/\varepsilon)$. There are $O(\varepsilon^{-2})$ accepted learner outputs. Expressing each orthonormal direction, and each subsequent pure target, directly as a linear combination of these original outputs keeps every intermediate MPS bond jointly polynomial in $n,d,D,1/\varepsilon$. The explicit operations and support-matrix estimation consequently require only polynomially many copies. After the quantum data have been collected, spectral cutoff and proper pure-target optimization are entirely classical. Combining their costs proves
Theorem~\ref{thm:proper-agnostic-mps-results}.

The same quantum data also allow estimation of $\opt_D(\rho)$. More generally, a core constructed for $\mps_{D_{\max}}$ can be reused to fit every smaller bond dimension and estimate its optimal score, without further quantum measurements. This gives a direct way to compare the approximation quality available at different bond dimensions; see Corollary~\ref{cor:mps-optimal-scores}.

The same conclusions hold, under vectorization, for Hilbert--Schmidt-normalized MPOs. Vectorization identifies an MPO of bond dimension $D$ with an MPS of the same bond dimension and local dimension $d^2$. Comparator-dual compression, proper pure-target optimization, and proper agnostic learning therefore transfer directly under the substitution $d\mapsto d^2$. These statements concern Hilbert--Schmidt-normalized operators and do not impose positivity or trace normalization; the precise results are given in Appendix~\ref{sec:matrix-product-operators}.

\subsection{Extension to tree tensor networks}
\label{sec:results-ttn}

Let $T=(V,E)$ be a rooted tree with $|V|=n$ and maximum degree $\Delta_T$, and let $\mathcal H_T:=\bigotimes_{v\in V}\mathcal H_v$ where $\mathcal H_v\cong\mathbb C^d$. For an edge $e\in E$, deleting $e$ induces a bipartition of the vertices. Write $\sr_e(u)$ for the Schmidt rank of $u$ across this bipartition.  The normalized TTNs on $T$ of bond dimension at most $D$ are precisely the states satisfying $\sr_e(u)\leq D$ for every $e\in E$. We denote this class by $\ttn_{T,D}$.

For the learning subroutine, we embed bounded-degree TTNs into MPSs of polynomial bond dimension by choosing a suitable ordering of the vertices. For each vertex $v$, choose a child of maximum subtree size as the ``heavy child'', visit $v$ first, then visit all non-heavy child subtrees, and visit the heavy child subtree last. Denote the resulting vertex order by $\pi$. By construction, $\pi$ is a preorder traversal of the rooted tree, with the heavy child visited last. The preorder property implies that every rooted subtree is a contiguous interval in $\pi$. The heavy-child-last rule ensures that every prefix cut crosses at most
\begin{equation}
w_T:=\Delta_T\left(1+\lceil\log_2n\rceil\right)
\label{eq:results-tree-width}
\end{equation}
tree edges.

To see the second claim, fix a prefix of $\pi$. Only vertices on the unfinished root-to-current path can have edges from the prefix to unvisited vertices. A proper ancestor on this path can contribute such an edge only if the traversal descended from it through a non-heavy child. If $u$ is a non-heavy child of $v$, then the heavy child of $v$ has subtree size at least $|T_u|$, and therefore
\begin{equation*}
2|T_u|\leq |T_v|-1,
\end{equation*}
so that
\begin{equation*}
|T_u|<\frac{|T_v|}{2}.
\end{equation*}
Thus the subtree size decreases by at least a factor of two each time a proper ancestor can contribute a crossing edge. There are therefore at most $\lceil\log_2n\rceil$ such ancestors, and each vertex on this path contributes at most $\Delta_T$ crossing edges, giving the bound above.

If a prefix cuts $k$ tree edges, then the corresponding bipartition exposes at most $k$ virtual indices, each of dimension at most $D$, so a bond-$D$ TTN has Schmidt rank at most $D^k$ across that cut. By the MPS cut-rank characterization,
\begin{equation}
\ttn_{T,D}\subseteq\mps_{D^{w_T}}^\pi.
\label{eq:results-ttn-mps-embedding}
\end{equation}
For fixed $D$ and $\Delta_T$,
\begin{equation*}
D^{w_T}=n^{O(\Delta_T\log D)},
\end{equation*}
so the resulting MPS bond dimension is polynomial in $n$.

We will also use the reverse conversion, from MPSs in the ordering $\pi$ to TTNs on $T$. Each tree edge separates a rooted subtree from its complement, and every rooted subtree is a contiguous interval in $\pi$. Separating a contiguous interval from an open-boundary MPS cuts at most two virtual bonds, so an MPS of bond dimension $B$ has Schmidt rank at most $B^2$ across every tree edge. By the edge-rank characterization of TTNs, it therefore admits a TTN representation on $T$ of bond dimension at most $B^2$.

Comparator-dual compression extends to the tree without changing the approximation guarantee.  Define
\begin{equation*}
\|x\|_{T,D}:=\sup_{\psi\in\ttn_{T,D}}|\langle\psi|x\rangle|.
\end{equation*}
We process the tree from the leaves toward the root. Once a child subtree has been processed, we compress all information from that subtree into a single boundary space of dimension at most $K$, which becomes the effective virtual index connecting the subtree to its parent.  At the next vertex, group the physical space at that vertex with the boundary spaces of its processed children and apply the same one-step SVD shrinkage used in the MPS construction.  The matrix identities behind \eqref{eq:results-one-step-compression} and \eqref{eq:results-one-step-energy-drop} do not depend on the underlying graph.  The tree comparator norm obeys the same three properties needed in the chain proof: it is nonincreasing under contraction of a pendant part, errors in orthogonal boundary sectors obey a squared-sum bound, and a loss tensor with singular values at most $\sqrt\Delta$ has squared comparator norm at most $D\Delta$.  Consequently, for every $u\in\mathcal H_T$ one can construct a bond-$K$ TTN $\widetilde u_K$ satisfying
\begin{equation}
\|u-\widetilde u_K\|_{T,D}\leq\|u\|_2\sqrt{\frac{D}{K+1}},\qquad\|\widetilde u_K\|_2\leq\|u\|_2.
\label{eq:results-ttn-compression-guarantee}
\end{equation}
If $u$ is an explicit TTN of bond dimension $B$, the construction uses $\poly(n,d,B^{\Delta_T},K^{\Delta_T})$ arithmetic operations.

The proper pure-target optimizer also has a tree analogue.  Put the target and candidate TTNs in rooted isometric form.  Contracting the two subtrees rooted at a non-root vertex $x$, while leaving the parent virtual legs open, defines a cross environment $X_x$.  These environments satisfy
\begin{equation}
X_x=V_x^\dagger\left( I_{\mathcal H_x}\otimes\bigotimes_{w\in\operatorname{ch}(x)}X_w\right)A_x.
\label{eq:ttn-cross-environment-update-results}
\end{equation}
At the root, the same contraction is the scalar overlap between the target and the candidate.

The update is stable in operator norm. In particular, if the cross environments $X_w$ associated with the child subtrees are perturbed to $Y_w$, then
\begin{equation*}
\left\|\bigotimes_wX_w-\bigotimes_wY_w\right\|_\infty\leq \sum_w\|X_w-Y_w\|_\infty,
\end{equation*}
and perturbing the local candidate tensor contributes at most its operator-norm error.  A bottom-up dynamic program can therefore discretize the local isometries and subtree environments, retain one exact subtree in each occupied environment cell, and propagate the search toward the root.  For a target of bond dimension at most $K$ and norm at most one, the resulting runtime is bounded by
\begin{equation*}
n\left(\frac{Cn\sqrt{KD}}{\eta}\right)^{O(\Delta_TKD+dD^{\Delta_T+1})}\poly(K,D,d).
\end{equation*}
Thus the optimizer is polynomial in $n$ for fixed $d,D,\Delta_T,K,\eta$.  Combining it with \eqref{eq:results-ttn-compression-guarantee}, with $K=O(D/\eta^2)$, gives proper optimization against an arbitrary explicit TTN target in polynomial time for fixed $d,D,\Delta_T,\eta$.

Finally, apply the general reduction with $\mathcal C=\ttn_{T,D}$ and with the ambient classes given by MPSs in the ordering $\pi$.  The inclusion \eqref{eq:results-ttn-mps-embedding} verifies the embedding condition, the MPS improper learner supplies the ambient improper learner, and the MPS-to-TTN conversion followed by the proper TTN optimizer supplies the pure-target optimizer.  This yields the following result.

\begin{theorem}[Proper agnostic learning of bounded-degree TTNs]
\label{thm:proper-agnostic-ttn-results}
For every rooted tree $T$ with $n$ vertices and maximum degree at most $\Delta_T$, every $d,D\in\mathbb N$, and every $\varepsilon,\delta\in(0,1/4)$, there is a randomized quantum algorithm which, given copies of an arbitrary density operator $\rho$ on $\mathcal H_T$, outputs a classical description of a normalized state $\widehat\psi\in\ttn_{T,D}$ such that, with probability at least $1-\delta$,
\begin{equation}
\langle\widehat\psi|\rho|\widehat\psi\rangle\geq\max_{\psi\in\ttn_{T,D}}\langle\psi|\rho|\psi\rangle-\varepsilon.
\label{eq:ttn-proper-agnostic-guarantee-results}
\end{equation}
With $w_T:=\Delta_T(1+\lceil\log_2n\rceil)$, the copy complexity is
bounded by $\poly\!\left(n,d,D^{w_T},1/\varepsilon,\log(1/\delta)\right).$
For fixed $d,D,\Delta_T,\varepsilon$, the runtime is
bounded by $n^{F(d,D,\Delta_T,1/\varepsilon)}\poly(\log(1/\delta))$
for an explicit computable function $F$ independent of $n$.
\end{theorem}

\subsection{Learning with branch structure}
\label{sec:results-branching}

We also learn descriptions as coherent superpositions of $r$ normalized bond-$D$ MPS branches. For a branch tuple $\boldsymbol\psi=(\psi_1,\ldots,\psi_r)$, define its total $k$-local interference by
\begin{equation}
\mathcal I_k(\boldsymbol\psi):=\sum_{a\neq b}\sum_{\substack{A\subseteq[n]\\|A|\leq k}}\left\|\tr_{\overline A}(|\psi_b\rangle\langle\psi_a|)\right\|_2^2,\qquad \overline A:=[n]\setminus A.
\label{eq:branching-total-interference-results}
\end{equation}
Here $\|\cdot\|_2$ is the Hilbert--Schmidt norm, and the sum includes $A=\varnothing$. Let $\mathcal L_{r,D,k}(\mu)$ be the class of normalized states $|\Phi\rangle=\sum_{a=1}^r c_a|\psi_a\rangle$ admitting such a branch tuple with $\mathcal I_k(\boldsymbol\psi)\leq\mu$. The normalization is imposed on the combined state so the coefficients need not have unit Euclidean norm.

\begin{theorem}[Agnostic learning with a relaxed interference bound]
\label{thm:agnostic-branching-results}
Let $n,d,r,D\in\mathbb N$, $k\in\{0,\ldots,n\}$, $0\leq\mu<1/64$, and $0<\sigma\leq1/64-\mu$, and assume $\mathcal L_{r,D,k}(\mu)\neq\varnothing$. Given copies of an arbitrary density operator $\rho$ on $(\mathbb C^d)^{\otimes n}$ and $\varepsilon,\delta\in(0,1/4)$, a randomized quantum algorithm returns $r$ normalized states $\widehat\psi_a\in\mps_D$ and coefficients $\widehat c$ such that
\begin{equation}
|\widehat\Phi\rangle:=\sum_{a=1}^r\widehat c_a|\widehat\psi_a\rangle,\qquad \|\widehat\Phi\|_2=1,\qquad \mathcal I_k(\widehat{\boldsymbol\psi})\leq\mu+\sigma,
\label{eq:branching-output-results}
\end{equation}
and, with probability at least $1-\delta$,
\begin{equation}
\langle\widehat\Phi|\rho|\widehat\Phi\rangle\geq\max_{\Phi\in\mathcal L_{r,D,k}(\mu)}\langle\Phi|\rho|\Phi\rangle-\varepsilon.
\label{eq:branching-learning-guarantee-results}
\end{equation}
The copy complexity is $\poly(n,d,rD,1/\varepsilon,\log(1/\delta))$. For fixed $d,r,D,k,\varepsilon,\sigma$, the runtime is polynomial in $n$ and $\log(1/\delta)$.
\end{theorem}

Every such superposition has MPS bond at most $rD$, so the relevant-subspace reduction and comparator-dual compression apply with comparator bond $rD$. To obtain the decomposition, the classical optimizer searches jointly over the branch tensors and optimizes their coefficients using the Gram matrix. For each branch pair, the local interference terms can be evaluated recursively as the branches are contracted. Grouping subsets by their cardinality gives $k+1$ bounded matrix environments that can be tracked alongside the overlap environments. Appendix~\ref{sec:learning-with-branching-structure} proves the guarantee. The algorithm finds a nearly optimal decomposition; recovery of a particular original decomposition requires an additional identifiability assumption.

\section{Discussion}

Our results establish proper agnostic learning of MPSs and bounded-degree TTNs by using improper learning to construct a reduced objective that approximately preserves all scores in the comparison class. A classical optimization step then finds a nearly optimal MPS or TTN of the chosen bond dimension. The guarantee is relative to the best attainable score in the model class, which need not be large; moreover, near-optimality within the class should not in general be interpreted as reconstruction of the unknown state. For MPSs, the ability to estimate the optimum at several bond dimensions from the same quantum data also allows the approximation quality of the model class to be assessed.

The main computational limitation is the classical search. For MPSs, the copy complexity is jointly polynomial in $n,d,D,1/\varepsilon$, and $\log(1/\delta)$, whereas the exponent of the runtime bound depends on $d,D$, and $1/\varepsilon$. Already for product states, corresponding to $D=1$, proper agnostic learning with runtime polynomial in both $n$ and $1/\varepsilon$ would imply $\textsf{NP}\subseteq\textsf{BQP}$~\cite{BakshiEtAl2025}. Thus one cannot in general expect polynomial dependence on all parameters. For TTNs, the copy bound also incurs the ambient bond $D^{w_T}$ introduced by the MPS embedding. Improving these parameter dependences remains an important direction. It would also be useful to replace the adaptive subspace projections and postselection in our reduction with local measurements on individual copies while retaining the proper agnostic guarantee.

Our compression algorithm also applies to explicitly given MPSs, without quantum data. It produces an MPS of bond $O(D/\alpha)$ whose squared Euclidean approximation error is at most $1+\alpha$ times that of the best bond-$D$ approximation, while also controlling overlap errors uniformly over all normalized bond-$D$ MPSs with a bound independent of the chain length. This comparator-dual viewpoint may be useful more broadly in structured approximation problems where only a restricted class of test states needs to be preserved. Extending the compression and proper optimization procedures to tensor networks with cycles is a natural further direction. For periodic-boundary MPSs, our corollary remains improper, returning an open-boundary MPS of bond $D^2$ when competing with bond-$D$ periodic MPSs; proper agnostic learning in this setting is not established here.

We have also considered coherent superpositions of several low-bond MPS branches subject to a local noninterference constraint. This shows that the same framework can learn structured decompositions rather than only a single tensor-network state. Our guarantee identifies a nearly optimal decomposition satisfying the prescribed branch complexity and an approximately preserved interference constraint, but recovery of a particular underlying decomposition requires an additional identifiability assumption. Developing broader and more intrinsic notions of learnable branch structure, and understanding their relation to other notions of wavefunction branching~\cite{Riedel2017,riedel2025wavefunction,taylor2025non, taylor2025wavefunction, pilatowsky2026emergent}, would be an interesting direction.

\subsection*{Acknowledgements}

We thank Luke Coffman and Ishaan Kannan for valuable discussions.  We acknowledge the use of ChatGPT 5.5 and 6 as research aids to help identify errors, work through intermediate steps of calculations and derivations, and suggest edits to the text. All AI-assisted calculations and derivations were independently checked by the authors, and all results and arguments were written up and verified by the authors. This work is supported by the U.S. Department of Energy, Office of Science, under Award Number DE-SC0021013 titled ‘Emergent Phenomena in Quantum Dynamics: From Chaos to Spacetime’. CCVP received support from Caltech SURF Fellowship during part of this work.

$$$$

\noindent {\LARGE \bfseries Appendices}
\appendix

\section{Related work}
\label{sec:related-work}

Quantum state tomography can exploit low-rank or tensor-network structure to reduce reconstruction costs. This includes compressed-sensing approaches to low-rank tomography~\cite{GrossEtAl2010}, reconstruction methods for MPSs and MPOs~\cite{CramerEtAl2010,LandonCardinalLiuPoulin2010,BaumgratzEtAl2013}, and, more recently, sketch tomography based on classical shadows and tensor-train reconstruction~\cite{TangEtAl2025Sketch}. These results concern reconstruction under structural assumptions on the unknown state. Our objective is different: we seek a nearly optimal state within a chosen model class even when the unknown state is poorly approximated by that class.

Agnostic learning has been studied for several structured classes, including stabilizer product states and discrete product-state families~\cite{GrewalEtAl2026,ChenEtAl2025}. Most directly relevant is Ref.~\cite{BakshiEtAl2025}, which proves proper agnostic learning of product states and improper agnostic learning of MPSs. Subsequent work improves the system-size dependence of improper MPS learning~\cite{LinChiaHung2025} and extends improper agnostic learning to trees and more general known graphs~\cite{CaroMcHughStrelchuk2026}. Our results differ in requiring the output itself to satisfy the chosen MPS or TTN bond-dimension constraint.

The reduction from improper to proper learning is related to quantum agnostic boosting and to learning through structured decompositions~\cite{ArunachalamEtAl2026Boosting,ArunachalamDutt2026Structure,ArunachalamDutt2026Extent}. In particular, residual weak learning has been used to construct pure-state decompositions whose overlaps are uniformly accurate over structured comparator classes~\cite{ArunachalamDutt2026Extent}. Here we instead consider arbitrary mixed inputs and construct a low-rank operator that approximately preserves the score of every comparator. Proper learning then reduces to a finite collection of pure-state approximation problems.

Our proper optimizer is also related to earlier algorithms for optimization over fixed-bond-dimension MPSs, which use discretization and dynamic programming to minimize one-dimensional local Hamiltonian energies~\cite{SchuchCirac2010,AharonovAradIrani2010}. We use a similar dynamic-programming structure, with cross environments recording the information needed to optimize overlap with an explicit target. The additional comparator-dual compression step bounds the target bond dimension independently of the chain length, preventing the search from depending exponentially on the bond dimension of the original target.

Finally, the compression procedure is related to classical tensor approximation and matrix sketching. Tensor-train SVD (TT-SVD), hierarchical SVD, and randomized tensor-rounding methods provide standard low-rank approximation procedures~\cite{Oseledets2011,Grasedyck2010,AlDaasEtAl2023}, while relative-error tensor approximation has been studied using rank enlargement and parameterized algorithms~\cite{SongWoodruffZhong2019,MahankaliWoodruffZhang2024}. Our singular-value shrinkage is reminiscent of Frequent Directions, which subtracts a common amount from squared singular values and charges the accumulated error to the resulting decrease in Frobenius norm~\cite{Liberty2013,GhashamiEtAl2016}; related ideas have also been applied to Tucker and tensor-train decompositions~\cite{CheWeiYan2026}. Here the shrinkage is combined with contractive retained maps to control overlap errors uniformly over all bond-$D$ comparators, with no dependence on the number of sites. The same construction also gives a bond-$O(D/\alpha)$ approximation whose squared Euclidean error is at most $(1+\alpha)$ times the best bond-$D$ error.

\section{Improper to Proper Learning}
\label{sec:improper-to-proper}

We now formalize the improper-to-proper reduction described in the main text. Let $\mathcal H$ be a finite-dimensional Hilbert space and let 
\begin{equation} 
\varnothing\neq\mathcal C\subseteq\{\phi\in\mathcal H:\|\phi\|_2=1\} 
\end{equation}
be a compact set of normalized comparator states. For every positive semidefinite operator $X\succeq0$, define 
\begin{equation} 
\opt_{\mathcal C}(X):=\max_{\phi\in\mathcal C}\bra{\phi}X\ket{\phi}. 
\end{equation}
Given copies of a density operator $\rho$, an improper learner may output any efficiently described normalized state $b\in\mathcal H$ satisfying 
\begin{equation} 
\bra{b}\rho\ket{b}\geq\opt_{\mathcal C}(\rho)-\varepsilon,
\end{equation}
whereas a proper learner must output a state $\widehat\phi\in\mathcal C$ satisfying the same guarantee. Let $(\mathcal V_R)_{R\geq1}$ be a nested family of explicit vector classes,
\begin{equation}
\mathcal V_1\subseteq\mathcal V_2\subseteq\cdots\subseteq\mathcal H,
\end{equation}
whose elements need not be normalized. Write 
\begin{equation}
\mathcal A_R:=\{u\in\mathcal V_R:\|u\|_2=1\}, \end{equation}
and
\begin{equation}
\opt_{\mathcal A_R}(X):=\sup_{u\in\mathcal A_R}\bra{u}X\ket{u}. 
\end{equation}

\begin{definition}\label{def:properization-data}
Call $(\mathcal{C},(\mathcal{V}_R)_{R\ge 1})$ a properization pair if the following hold.
\begin{itemize}
\item There is a fixed parameter $R'$ such that $\mathcal C\subseteq\mathcal A_{R'}$, together with a fixed explicit state $\phi_0\in\mathcal C$.

\item There is a computable monotone function $\Lambda$ such that, for every $u_a\in\mathcal V_{R_a}$ and $c_a\in\mathbb C$,
\begin{equation}
\sum_{a=1}^q c_a u_a\in\mathcal V_{\Lambda(R_1,\ldots,R_q)}.
\end{equation}
\item There is an ambient improper learner. Given copies of an arbitrary density operator $\omega$, a comparator parameter $R$, an error $\gamma\in(0,1)$, and a failure probability $\beta\in(0,1)$, it outputs an explicit normalized state $b\in\mathcal A_{\widetilde R}$, where $\widetilde R$ is bounded by a computable function of the input parameters, such that with probability at least $1-\beta$,
    \begin{equation}
        \langle b|\omega|b\rangle\geq\opt_{\mathcal A_R}(\omega)-\gamma.
        \label{eq:general-improper-learner-guarantee}
    \end{equation}
    \item There is a proper pure target optimizer. Given an explicit normalized vector
    $u\in\mathcal A_B$ and an error
    $\eta\in(0,1)$, it outputs a normalized state
    $\widehat\phi\in\mathcal C$ such that
    \begin{equation}
    |\langle u|\widehat\phi\rangle|^2\geq \max_{\phi\in\mathcal C}|\langle u|\phi\rangle|^2-\eta.
        \label{eq:general-pure-target-optimizer}
    \end{equation}
    \item 
    For the explicit vectors produced above, one can efficiently compute
    norms and overlaps, normalize nonzero vectors, prepare the corresponding states, measure the projector onto the span of finitely many explicitly known orthonormal states, and estimate matrix elements $\langle u|\rho|v\rangle$ to our desired additive accuracy.
\end{itemize}

\end{definition}

\begin{theorem}\label{thm:general-finite-relevant-subspace}
Let $(\mathcal C,(\mathcal V_R)_{R\geq1})$ be a properization pair, let $\rho$ be an arbitrary density operator, and let $\theta\in(0,1)$. There is an algorithm which, on the event that all estimates and improper-learning calls succeed, constructs mutually orthonormal explicit states $s_1,\ldots,s_m$, with $m\leq2/\theta$, and projectors
\begin{equation}
P=\sum_{a=1}^m\ket{s_a}\bra{s_a},\qquad Q=I-P,
\end{equation}
such that
\begin{equation}
\bra{Q\phi}\rho\ket{Q\phi}\le\theta , \qquad \forall \phi\in\mathcal{C}.
\label{eq:general-residual-smallness}
\end{equation}
Consequently, for $A=P\rho P$, 
\begin{equation}   \sup_{\phi\in\mathcal{C}}|\bra{\phi}\rho\ket{\phi}-\bra{\phi}A\ket{\phi}|\le 2 \sqrt{\theta}.
    \label{eq:general-core-approximation}
\end{equation}
\end{theorem}
\begin{proof}
Start with $P_0:=0$ and $Q_0:=I$. Suppose that mutually orthonormal
explicit states $s_1,\ldots,s_j$ have been constructed, and set $P_j:=\sum_{\ell=1}^j|s_\ell\rangle\langle s_\ell|$ and $Q_j:=I-P_j$. The improper learner will use the same comparator parameter $R'$ in every round. Define
\begin{equation}
\mu_j:=\tr(Q_j\rho)=\tr(Q_j\rho Q_j),
\label{eq:general-residual-mass}
\end{equation}
and estimate $\mu_j$ to additive accuracy $\theta/16$. If the estimate $\widehat\mu_j$ satisfies
$\widehat\mu_j\leq\theta/8$, then
\begin{equation*}
\|Q_j\rho Q_j\|_\infty\leq\mu_j \leq\frac{3\theta}{16}<\theta,
\end{equation*}
and the procedure terminates. Otherwise, $\mu_j>\theta/16$. Postselecting on the $Q_j$ outcome of the measurement $\{P_j,Q_j\}$ prepares copies of
\begin{equation}
\overline\rho_j:=\frac{Q_j\rho Q_j}{\mu_j}.
\label{eq:general-normalized-residual-state}
\end{equation}
The expected number of input copies per residual copy is at most $16/\theta$. Run the improper learner on $\overline\rho_j$ with the fixed comparator parameter $R'$ and error $\theta/16$, and denote its normalized output by $b_j$. Set
\begin{equation}
a_j:=\langle b_j|Q_j\rho Q_j|b_j\rangle.
\label{eq:general-residual-score}
\end{equation}
Since $\mathcal C\subseteq\mathcal A_{R'}$, the learning guarantee
implies
\begin{align}
a_j&\geq \opt_{\mathcal A_{R'}}(Q_j\rho Q_j)-\mu_j\frac{\theta}{16}\nonumber\\&\geq\max_{\phi\in\mathcal C}\langle Q_j\phi|\rho|Q_j\phi\rangle-\frac{\theta}{16}.
\label{eq:general-fixed-comparator-residual-bound}
\end{align}
Here we used $\langle\phi|Q_j\rho Q_j|\phi\rangle=\langle Q_j\phi|\rho|Q_j\phi\rangle$.
In particular, the normalized residual of a comparator need not belong to $\mathcal A_{R'}$. Estimate $a_j$ to additive accuracy $\theta/16$.
This can be done by first measuring $\{P_j,Q_j\}$ and, on the $Q_j$ outcome, measuring $|b_j\rangle\langle b_j|$ since the joint success probability is $a_j$.
If $\widehat a_j\leq3\theta/4$, then
\begin{equation*}
\max_{\phi\in\mathcal C}\langle Q_j\phi|\rho|Q_j\phi\rangle
\leq a_j+\frac{\theta}{16}\leq\frac{7\theta}{8}<\theta,
\end{equation*}
and the procedure terminates. Otherwise,
$a_j>11\theta/16>\theta/2$. Define
\begin{equation}
q_j:=Q_jb_j, \qquad s_{j+1}:=\frac{q_j}{\|q_j\|_2}.
\label{eq:general-new-direction}
\end{equation}
The vector $q_j$ is explicit by the linear combination assumption. It is nonzero because $\langle    
q_j|\rho|q_j\rangle=a_j>0$,
and it is orthogonal to the previously constructed states.
Moreover,
\begin{equation}
\langle s_{j+1}|\rho|s_{j+1}\rangle=\frac{a_j}{\|q_j\|_2^2}\geq a_j>\frac{\theta}{2}.
\label{eq:general-new-direction-weight}
\end{equation}
Thus, after $m$ additions,
\begin{equation*}
1=\tr(\rho)\geq\sum_{\ell=1}^m\langle s_\ell|\rho|s_\ell\rangle
>m\frac{\theta}{2}.
\end{equation*}
Consequently, $m<2/\theta$, and the procedure terminates. At termination, either the residual operator norm is less than $\theta$, or its score against every comparator in $\mathcal C$
is less than $\theta$. In either case, the final projector $Q:=I-P$ satisfies \eqref{eq:general-residual-smallness}.

For the objective approximation, write
$\|x\|_\rho:=\sqrt{\langle x|\rho|x\rangle}$.
For every $\phi\in\mathcal C$,
\begin{equation*}
\|Q\phi\|_\rho\leq\sqrt\theta,\qquad \left|\|\phi\|_\rho-\|P\phi\|_\rho\right|\leq\|Q\phi\|_\rho.
\end{equation*}
Since $\rho$ is a density operator and $\phi$ is normalized,
both $\|\phi\|_\rho$ and $\|P\phi\|_\rho$ are at most one.
Therefore
\begin{align*}
\left|\langle\phi|\rho|\phi\rangle-\langle\phi|P\rho P|\phi\rangle\right|&=\left|\|\phi\|_\rho^2-\|P\phi\|_\rho^2\right|\\&\leq\left|\|\phi\|_\rho-\|P\phi\|_\rho\right|\left(\|\phi\|_\rho+\|P\phi\|_\rho\right)\\&\leq2\sqrt\theta.
\end{align*}
Taking the supremum proves
\eqref{eq:general-core-approximation}.
\end{proof}

\subsection{Optimizing on the relevant subspace}

Let $s_1,\ldots,s_m$, $P$, and $A=P\rho P$ be as above. Define the isometry $S:\mathbb C^m \to \mathcal H$ by $Se_a:=s_a$, and the support matrix
\begin{equation} G:=S^\dagger\rho S,\qquad G_{ab}=\langle s_a|\rho|s_b\rangle. \label{eq:general-support-matrix} \end{equation}
Since $S^\dagger S=I_m$ and $SS^\dagger=P$, we have $A=SGS^\dagger$.
Estimate each matrix element $G_{ab}=\langle s_a|\rho|s_b\rangle$ to additive accuracy $O(\xi/m)$, and symmetrize the resulting matrix to obtain a Hermitian estimate $G_0$. Then
\begin{equation}
\|G_0-G\|_\infty\leq\frac{\xi}{2}.
\label{eq:general-raw-support-estimate}
\end{equation}
Let $\widehat G:=(G_0)_+$ be obtained by replacing the negative eigenvalues of $G_0$ by zero. Since $G\succeq 0$, every negative eigenvalue of $G_0$ has absolute value at most $\xi/2$, and hence
\begin{equation} \widehat G\succeq 0,\qquad \|\widehat G-G\|_\infty\leq \xi. \label{eq:general-psd-support-estimate} \end{equation}
Define
$\widehat A:=S\widehat G S^\dagger.$
Then
\begin{equation} \|\widehat A-A\|_\infty=\|\widehat G-G\|_\infty\leq\xi. \label{eq:general-core-operator-error} \end{equation}
\begin{theorem}
\label{thm:general-low-rank-properization}
For every $\zeta\in(0,1)$, one can output a normalized $\widehat\phi\in\mathcal C$ such that
\begin{equation} \langle\widehat\phi|A|\widehat\phi\rangle\geq\opt_{\mathcal C}(A)-2\xi-\zeta. \label{eq:general-low-rank-properization-guarantee} \end{equation}
The number of calls to the proper pure-target optimizer is at most $\left(\frac{80}{\zeta}\right)^{\! 2m}$.
\end{theorem}

\begin{proof}
Since $A=P\rho P$, one has $\|A\|_\infty\leq1$. If $\xi>1$, return $\phi_0$, in that case the bound is immediate since $\opt_{\mathcal C}(A)\leq1$. Otherwise, $\xi\leq1$, and \eqref{eq:general-core-operator-error} gives $\|\widehat A\|_\infty\leq2$. If $\widehat A=0$, return the fixed explicit state $\phi_0\in\mathcal C$. In this case
\begin{equation} \opt_{\mathcal C}(A)\leq\|A-\widehat A\|_\infty\leq\xi, \label{eq:general-zero-estimated-core} \end{equation}
so \eqref{eq:general-low-rank-properization-guarantee} is immediate. We may therefore assume $\widehat A\neq 0$. Write $\widehat A=\sum_{j=1}^m \lambda_j |u_j\rangle\langle u_j|$, where $\lambda_j\geq 0$ and $u_1,\ldots,u_m$ form an orthonormal basis of $\operatorname{ran}P$. For $c\in\mathbb C^m$ with $\|c\|_2=1$, define
\begin{equation} v_c:=\sum_{j=1}^m\sqrt{\lambda_j}\,c_j u_j. \label{eq:general-vc-definition} \end{equation}
For every normalized $\phi\in\mathcal H$,
\begin{equation} \langle\phi|\widehat A|\phi\rangle=\max_{\|c\|_2=1}|\langle v_c|\phi\rangle|^2. \label{eq:general-mixed-to-pure-identity} \end{equation}
Consequently,
\begin{equation} \opt_{\mathcal C}(\widehat A)=\max_{\|c\|_2=1}\max_{\phi\in\mathcal C}|\langle v_c|\phi\rangle|^2. \label{eq:general-core-opt-as-pure-targets} \end{equation}

\noindent Let $\mathcal N$ be a Euclidean $\zeta/16$-net of the unit sphere in $\mathbb C^m$. It may be chosen with $|\mathcal N|\leq\left(\frac{80}{\zeta}\right)^{2m}. $ Since the vectors $u_j$ are orthonormal, 
\begin{equation}
 \|v_c\|_2^2=\sum_{j=1}^m \lambda_j |c_j|^2 \leq \|\widehat A\|_\infty \|c\|_2^2,
\end{equation}
so the map $c\mapsto v_c$ has operator norm at most
$\sqrt{\|\widehat A\|_\infty}\leq \sqrt{2}$.
 Hence, for unit vectors $c,c'$,
\begin{equation} \|v_c-v_{c'}\|_2\leq\sqrt 2\|c-c'\|_2. \label{eq:general-vc-lipschitz} \end{equation}
For every normalized $\phi$,
\begin{align} \left||\langle v_c|\phi\rangle|^2-|\langle v_{c'}|\phi\rangle|^2\right|&\leq\left(\|v_c\|_2+\|v_{c'}\|_2\right)\|v_c-v_{c'}\|_2 \nonumber\\ &\leq4\|c-c'\|_2. \label{eq:general-squared-overlap-lipschitz} \end{align}

\noindent For every $c\in\mathcal N$, set $\phi_c:=\phi_0$ if $v_c=0$. Otherwise, let
\begin{equation} w_c:=\frac{v_c}{\|v_c\|_2}. \label{eq:general-normalized-vc} \end{equation}
Apply the proper pure-target optimizer to $w_c$ with $\eta=\zeta/4$. The resulting $\phi_c\in\mathcal C$ satisfies
\begin{equation}
    |\langle w_c|\phi_c\rangle|^2\geq\max_{\phi\in\mathcal C}|\langle w_c|\phi\rangle|^2-\frac{\zeta}{4}
\end{equation}
and since $\|v_c\|_2^2\le 2$,
\begin{equation} |\langle v_c|\phi_c\rangle|^2\geq\max_{\phi\in\mathcal C}|\langle v_c|\phi\rangle|^2-\|v_c\|_2^2\,\frac{\zeta}{4}\ge \max_{\phi\in\mathcal C}|\langle v_c|\phi\rangle|^2-\frac{\zeta}{2}. \label{eq:general-rescaled-pure-target-guarantee} \end{equation}
The same rescaled inequality holds trivially when $v_c=0$.
Return the candidate $\widehat\phi$ with largest exactly computed $\widehat A$-score. Let $(c^\star,\phi^\star)$ attain the maximum in \eqref{eq:general-core-opt-as-pure-targets}, and choose $c\in\mathcal N$ with $\|c-c^\star\|_2\leq\frac{\zeta}{16}$. Then
\begin{align} \langle\phi_c|\widehat A|\phi_c\rangle&\geq|\langle v_c|\phi_c\rangle|^2\nonumber\\ &\geq|\langle v_c|\phi^\star\rangle|^2-\frac{\zeta}{2} \nonumber\\ &\geq|\langle v_{c^\star}|\phi^\star\rangle|^2-\frac{\zeta}{4}-\frac{\zeta}{2} \nonumber\\ &\geq\opt_{\mathcal C}(\widehat A)-\frac{3\zeta}{4}. \label{eq:general-net-candidate-score} \end{align}
Therefore
\begin{equation} \langle\widehat\phi|\widehat A|\widehat\phi\rangle\geq\opt_{\mathcal C}(\widehat A)-\zeta. \label{eq:general-estimated-core-optimization} \end{equation}
Using \eqref{eq:general-core-operator-error},
\begin{align} \langle\widehat\phi|A|\widehat\phi\rangle&\geq\langle\widehat\phi|\widehat A|\widehat\phi\rangle-\xi \nonumber\\&\geq\opt_{\mathcal C}(\widehat A)-\zeta-\xi \nonumber\\&\geq\opt_{\mathcal C}(A)-2\xi-\zeta. \label{eq:general-transfer-back-to-core} \end{align}
\end{proof}

\begin{lemma}
\label{lem:general-spectral-cutoff}
Let $A=P\rho P$ and let $\widehat A\succeq0$ satisfy
$\|\widehat A-A\|_\infty\leq\xi$, where $0<\xi<\tau<1$.
From the spectral decomposition
$\widehat A=\sum_j\lambda_j|u_j\rangle\langle u_j|$, define
\begin{equation}
\widehat A_\tau:=\sum_{\lambda_j>\tau}
\lambda_j|u_j\rangle\langle u_j|,\qquad\ell:=\operatorname{rank}(\widehat A_\tau).
\label{eq:general-truncated-core}
\end{equation}
Then
\begin{equation}
\|\widehat A_\tau-A\|_\infty\leq\xi+\tau,\qquad
\ell\leq\min\left\{m,\frac{1}{\tau-\xi}\right\}.
\label{eq:general-truncated-core-bounds}
\end{equation}
For every $\zeta\in(0,1)$, at most
$(80/\zeta)^{2\ell}$ calls to the proper pure-target optimizer
suffice to return $\widehat\phi\in\mathcal C$ satisfying
\begin{equation}
\langle\widehat\phi|A|\widehat\phi\rangle\geq\opt_{\mathcal C}(A)-2(\xi+\tau)-\zeta.
\label{eq:general-truncated-core-optimization}
\end{equation}
\end{lemma}

\begin{proof}
The first bound is straightforward since discarding eigenvalues at most $\tau$ changes $\widehat A$ by at most $\tau$ in operator norm.

Let $\Pi_\tau$ project onto the retained eigenspace.
If $\ell>0$, then
\begin{equation*}
1\geq\tr(A) \geq\tr(\Pi_\tau A) \geq\tr(\Pi_\tau\widehat A)-\ell\xi>\ell(\tau-\xi).
\end{equation*}
This proves the rank bound. Since $\|\widehat A_\tau\|_\infty\leq1+\xi\leq2$, the coefficient-net argument in the proof of Theorem~\ref{thm:general-low-rank-properization} applies in the $\ell$-dimensional retained eigenspace. It returns $\widehat\phi\in\mathcal C$ with
\begin{equation*}
\langle\widehat\phi|\widehat A_\tau|\widehat\phi\rangle\geq\opt_{\mathcal C}(\widehat A_\tau)-\zeta
\end{equation*}
using at most $(80/\zeta)^{2\ell}$ pure-target calls. When $\ell=0$, return $\phi_0$. Then,
\begin{align}
    \langle\widehat\phi| A|\widehat\phi\rangle&\geq \langle\widehat\phi|\widehat A_\tau|\widehat\phi\rangle-(\xi+\tau)\\
    &\ge \opt_{\mathcal{C}}(\widehat A_\tau)-\zeta-(\xi+\tau)\\
    &\ge \opt_{\mathcal{C}}(A)-\zeta -2(\xi+\tau)
\end{align}
\end{proof}

\begin{theorem}[Improper to proper theorem]
\label{thm:general-improper-to-proper}
Let $(\mathcal C,(\mathcal V_R)_{R\geq1})$ be a properization pair. For every density operator $\rho$ and every $\varepsilon,\delta\in(0,1/4)$, there is a randomized algorithm using $O(\varepsilon^{-2})$ improper-learning calls, all with the fixed comparator parameter $R'$, and at most
$\left(\frac{C}{\varepsilon}\right)^{O(1/\varepsilon)}$ proper pure-target calls, which outputs a normalized state $\widehat\phi\in\mathcal C$ such that, with probability at least $1-\delta$,
\begin{equation}
\langle\widehat\phi|\rho|\widehat\phi\rangle\geq\opt_{\mathcal C}(\rho)-\varepsilon.
\label{eq:general-proper-learning-guarantee}
\end{equation}
\end{theorem}

\begin{proof}
Choose
\begin{equation}
\theta:=\left(\frac{\varepsilon}{16}\right)^2,\qquad\xi:=\frac{\varepsilon}{32},
\qquad\tau:=\frac{\varepsilon}{8},\qquad\zeta:=\frac{\varepsilon}{4}.
\label{eq:general-main-parameters}
\end{equation}
Theorem~\ref{thm:general-finite-relevant-subspace} constructs
$P$ and $A=P\rho P$ such that
\begin{equation*}
\sup_{\phi\in\mathcal C}\left|\langle\phi|\rho|\phi\rangle-\langle\phi|A|\phi\rangle\right|\leq2\sqrt\theta.
\end{equation*}
Estimate the support matrix to obtain $\widehat A\succeq0$
with $\|\widehat A-A\|_\infty\leq\xi$, and apply Lemma~\ref{lem:general-spectral-cutoff}. The returned state satisfies
\begin{align*}
\langle\widehat\phi|\rho|\widehat\phi\rangle&\geq\opt_{\mathcal C}(\rho)-4\sqrt\theta-2(\xi+\tau)-\zeta\\&=\opt_{\mathcal C}(\rho)-\frac{13\varepsilon}{16}.
\end{align*}
The retained rank obeys
\begin{equation*}
\ell\leq\frac{1}{\tau-\xi}=\frac{32}{3\varepsilon}.
\end{equation*}
Since at most $2/\theta$ directions can be added, cap the procedure at $\lceil2/\theta\rceil+1$ iterations and distribute the failure probability over the estimates, improper-learning calls, and postselection steps. A postselection step requesting $N$ residual copies can be capped at $O(\theta^{-1}(N+\log(1/\beta)))$ attempts with failure probability at most $\beta$ because its success probability is greater than $\theta/16$ whenever the preceding estimate succeeds. On a declared failure, return $\phi_0$. A union bound gives overall success probability at least $1-\delta$.
\end{proof}

\section{Matrix Product States}\label{sec:matrix-product-states}
Fix local Hilbert spaces $\mathcal H_i\cong \mathbb C^d$ and denote $\mathcal H_{a:b}:=\bigotimes_{i=a}^b\mathcal H_i$.  As a special case, $\mathcal H:=\mathcal H_{1:n}$. A vector $u\in\mathcal H$ is an open-boundary matrix product state of bond dimension at most $D$ if there are integers
\begin{equation} r_0=r_n=1, \qquad r_i\leq D\quad(1\leq i\leq n-1), \label{eq:mps-bond-dimensions} \end{equation}
and matrices
\begin{equation} A_i^s:\mathbb C^{r_{i-1}}\to \mathbb C^{r_i}, \qquad 1\leq s\leq d, \label{eq:mps-local-matrices} \end{equation}
such that
\begin{equation} u_{s_1,\ldots,s_n}=A_n^{s_n}\cdots A_1^{s_1}. \label{eq:mps-amplitudes} \end{equation}
Let $\mps_D$ denote the normalized open-boundary MPSs of bond dimension at most $D$. For a positive semidefinite operator $X\succeq 0$, write
\begin{equation} \opt_D(X):=\max_{\psi\in\mps_D}\langle\psi|X|\psi\rangle. \label{eq:mps-optimum} \end{equation}

\begin{lemma}
\label{lem:mps-cut-rank-characterization}
A vector $u\in\mathcal H$ has an open-boundary MPS representation of bond dimension at most $D$ if and only if
\begin{equation} \operatorname{rank}_{\mathcal H_{1:i}\,|\,\mathcal H_{i+1:n}}(u)\leq D \qquad \text{for every }1\leq i\leq n-1. \label{eq:mps-cut-rank-condition} \end{equation}
\end{lemma}

\begin{proof}
Standard.
\end{proof}

By Lemma~\ref{lem:mps-cut-rank-characterization}, $\mps_D$ is the intersection of the unit sphere with finitely many closed rank-constrained sets, and is therefore compact.

\subsection{Improper learning and explicit operations}

\begin{lemma}
\label{lem:mps-linear-combinations}
Let $u_a$ be an open-boundary MPS of bond dimension at most $R_a$ for $a=1,\ldots,q$. Then, for arbitrary coefficients $c_a\in\mathbb C$, the vector $u:=\sum_{a=1}^q c_a u_a$ has an explicit open-boundary MPS representation of bond dimension at most $\sum_{a=1}^q R_a$.
\end{lemma}

\begin{proof}
Standard.
\end{proof}

\begin{theorem}[Agnostic improper learning of MPSs]
\label{thm:mps-improper-learner}
Given copies of an arbitrary normalized $n$-qudit state $\omega$, a comparator bond $R$, an error $\gamma\in(0,1)$, and a failure probability $\beta\in(0,1)$, there is an algorithm which outputs a normalized explicit open-boundary MPS $b$ of bond dimension at most $d n^2 p(R,1/\gamma) $ for a polynomial $p$ such that, with probability at least $1-\beta$,
\begin{equation} \langle b|\omega|b\rangle\geq\opt_R(\omega)-\gamma. \label{eq:mps-improper-guarantee} \end{equation}
The runtime and copy complexity are bounded by $\poly\left(n,d,R,1/\gamma,\log(1/\beta)\right).$
\end{theorem}

\begin{proof}
See \cite[Theorem~B.2]{BakshiEtAl2025}.
\end{proof}
\begin{lemma}
\label{lem:mps-explicit-operations}
The explicit operations in item~5 of Definition~\ref{def:properization-data} can be implemented for MPSs with overhead polynomial in the system size, the local dimension, the number of input vectors, their bond dimensions, and the inverse accuracy.
\end{lemma}

\begin{proof}
The product state $|0\rangle^{\otimes n}$ supplies a fixed explicit element of $\mps_D$. Norms and overlaps are computed
by transfer matrix contractions, and
Lemma~\ref{lem:mps-linear-combinations} provides explicit linear combinations. Normalizing a nonzero vector does not
increase its bond dimension.

A normalized open-boundary MPS of bond $B$ has a sequential isometric preparation using an ancilla of dimension at most
$B$ and $n$ local isometries. Consequently, a preparation unitary and its inverse implement the reflection $I-2|s\rangle\langle s|$ about an explicit normalized MPS $s$, with polynomial overhead.

For mutually orthonormal MPSs $s_1,\ldots,s_m$, let
$P:=\sum_{a=1}^m|s_a\rangle\langle s_a|$. Their rank-one projectors are mutually orthogonal, so
\begin{equation*}
\prod_{a=1}^m \left(I-2|s_a\rangle\langle s_a|\right) =I-2P.
\end{equation*}
A controlled implementation of this reflection, with a
Hadamard gate before and after the control,  
implements the projective measurement $\{P,I-P\}$. Its cost is polynomial in $m$ and in the costs of the individual MPS preparations. The diagonal matrix elements
$\langle s_a|\rho|s_a\rangle$ are projective probabilities.
The real and imaginary parts of the off-diagonal elements
are obtained by polarization, using the normalized states $\frac{s_a\pm s_b}{\sqrt2}$, $\frac{s_a\pm i s_b}{\sqrt2}$. These states have bond dimension at most the sum of the two input bonds. For general explicit vectors, the same argument uses the normalized nonzero combinations and their known
norms. Repeated measurements estimate the required
probabilities with polynomial copy complexity and logarithmic dependence on the inverse failure probability.
\end{proof}

\subsection{Comparator dual compression}

We first define the norm in which the target will be compressed. Let $\mathcal K_0,\ldots,\mathcal K_m$ be finite-dimensional Hilbert spaces and write
\begin{equation} \mathcal K_{a:b}:=\bigotimes_{j=a}^b\mathcal K_j. \label{eq:mps-general-chain-notation} \end{equation}
For $T\in\mathcal K_{0:m}$, define
\begin{equation} \|T\|_{\mathrm{MPS}(D)}:=\sup_Z |\langle Z,T\rangle|, \label{eq:mps-comparator-dual-norm} \end{equation}
where the supremum is over all $Z\in\mathcal K_{0:m}$ satisfying
\begin{equation} \|Z\|_2\leq 1,\qquad\operatorname{rank}_{\mathcal K_{0:j}\,|\,\mathcal K_{j+1:m}}(Z)\leq D \quad \text{for every }0\leq j\leq m-1. \label{eq:mps-comparator-dual-admissibility} \end{equation}
If $\mathcal K_0\cong\mathbb C$ and $\mathcal K_i=\mathcal H_i$ for $1\leq i\leq n$, then Lemma~\ref{lem:mps-cut-rank-characterization} gives
\begin{equation} \|T\|_{\mathrm{MPS}(D)}=\sup_{\psi\in\mps_D}|\langle\psi,T\rangle|. \label{eq:mps-dual-norm-physical-interpretation} \end{equation}

\begin{lemma}
\label{lem:mps-prefix-contraction}
Let a contiguous prefix of a chain have Hilbert space $\mathcal U$ and let the remaining suffix have Hilbert space $\mathcal R$. Replace the prefix by one boundary space $\mathcal B$. If $C:\mathcal B\longrightarrow\mathcal U $
is a contraction, then every $X\in\mathcal B\otimes\mathcal R$ satisfies
\begin{equation} \|(C\otimes I_{\mathcal R})X\|_{\mathrm{MPS}(D)}\leq\|X\|_{\mathrm{MPS}(D)}, \label{eq:mps-prefix-contraction-bound} \end{equation}
where the norm on the right is taken on the quotient chain with boundary space $\mathcal B$.
\end{lemma}

\begin{proof}
Let $Z$ be an admissible comparator for the expanded chain and define
\begin{equation} Z':=(C^\dagger\otimes I_{\mathcal R})Z. \label{eq:mps-pulled-back-comparator} \end{equation}
Since $C$ is a contraction, $\|Z'\|_2\leq 1$. Across every cut of the quotient chain, $Z'$ is obtained by applying a linear map on one side of the corresponding cut of $Z$ so its Schmidt rank cannot increase. Thus $Z'$ is an admissible comparator for the quotient chain, and
\begin{equation} |\langle Z,(C\otimes I_{\mathcal R})X\rangle|=|\langle Z',X\rangle|\leq\|X\|_{\mathrm{MPS}(D)}. \label{eq:mps-prefix-contraction-proof} \end{equation}
Taking the supremum over $Z$ proves the claim.
\end{proof}

\begin{lemma}
\label{lem:mps-boundary-direct-sum}
Suppose a boundary space decomposes as
$\mathcal B=\mathcal B_1\oplus\mathcal B_2. $ For $X_a\in\mathcal B_a\otimes\mathcal R$, one has
\begin{equation} \|X_1\oplus X_2\|_{\mathrm{MPS}(D)}^2\leq\|X_1\|_{\mathrm{MPS}(D)}^2+\|X_2\|_{\mathrm{MPS}(D)}^2, \label{eq:mps-boundary-direct-sum-bound} \end{equation}
where the two norms on the right use the corresponding boundary spaces $\mathcal B_1$ and $\mathcal B_2$.
\end{lemma}

\begin{proof}
Let $Z=Z_1\oplus Z_2$ be an admissible comparator for the direct-sum boundary. Projection onto either boundary sector cannot increase any Schmidt rank, and
\begin{equation} \|Z_1\|_2^2+\|Z_2\|_2^2=\|Z\|_2^2\leq1. \label{eq:mps-projected-comparator-norms} \end{equation}
Therefore, by Cauchy--Schwarz,
\begin{align} |\langle Z,X_1\oplus X_2\rangle|&\leq\|Z_1\|_2\|X_1\|_{\mathrm{MPS}(D)}+\|Z_2\|_2\|X_2\|_{\mathrm{MPS}(D)} \nonumber\\ &\leq\left(\|X_1\|_{\mathrm{MPS}(D)}^2+\|X_2\|_{\mathrm{MPS}(D)}^2\right)^{1/2}. \label{eq:mps-boundary-direct-sum-proof} \end{align}
Taking the supremum and squaring proves \eqref{eq:mps-boundary-direct-sum-bound}.
\end{proof}

\begin{lemma}
\label{lem:mps-flat-loss-bound}
Let $L\in\mathcal B\otimes\mathcal R$ be seen as a matrix from $\mathcal R$ to the boundary space $\mathcal B$. If every singular value of $L$ is at most $\sqrt\Delta$, then
\begin{equation} \|L\|_{\mathrm{MPS}(D)}^2\leq D\Delta. \label{eq:mps-flat-loss-bound} \end{equation}
\end{lemma}

\begin{proof}
Let $Z$ be the matrix associated with an admissible comparator. Across the boundary cut, $\operatorname{rank}(Z)\leq D$ and $\|Z\|_2\leq1$. Writing $s_j(X)$ for the $j$-th singular value of $X$, von Neumann's trace inequality and Cauchy--Schwarz give
\begin{align} |\langle Z,L\rangle|&=|\operatorname{tr}(Z^\dagger L)| \nonumber\\ &\leq\sum_{j=1}^D s_j(Z)s_j(L) \nonumber\\ &\leq\left(\sum_{j=1}^D s_j(Z)^2\right)^{1/2}\left(\sum_{j=1}^D s_j(L)^2\right)^{1/2} \nonumber\\ &\leq\sqrt{D\Delta}. \label{eq:mps-flat-loss-proof} \end{align}
Taking the supremum proves \eqref{eq:mps-flat-loss-bound}.
\end{proof}

\begin{theorem}[Comparator-dual compression for MPSs]
\label{thm:mps-comparator-dual-compression}
\label{compression-thm}
For every vector $u\in\mathcal H$, every $D\geq 1$, and every $K\geq 1$, one can construct an MPS $\widetilde u_K$ of bond dimension at most $K$ satisfying
\begin{equation} \|u-\widetilde u_K\|_{\mathrm{MPS}(D)}\leq\|u\|_2\sqrt{\frac{D}{K+1}}. \label{eq:mps-compression-guarantee} \end{equation}
Moreover, $ \|\widetilde u_K\|_2\leq\|u\|_2 $. If $u$ is an explicit MPS of bond dimension $B$, the construction uses
$\poly(n,d,B,K) $ arithmetic operations.
\end{theorem}
\noindent Algorithm~\ref{alg:mps-dual-compression} gives the construction.

\begin{algorithm}[htbp]
\caption{Comparator-dual compression of an explicit MPS}
\label{alg:mps-dual-compression}
\small
\begin{algorithmic}[1]
\Require An explicit $n$-site MPS $u$ and a bond cap $K\geq1$.
\Ensure An MPS $\widetilde u_K$ of bond at most $K$ satisfying
        Theorem \ref{compression-thm}.

\State $\mathcal B_0\gets\mathbb C$; $T_1\gets u$.
\For{$i=1,\ldots,n-1$}
    \If{$T_i=0$}
        \State \Return the zero MPS.
    \EndIf
    \State Regard $T_i$ as a matrix $M_i:\mathcal H_{i+1:n}\to\mathcal B_{i-1}\otimes\mathcal H_i$.

    \State Compute a compact SVD
$M_i=\sum_{j=1}^{p} \sigma_j|u_j\rangle\langle v_j|$, with $\sigma_1\geq\cdots\geq\sigma_p>0$.
    \State $k\gets\min\{K,p\}$; $\Delta_i\gets0$.
    \If{$p>K$}
        \State $\Delta_i\gets\sigma_{K+1}^2$.
    \EndIf
    \State $b_j\gets\sqrt{\sigma_j^2-\Delta_i}$ for $j=1,\ldots,k$; $\mathcal B_i\gets\mathbb C^k$.

    \State Store $C_i^{\mathrm{keep}}\gets\sum_{j=1}^{k}(b_j/\sigma_j)|u_j\rangle\langle j|$.
    \State $T_{i+1}\gets \sum_{j=1}^{k}b_j|j\rangle\langle v_j|$.
\EndFor

\State \Return the MPS $\widetilde u_K$ represented by $(C_1^{\mathrm{keep}},\ldots, C_{n-1}^{\mathrm{keep}},T_n)$.
\end{algorithmic}
\end{algorithm}

\begin{proof}
Set $\mathcal B_0:=\mathbb C$ and $T_1:=u$. Suppose that, at step $i$, we have a boundary space $\mathcal B_{i-1}$ with $\dim\mathcal B_{i-1}\leq K $ and a current remainder
$T_i\in\mathcal B_{i-1}\otimes\mathcal H_{i:n}. $ 
The vector $T_i$ should be viewed as the unprocessed remainder after sites $1,\ldots,i-1$ have been compressed into the effective left boundary $\mathcal B_{i-1}$. Thus all information from the processed prefix that is retained by the construction is carried by a single boundary index of dimension at most $K$. The $i$th step moves this effective boundary one site to the right by replacing $\mathcal B_{i-1}\otimes\mathcal H_i$ by a new boundary space $\mathcal B_i$ of dimension at most $K$. Regard $T_i$ as a matrix $M_i:\mathcal H_{i+1:n}\longrightarrow\mathcal B_{i-1}\otimes\mathcal H_i$
and take an SVD
\begin{equation} M_i=U\Sigma V^\dagger=\sum_{j=1}^p\sigma_j|u_j\rangle\langle v_j|,\qquad \sigma_1\geq\cdots\geq\sigma_p>0. \label{eq:mps-compression-svd} \end{equation}

The following step is a modified rank-$K$ SVD truncation. Ordinary truncation would keep the first $K$ singular values unchanged and discard the remaining tail. Instead, when $p>K$, we set $\Delta_i=\sigma_{K+1}^2$ and decrease the square of each retained singular value by the same amount $\Delta_i$. This additional shrinkage is deliberately more lossy in Euclidean norm but it makes the loss flat at scale $\sqrt{\Delta_i}$ while forcing the squared Euclidean norm to decrease by at least $(K+1)\Delta_i$. The former makes the loss difficult for a bond-$D$ comparator to detect, and the latter gives the telescoping estimate needed to obtain the factor $D/(K+1)$ in the final error bound. Let
\begin{equation} k:=\min\{K,p\},\qquad \Delta_i:=\begin{cases} 0,&p\leq K,\\ \sigma_{K+1}^2,&p>K. \end{cases} \label{eq:mps-compression-delta} \end{equation}
For $j=1,\ldots,k$, define
\begin{equation} b_j:=\sqrt{\sigma_j^2-\Delta_i}. \label{eq:mps-compression-shrunken-singular-values} \end{equation}
Let $\mathcal B_i:=\mathbb C^k$ and define the new remainder
\begin{equation} T_{i+1}:=\sum_{j=1}^k b_j|j\rangle\langle v_j|\in\mathcal B_i\otimes\mathcal H_{i+1:n}. \label{eq:mps-compression-next-remainder} \end{equation}
Define the retained map
\begin{equation}
C_i^{\mathrm{keep}}:=\sum_{j=1}^k\frac{b_j}{\sigma_j}|u_j\rangle\langle j|
:\mathcal B_i\longrightarrow\mathcal B_{i-1}\otimes\mathcal H_i.
\label{eq:mps-compression-kept-map}
\end{equation}

We next represent the part removed at this step by a loss tensor $L_i$ and a corresponding map $C_i^{\mathrm{loss}}$. If $p\leq K$, there is no loss, so set $L_i=0$ and $C_i^{\mathrm{loss}}=0$. Suppose now that $p>K$, so that $k=K$. Write
\begin{equation}
U_k=[u_1\ \cdots\ u_k],\qquad V_k=[v_1\ \cdots\ v_k],
\label{eq:mps-compression-leading-singular-vectors}
\end{equation}
\begin{equation}
U_>=[u_{k+1}\ \cdots\ u_p],\qquad
V_>=[v_{k+1}\ \cdots\ v_p],\qquad
\Sigma_>:=\operatorname{diag}(\sigma_{k+1},\ldots,\sigma_p).
\label{eq:mps-compression-tail-singular-vectors}
\end{equation}
Let $\mathcal L_i:=\mathbb C^k\oplus\mathbb C^{p-k}$, and define
\begin{equation} L_i:=\begin{bmatrix} \sqrt{\Delta_i}\,V_k^\dagger\\ \Sigma_>V_>^\dagger \end{bmatrix}:\mathcal H_{i+1:n}\longrightarrow\mathcal L_i, \label{eq:mps-compression-loss-remainder} \end{equation}
\begin{equation} C_i^{\mathrm{loss}}:=\left[ U_k\operatorname{diag}\left(\frac{\sqrt{\Delta_i}}{\sigma_1},\ldots,\frac{\sqrt{\Delta_i}}{\sigma_k}\right) \quad U_> \right]:\mathcal L_i\longrightarrow\mathcal B_{i-1}\otimes\mathcal H_i. \label{eq:mps-compression-loss-map} \end{equation}
One can check that
\begin{equation} M_i=C_i^{\mathrm{keep}}T_{i+1}+C_i^{\mathrm{loss}}L_i. \label{eq:mps-compression-one-step-decomposition} \end{equation}
Indeed,
\begin{equation} C_i^{\mathrm{keep}}T_{i+1}=U_k\operatorname{diag}\left(\frac{b_1^2}{\sigma_1},\ldots,\frac{b_k^2}{\sigma_k}\right)V_k^\dagger, \label{eq:mps-compression-kept-product} \end{equation}
and
\begin{equation} C_i^{\mathrm{loss}}L_i=U_k\operatorname{diag}\left(\frac{\Delta_i}{\sigma_1},\ldots,\frac{\Delta_i}{\sigma_k}\right)V_k^\dagger+U_>\Sigma_>V_>^\dagger. \label{eq:mps-compression-loss-product} \end{equation}
Since $b_j^2+\Delta_i=\sigma_j^2$, \eqref{eq:mps-compression-kept-product} and \eqref{eq:mps-compression-loss-product} sum to $M_i$. Let $C_i:=\left[C_i^{\mathrm{keep}}\quad C_i^{\mathrm{loss}}\right].$
When $p>K$,
\begin{align} C_iC_i^\dagger&=U_k\operatorname{diag}\left(\frac{b_1^2+\Delta_i}{\sigma_1^2},\ldots,\frac{b_k^2+\Delta_i}{\sigma_k^2}\right)U_k^\dagger+U_>U_>^\dagger \nonumber\\ &=U_kU_k^\dagger+U_>U_>^\dagger\preceq I_{\mathcal B_{i-1}\otimes\mathcal H_i}. \label{eq:mps-compression-map-contraction} \end{align}
When $p\leq K$, the same conclusion follows from $C_i=C_i^{\mathrm{keep}}=U$. Thus $C_i$ is a contraction in all cases. Moreover, every singular value of $L_i$ is at most $\sqrt{\Delta_i}$. Hence Lemma~\ref{lem:mps-flat-loss-bound} gives
\begin{equation} \|L_i\|_{\mathrm{MPS}(D)}^2\leq D\Delta_i. \label{eq:mps-compression-one-step-loss-bound} \end{equation}

\noindent After processing sites $1,\ldots,n-1$, keep the final tensor $T_n$. Define
\begin{equation} \widetilde T_n:=T_n,\qquad \widetilde T_i:=C_i^{\mathrm{keep}}\widetilde T_{i+1},\qquad i=n-1,\ldots,1, \label{eq:mps-compression-retained-tensors} \end{equation}
and set
\begin{equation} \widetilde u_K:=\widetilde T_1. \label{eq:mps-compression-output} \end{equation}
Every retained boundary space has dimension at most $K$, so $\widetilde u_K$ is an MPS of bond dimension at most $K$. Define the accumulated error
\begin{equation} E_i:=T_i-\widetilde T_i. \label{eq:mps-compression-accumulated-error} \end{equation}
Then $E_n=0$, and $E_i=C_i(E_{i+1}\oplus L_i). $
By Lemmas~\ref{lem:mps-prefix-contraction} and \ref{lem:mps-boundary-direct-sum}, and \eqref{eq:mps-compression-one-step-loss-bound},
\begin{align} \|E_i\|_{\mathrm{MPS}(D)}^2&\leq\|E_{i+1}\oplus L_i\|_{\mathrm{MPS}(D)}^2 \nonumber\\ &\leq\|E_{i+1}\|_{\mathrm{MPS}(D)}^2+\|L_i\|_{\mathrm{MPS}(D)}^2 \nonumber\\ &\leq\|E_{i+1}\|_{\mathrm{MPS}(D)}^2+D\Delta_i. \label{eq:mps-compression-error-step} \end{align}
Telescoping gives
\begin{align}
\|u-\widetilde u_K\|_{\mathrm{MPS}(D)}^2 &=\|E_1\|_{\mathrm{MPS}(D)}^2-\|E_n\|_{\mathrm{MPS}(D)}^2\\
&=\sum_{i=1}^{n-1}\left(\|E_i\|_{\mathrm{MPS}(D)}^2-\|E_{i+1}\|_{\mathrm{MPS}(D)}^2\right)\leq D\sum_{i=1}^{n-1}\Delta_i.
    \label{eq:mps-compression-error-telescope}
\end{align}

Reshaping does not change Euclidean norm, so \eqref{eq:mps-compression-svd} and \eqref{eq:mps-compression-next-remainder} give
\begin{equation} \|T_i\|_2^2=\sum_{j=1}^p\sigma_j^2,\qquad \|T_{i+1}\|_2^2=\sum_{j=1}^k b_j^2. \label{eq:mps-compression-remainder-norms} \end{equation}
If $p>K$, then
\begin{align} \|T_i\|_2^2-\|T_{i+1}\|_2^2&=\sum_{j>K}\sigma_j^2+K\Delta_i \nonumber\\ &\geq\sigma_{K+1}^2+K\Delta_i \nonumber\\ &=(K+1)\Delta_i. \label{eq:mps-compression-energy-drop} \end{align}
If $p\leq K$, then $\Delta_i=0$, so the inequality is true. Therefore
\begin{align} (K+1)\sum_{i=1}^{n-1}\Delta_i&\leq\sum_{i=1}^{n-1}\left(\|T_i\|_2^2-\|T_{i+1}\|_2^2\right) \nonumber\\ &=\|T_1\|_2^2-\|T_n\|_2^2 \nonumber\\ &\leq\|u\|_2^2. \label{eq:mps-compression-delta-sum} \end{align}
Combining \eqref{eq:mps-compression-error-telescope} and \eqref{eq:mps-compression-delta-sum} proves \eqref{eq:mps-compression-guarantee}. Each $C_i^{\mathrm{keep}}$ is a contraction, and $\|T_{i+1}\|_2\leq\|T_i\|_2$. Hence
\begin{equation} \|\widetilde u_K\|_2\leq\|T_n\|_2\leq\|u\|_2.\label{eq:mps-compression-output-norm-proof} \end{equation}

If $u$ is an explicit MPS of bond $B$, each current remainder has an explicit MPS representation whose retained left boundary has dimension at most $K$ and whose remaining internal bonds have dimension at most $B$. The required reduced Gram matrices, SVDs, and updated tensors therefore have dimensions polynomial in $d$, $B$, and $K$.
\end{proof}
For $u\in\mathcal H$, define
\begin{equation}
e_D(u)^2:=\min_{\substack{a\in\mathbb C\\\psi\in\mps_D}}
\|u-a\psi\|_2^2=\|u\|_2^2-\|u\|_{\mathrm{MPS}(D)}^2.
\label{eq:mps-best-euclidean-error}
\end{equation}
\begin{corollary}
\label{cor:mps-relative-error-compression}
For every $K\geq D$, the vector $\widetilde u_K$ constructed
in Theorem~\ref{thm:mps-comparator-dual-compression} satisfies
\begin{equation}
\|u-\widetilde u_K\|_2^2
\leq \frac{K+1}{K+1-D}\,e_D(u)^2
\label{eq:mps-relative-euclidean-error}
\end{equation}
and
\begin{equation}
\|u-\widetilde u_K\|_{\mathrm{MPS}(D)}^2
\leq\frac{D}{K+1-D}\,e_D(u)^2.
\label{eq:mps-relative-comparator-error}
\end{equation}
In particular, for $\alpha\in(0,1)$, choosing
$K=D+\lceil D/\alpha\rceil$ gives a bond-$K$ approximation
whose squared Euclidean error is at most
$(1+\alpha)e_D(u)^2$. If $\|u\|_2=1$, one can return a normalized
$\widehat u_K\in\mps_K$ satisfying
\begin{equation}
1-|\langle u|\widehat u_K\rangle|^2\leq(1+\alpha)\left(
1-\max_{\psi\in\mps_D}|\langle u|\psi\rangle|^2\right).
\label{eq:mps-relative-infidelity}
\end{equation}
The runtime is the same as in Theorem~\ref{thm:mps-comparator-dual-compression}.
\end{corollary}

\begin{proof}
From the proof of Theorem~\ref{thm:mps-comparator-dual-compression}, we have
\begin{equation}
\|u-\widetilde u_K\|_{\mathrm{MPS}(D)}^2\leq D\sum_{i=1}^{n-1}\Delta_i,\qquad (K+1)\sum_{i=1}^{n-1}\Delta_i\leq \|u\|_2^2-\|T_n\|_2^2.
\label{eq:mps-relative-starting-bounds}
\end{equation}
 Recall that  $T_i=C_i(T_{i+1}\oplus L_i)$, where $C_i$ is a contraction and $\|L_i\|_{\mathrm{MPS}(D)}^2\leq D\Delta_i$. Applying Lemmas~\ref{lem:mps-prefix-contraction} and \ref{lem:mps-boundary-direct-sum} gives
\begin{equation}
\|T_i\|_{\mathrm{MPS}(D)}^2\leq \|T_{i+1}\|_{\mathrm{MPS}(D)}^2+\|L_i\|_{\mathrm{MPS}(D)}^2\leq \|T_{i+1}\|_{\mathrm{MPS}(D)}^2+D\Delta_i.
\label{eq:mps-relative-remainder-recursion}
\end{equation}
Since $T_1=u$, telescoping \eqref{eq:mps-relative-remainder-recursion} and using $\|T_n\|_{\mathrm{MPS}(D)}\leq\|T_n\|_2$ yields
\begin{equation}
\|u\|_{\mathrm{MPS}(D)}^2\leq \|T_n\|_2^2+D\sum_{i=1}^{n-1}\Delta_i.
\label{eq:mps-relative-remainder-telescope}
\end{equation}
Therefore, using $e_D(u)^2=\|u\|_2^2-\|u\|_{\mathrm{MPS}(D)}^2$,
\begin{equation}
\|u\|_2^2-\|T_n\|_2^2\leq \|u\|_2^2-\|u\|_{\mathrm{MPS}(D)}^2+D\sum_{i=1}^{n-1}\Delta_i=e_D(u)^2+D\sum_{i=1}^{n-1}\Delta_i.
\label{eq:mps-relative-norm-drop-upper-bound}
\end{equation}
Combining this with the second inequality in \eqref{eq:mps-relative-starting-bounds} gives
\begin{equation}
(K+1-D)\sum_{i=1}^{n-1}\Delta_i\leq e_D(u)^2,
\label{eq:mps-relative-delta-bound}
\end{equation}
and hence, since $K\geq D$,
\begin{equation}
\sum_{i=1}^{n-1}\Delta_i\leq \frac{e_D(u)^2}{K+1-D}.
\label{eq:mps-relative-delta-sum-bound}
\end{equation}
Substituting \eqref{eq:mps-relative-delta-sum-bound} into the first inequality in \eqref{eq:mps-relative-starting-bounds} gives
\begin{equation}
\|u-\widetilde u_K\|_{\mathrm{MPS}(D)}^2\leq \frac{D}{K+1-D}\,e_D(u)^2,
\label{eq:mps-relative-comparator-error-proof}
\end{equation}

We now bound the Euclidean error. Recall that $E_i:=T_i-\widetilde T_i$, so $E_n=0$ and $E_i=C_i(E_{i+1}\oplus L_i)$. Since $C_i$ is a contraction,
\begin{equation}
\|E_i\|_2^2\leq \|E_{i+1}\oplus L_i\|_2^2=\|E_{i+1}\|_2^2+\|L_i\|_2^2.
\label{eq:mps-relative-euclidean-error-recursion}
\end{equation}
At a truncating step, the definition of $L_i$ gives
\begin{equation}
\|L_i\|_2^2=K\Delta_i+\sum_{j>K}\sigma_j^2=\|T_i\|_2^2-\|T_{i+1}\|_2^2.
\label{eq:mps-relative-loss-norm}
\end{equation}
At an exact step both sides of \eqref{eq:mps-relative-loss-norm} are zero. Telescoping \eqref{eq:mps-relative-euclidean-error-recursion} therefore gives
\begin{equation}
\|u-\widetilde u_K\|_2^2=\|E_1\|_2^2\leq \sum_{i=1}^{n-1}\|L_i\|_2^2=\sum_{i=1}^{n-1}\left(\|T_i\|_2^2-\|T_{i+1}\|_2^2\right)=\|u\|_2^2-\|T_n\|_2^2.
\label{eq:mps-relative-euclidean-vs-norm-drop}
\end{equation}
Using \eqref{eq:mps-relative-norm-drop-upper-bound} and \eqref{eq:mps-relative-delta-sum-bound},
\begin{equation}
\|u-\widetilde u_K\|_2^2\leq e_D(u)^2+D\sum_{i=1}^{n-1}\Delta_i\leq \left(1+\frac{D}{K+1-D}\right)e_D(u)^2=\frac{K+1}{K+1-D}\,e_D(u)^2,
\label{eq:mps-relative-euclidean-error-proof}
\end{equation}
proving \eqref{eq:mps-relative-euclidean-error}.

Finally, suppose $\|u\|_2=1$. If $\widetilde u_K\neq0$, set $\widehat u_K:=\widetilde u_K/\|\widetilde u_K\|_2$. Since $\widetilde u_K$ is a scalar multiple of $\widehat u_K$,
\begin{equation}
1-|\langle u|\widehat u_K\rangle|^2=\min_{a\in\mathbb C}\|u-a\widehat u_K\|_2^2\leq \|u-\widetilde u_K\|_2^2\leq(1+\alpha)e_D(u)^2.
\label{eq:mps-relative-infidelity-proof}
\end{equation}
Using \eqref{eq:mps-best-euclidean-error}, this is \eqref{eq:mps-relative-infidelity}. If $\widetilde u_K=0$, return any normalized product state $\widehat u_K$. The Euclidean error bound and the chosen value of $K$ give
\begin{equation}
1=\|u-\widetilde u_K\|_2^2\leq(1+\alpha)e_D(u)^2.
\label{eq:mps-relative-zero-output}
\end{equation}
Since $1-|\langle u|\widehat u_K\rangle|^2\leq1$, the relative infidelity guarantee follows.
\end{proof}

\subsection{Proper optimization against a bounded-bond target}

\begin{lemma}
\label{lem:mps-sequential-isometric-form}
Every normalized state $\psi\in\mps_D$ admits a representation with bond dimensions $r_i:=\min\{D,d^i,d^{n-i}\}, $ where $ r_0=r_n=1,$
and local isometries
\begin{equation} A_i:=\sum_{s=1}^d |s\rangle\otimes A_i^s \,\,:\,\,\mathbb C^{r_{i-1}}\longrightarrow\mathcal H_i\otimes\mathbb C^{r_i},\qquad A_i^\dagger A_i=I_{r_{i-1}}. \label{eq:mps-candidate-isometries} \end{equation}
The same statement holds for a normalized target of bond at most $K$, with $k_i:=\min\{K,d^i,d^{n-i}\}$ in place of $r_i$.
\end{lemma}

\begin{proof}
Let $q_i$ be the actual Schmidt rank of $\psi$ across cut $i$ with $q_0=q_n=1$. Choosing orthonormal bases of the successive right Schmidt supports gives local isometries
\begin{equation}
\widetilde A_i:\mathbb C^{q_{i-1}}\longrightarrow\mathcal H_i\otimes\mathbb C^{q_i}
\label{eq:mps-unpadded-sequential-isometries}
\end{equation}
whose contraction is $\psi$. Since $q_i\leq r_i$, embed each $\mathbb C^{q_i}$ into the first $q_i$ coordinates of $\mathbb C^{r_i}$. This specifies the action of $A_i$ on the embedded incoming support. Extend its images to $r_{i-1}$ orthonormal vectors in $\mathcal H_i\otimes\mathbb C^{r_i}$, which is possible because $r_{i-1}\leq d\,r_i$. The resulting $A_i$ is an isometry. Starting from the one-dimensional left boundary, the contraction remains in the embedded Schmidt supports, so the added columns do not change $\psi$. Replacing $D$ by $K$ proves the lemma.
\end{proof}

Let $v$ be a normalized target MPS of bond at most $K$, and let $\psi$ be a normalized candidate MPS of bond at most $D$, in the form of Lemma~\ref{lem:mps-sequential-isometric-form}. For $\alpha=1,\ldots,k_i$, define
\begin{equation} |v_\alpha^{(i)}\rangle:=\sum_{s_1,\ldots,s_i}(V_i^{s_i}\cdots V_1^{s_1})_{\alpha,1}|s_1\cdots s_i\rangle, \label{eq:mps-target-prefix-vectors} \end{equation}
and, for $\beta=1,\ldots,r_i$, define
\begin{equation} |\psi_\beta^{(i)}\rangle:=\sum_{s_1,\ldots,s_i}(A_i^{s_i}\cdots A_1^{s_1})_{\beta,1}|s_1\cdots s_i\rangle. \label{eq:mps-candidate-prefix-vectors} \end{equation}
The cross environment after site $i$ is the matrix $X_i\in\mathbb C^{k_i\times r_i}$ defined by
\begin{equation} (X_i)_{\alpha,\beta}:=\langle\psi_\beta^{(i)}|v_\alpha^{(i)}\rangle. \label{eq:mps-cross-environment} \end{equation}
One has $X_0=1$ and
\begin{equation} X_i=\sum_{s=1}^d V_i^s X_{i-1}(A_i^s)^\dagger. \label{eq:mps-cross-environment-update} \end{equation}
At the final cut, $X_n$ is a scalar and
$X_n=\langle\psi|v\rangle.$ 
For fixed target isometry $V$ and candidate isometry $A$, define
\begin{equation} \Phi^A(X):=\sum_{s=1}^d V^s X(A^s)^\dagger. \label{eq:mps-local-environment-map} \end{equation}
Intuitively, $\Phi^A$ is the one-site cross-environment update. It takes the overlap information accumulated up to the previous cut and propagates it through the next target and candidate tensors.
\begin{lemma}
\label{lem:mps-environment-stability}
For all matrices $X$ and $Y$ of the appropriate dimensions,
\begin{equation} \|\Phi^A(X)-\Phi^A(Y)\|_1\leq\|X-Y\|_1. \label{eq:mps-environment-contraction} \end{equation}
If $A'$ is another isometry with the same input and output spaces, then
\begin{equation} \|\Phi^A(X)-\Phi^{A'}(X)\|_1\leq\|A-A'\|_\infty\|X\|_1. \label{eq:mps-environment-local-lipschitz} \end{equation}
\end{lemma}

\begin{proof}
It is enough to prove $\|\Phi^A(W)\|_1\leq\|W\|_1 $
for every $W$. Let $Z$ be any matrix with $\|Z\|_\infty\leq 1$. 

Since
\begin{equation}
A=\sum_{s=1}^d |s\rangle\otimes A^s,\qquad V=\sum_{s=1}^d |s\rangle\otimes V^s,
\end{equation}
we have
\begin{align}
\tr\!\left(Z^\dagger\Phi^A(W)\right) &=\sum_{s=1}^d\tr\!\left(Z^\dagger V^s W(A^s)^\dagger\right) \nonumber\\ &=\sum_{s=1}^d\tr\!\left((A^s)^\dagger Z^\dagger V^s W\right) \nonumber\\
&=\tr\!\left(A^\dagger(I_d\otimes Z^\dagger)VW\right).
\label{eq:mps-environment-duality-identity}
\end{align}

\noindent By Schatten norm duality and submultiplicativity,
\begin{align}
\left|\tr\!\left(A^\dagger(I_d\otimes Z^\dagger)VW\right) \right| &\leq \left\|A^\dagger(I_d\otimes Z^\dagger)V\right\|_\infty
\|W\|_1 \nonumber\\&\leq\|A\|_\infty \|I_d\otimes Z^\dagger\|_\infty\|V\|_\infty\|W\|_1 \nonumber\\ &\leq \|W\|_1,
\label{eq:mps-environment-duality-bound}
\end{align}
The same duality calculation then gives
\begin{equation} \left|\operatorname{tr}\!\left[Z^\dagger\left(\Phi^A(X)-\Phi^{A'}(X)\right)\right]\right|\leq\|A-A'\|_\infty\|X\|_1. \label{eq:mps-environment-local-duality-bound} \end{equation}
Taking the supremum over $Z$ proves \eqref{eq:mps-environment-local-lipschitz}.
\end{proof}

\begin{theorem}
\label{thm:mps-bounded-target-optimizer}
\label{bond-k-thm}
Let $v$ be an explicitly given MPS of bond dimension at most $K$ and Euclidean norm at most one. For every $D\geq 1$ and every $\eta\in(0,1)$, there is a deterministic algorithm which outputs a normalized state $\widehat\psi\in\mps_D$ such that
\begin{equation} |\langle v|\widehat\psi\rangle|\geq\max_{\psi\in\mps_D}|\langle v|\psi\rangle|-\eta. \label{eq:mps-bounded-target-guarantee} \end{equation}
The runtime is bounded by
\begin{equation} n\left(\frac{Cn\sqrt D}{\eta}\right)^{O(KD+dD^2)} \label{eq:mps-bounded-target-runtime} \end{equation}
for a universal constant $C$.
\end{theorem}

Algorithm~\ref{alg:mps-bounded-target-dp} implements the
bounded-target optimizer.

\begin{algorithm}[htbp]
\caption{Proper MPS optimization against a bounded-bond target}
\label{alg:mps-bounded-target-dp}
\small
\begin{algorithmic}[1]
\Require An explicit MPS $v$ of bond at most $K$
         with $\|v\|_2\leq1$;
         a candidate bond $D\geq1$; $\eta\in(0,1)$.
\Ensure A normalized $\widehat\psi\in\mps_D$ satisfying Theorem \ref{bond-k-thm}.

\If{$\|v\|_2\leq\eta$}
    \State \Return a fixed normalized product state.
\EndIf

\State Put $w:=v/\|v\|_2$ in sequential isometric form with local tensors $V_1,\ldots,V_n$.

\State $h\gets\eta/(4n)$; $q\gets\eta/(8n)$.

\State For each site $i$, construct an operator-norm $h$-net $\mathcal N_i^{\mathrm{loc}}$ of the candidate isometries.

\State For each cut $i<n$, construct a Frobenius-norm $q/\sqrt D$-net $\mathcal N_i^{\mathrm{env}}$ of the Frobenius unit ball in $\mathbb C^{k_i\times r_i}$.

\State Initialize $\mathcal R_0$ with the empty tensor prefix and exact environment $X_0=1$.

\For{$i=1,\ldots,n-1$}
    \State Initialize an empty table $\mathcal R_i$, indexed by $\mathcal N_i^{\mathrm{env}}$.

    \ForAll{retained pairs $(\mathbf A,X)$ in $\mathcal R_{i-1}$ and tensors $A_i\in\mathcal N_i^{\mathrm{loc}}$}

        \State $Y\gets\sum_{s=1}^{d}V_i^sX(A_i^s)^\dagger$.

        \State Choose a nearest point
               $Z\in\mathcal N_i^{\mathrm{env}}$ to $Y$ in Frobenius norm.

        \If{cell $Z$ of $\mathcal R_i$ is empty}
            \State Store $((\mathbf A,A_i),Y)$ in cell $Z$.
        \EndIf
    \EndFor
\EndFor

\State Extend every pair $(\mathbf A,X)$ in $\mathcal R_{n-1}$ by every $A_n\in\mathcal N_n^{\mathrm{loc}}$.

\State For each extension, compute $X_n:=\sum_{s=1}^{d}V_n^sX(A_n^s)^\dagger$.

\State \Return the MPS defined by the extension with largest $|X_n|$.
\end{algorithmic}
\end{algorithm}

\begin{proof}
If $\|v\|_2\leq\eta$, return any fixed product state. Since $\max_{\psi\in\mps_D}|\langle v|\psi\rangle|\leq\|v\|_2\leq\eta,$
the statement is immediate. Otherwise, set $w:=\frac{v}{\|v\|_2}.$ An additive $\eta$ optimizer for $w$ is also an additive $\eta$ optimizer for $v$ since $\|v\|_2\leq 1$. We may therefore assume that $v$ is normalized. Put $v$ in the form of Lemma~\ref{lem:mps-sequential-isometric-form}. We search over candidate MPSs in the same padded form. For every site $i$, choose an operator-norm $h$-net $\mathcal N_i^{\mathrm{loc}}$ of the candidate isometries, where
\begin{equation} |\mathcal N_i^{\mathrm{loc}}|\leq\left(\frac{C\sqrt D}{h}\right)^{O(dD^2)}. \label{eq:mps-local-net-size} \end{equation}
For every cut $1\leq i\leq n-1$, choose a Frobenius-norm $q/\sqrt D$-net $\mathcal N_i^{\mathrm{env}}$ of the Frobenius unit ball in $\mathbb C^{k_i\times r_i}$. By Lemma~\ref{lem:mps-environment-stability}, $\|X_i\|_1\leq 1$ for every exact environment $X_i$. Hence $\|X_i\|_2\leq 1$, and since $X_i$ has at most $D$ columns,
\begin{equation} \|X-Y\|_1\leq\sqrt D\|X-Y\|_2. \label{eq:mps-trace-frobenius-comparison} \end{equation}
Thus $\mathcal N_i^{\mathrm{env}}$ is a trace-norm $q$-net for all possible environments and
\begin{equation} |\mathcal N_i^{\mathrm{env}}|\leq\left(\frac{C\sqrt D}{q}\right)^{O(KD)}. \label{eq:mps-environment-net-size} \end{equation}

The dynamic program begins at cut $0$ with the empty prefix and $X_0=1$. Suppose retained prefixes have been constructed through cut $i-1$, where $1\leq i\leq n-1$. For every retained prefix and every $A_i\in\mathcal N_i^{\mathrm{loc}}$, append $A_i$ and compute the exact next environment using \eqref{eq:mps-cross-environment-update}. Assign the resulting environment to a nearest point of $\mathcal N_i^{\mathrm{env}}$, and retain one exact prefix in every occupied cell. Two exact environments in the same cell differ by at most $2q$ in trace norm.

At the final site, for every retained prefix and every $A_n\in\mathcal N_n^{\mathrm{loc}}$, compute the exact scalar overlap $X_n$ and return the normalized candidate with largest absolute overlap. Every candidate is normalized because all local tensors are isometries and $r_0=r_n=1$.

Let $\psi^\star$ be an optimal bond-$D$ MPS, with local isometries $A_1^\star,\ldots,A_n^\star$ and exact environments $X_0^\star,\ldots,X_n^\star$. At every site choose a local net tensor $\widetilde A_i$ satisfying $\|\widetilde A_i-A_i^\star\|_\infty\leq h. $ Follow this path through the dynamic program. Suppose that after cut $i-1$ a retained prefix has exact environment $\widehat X_{i-1}$ and define
\begin{equation} e_{i-1}:=\|\widehat X_{i-1}-X_{i-1}^\star\|_1. \label{eq:mps-environment-error} \end{equation}
Before cell merging at cut $i$, the candidate environment $Y_i:=\Phi^{\widetilde A_i}(\widehat X_{i-1})$ satisfies, by Lemma~\ref{lem:mps-environment-stability},
\begin{equation} \|Y_i-X_i^\star\|_1\leq e_{i-1}+h. \label{eq:mps-environment-error-before-merging} \end{equation}
After cell merging, $e_i\leq e_{i-1}+h+2q.$ Since $e_0=0$, induction gives $e_{n-1}\leq(n-1)(h+2q).$ At the last site there is no merging, so the final overlap error is at most
\begin{equation} |\widehat X_n-X_n^\star|\leq nh+2(n-1)q. \label{eq:mps-final-overlap-error} \end{equation}
Choose
\begin{equation} h:=\frac{\eta}{4n},\qquad q:=\frac{\eta}{8n}. \label{eq:mps-net-parameters} \end{equation}
Then the right-hand side of \eqref{eq:mps-final-overlap-error} is at most $\eta/2$. By the reverse triangle inequality, the surviving path has absolute overlap at least
\begin{equation} |\langle v|\psi^\star\rangle|-\frac{\eta}{2}. \label{eq:mps-surviving-path-overlap} \end{equation}
The algorithm returns the candidate with largest absolute overlap, so \eqref{eq:mps-bounded-target-guarantee} follows. At every cut the number of retained prefixes is at most \eqref{eq:mps-environment-net-size}, and each is extended using at most \eqref{eq:mps-local-net-size} local tensors. Substituting \eqref{eq:mps-net-parameters} proves \eqref{eq:mps-bounded-target-runtime}.
\end{proof}

\subsection{Proper optimization against an arbitrary explicit target}

\begin{theorem}
\label{thm:mps-arbitrary-target-optimizer}
\label{proper-compress-thm-mps}
Let $u$ be a normalized explicit MPS of arbitrary bond dimension $B$. For every target bond dimension $D\geq 1$ and every $\eta\in(0,1)$, there is a deterministic algorithm which outputs a normalized state $\widehat\psi\in\mps_D$ such that
\begin{equation} |\langle u|\widehat\psi\rangle|^2\geq\max_{\psi\in\mps_D}|\langle u|\psi\rangle|^2-\eta. \label{eq:mps-arbitrary-target-guarantee} \end{equation}
Its runtime is bounded by
\begin{equation} n\left(\frac{Cn\sqrt D}{\eta}\right)^{O(D^2/\eta^2+dD^2)}\poly(d,B,D,1/\eta) \label{eq:mps-arbitrary-target-runtime} \end{equation}
for a universal constant $C$.
\end{theorem}

\begin{proof}
Choose $K+1\geq\frac{64D}{\eta^2}.$ By Theorem~\ref{thm:mps-comparator-dual-compression}, construct an MPS $\widetilde u_K$ of bond at most $K$ such that
\begin{equation} \sup_{\psi\in\mps_D}|\langle\psi|u-\widetilde u_K\rangle|\leq\sqrt{\frac{D}{K+1}}\leq\frac{\eta}{8}, \label{eq:mps-arbitrary-target-compression-error} \end{equation}
and $\|\widetilde u_K\|_2\leq 1$. Define $ \mathcal A_D(x):=\max_{\psi\in\mps_D}|\langle x|\psi\rangle|.$ Then
\begin{equation} |\mathcal A_D(u)-\mathcal A_D(\widetilde u_K)|\leq\frac{\eta}{8}. \label{eq:mps-overlap-functional-stability} \end{equation}
Apply Theorem~\ref{thm:mps-bounded-target-optimizer} to $\widetilde u_K$ with amplitude error $\eta/4$. This gives $\widehat\psi\in\mps_D$ satisfying
\begin{equation} |\langle\widetilde u_K|\widehat\psi\rangle|\geq\mathcal A_D(\widetilde u_K)-\frac{\eta}{4}. \label{eq:mps-compressed-target-optimization} \end{equation}
Therefore
\begin{align} |\langle u|\widehat\psi\rangle|&\geq|\langle\widetilde u_K|\widehat\psi\rangle|-|\langle u-\widetilde u_K|\widehat\psi\rangle| \nonumber\\ &\geq\mathcal A_D(u)-\frac{\eta}{2}. \label{eq:mps-arbitrary-target-amplitude-bound} \end{align}
If $a,b\in[0,1]$ and $b\geq a-\eta/2$, then $a^2-b^2=(a-b)(a+b)\leq\eta.$ Taking $a=\mathcal A_D(u)$ and $b=|\langle u|\widehat\psi\rangle|$ proves \eqref{eq:mps-arbitrary-target-guarantee}. The runtime follows from $K=O(D/\eta^2)$, Theorem~\ref{thm:mps-comparator-dual-compression}, and Theorem~\ref{thm:mps-bounded-target-optimizer}.
\end{proof}

\subsection{Proper agnostic learning of MPSs}

\begin{theorem}[Proper agnostic learning of MPSs]
\label{thm:proper-agnostic-mps}
\label{full-proper-agnostic-mps}
For every $n,d,D\in\mathbb N$ and every $\varepsilon,\delta\in(0,1/4)$, there is a randomized quantum algorithm which, given copies of an arbitrary density operator $\rho$ on $(\mathbb C^d)^{\otimes n}$, outputs a classical description of a normalized state $\widehat\psi\in\mps_D $ such that, with probability at least $1-\delta$,
\begin{equation} \langle\widehat\psi|\rho|\widehat\psi\rangle\geq\opt_D(\rho)-\varepsilon. \label{eq:mps-proper-agnostic-guarantee} \end{equation}
The copy complexity is bounded by $\poly\!\left(n,d,D,1/\varepsilon,\log(1/\delta)\right).$
The runtime is bounded by $\poly\!\left(n,d,D,1/\varepsilon,\log(1/\delta)\right)
\left(\frac{Cn\sqrt D}{\varepsilon}\right)^{
O(D^2/\varepsilon^2+dD^2)}.$
In particular, the runtime is polynomial in $n$ for
fixed $d,D,\varepsilon$.
\end{theorem}

\begin{proof}
Take
\begin{equation}
\mathcal C:=\mps_D,\qquad \mathcal V_R:=\left\{u\in\mathcal H: \operatorname{rank}_{\mathcal H_{1:i}\,|\,\mathcal H_{i+1:n}}(u)\leq R \ \text{for every }1\leq i\leq n-1 \right\}.
\label{eq:mps-properization-pair}
\end{equation}
By Lemma~\ref{lem:mps-cut-rank-characterization}, $\mathcal V_R$ is precisely the class of possibly unnormalized open-boundary MPSs of bond dimension at most $R$. In particular, $\mathcal A_R=\mps_R$, and hence $\mathcal C=\mathcal A_D$, so the embedding condition holds with
$R'=D$. Lemma~\ref{lem:mps-linear-combinations} verifies the linear combination property in Definition~\ref{def:properization-data} with
\begin{equation} \Lambda(R_1,\ldots,R_q):=\sum_{a=1}^q R_a. \label{eq:mps-lambda-function} \end{equation}
Theorem~\ref{thm:mps-improper-learner} supplies the improper ambient learner, Theorem~\ref{thm:mps-arbitrary-target-optimizer} supplies the proper pure-target optimizer, and Lemma~\ref{lem:mps-explicit-operations} supplies the remaining explicit operations. $|0\rangle^{\otimes n}$ belongs to $\mathcal C$. Therefore every hypothesis of Theorem~\ref{thm:general-improper-to-proper} is verified, giving \eqref{eq:mps-proper-agnostic-guarantee}.

It remains to bound the complexity. Every improper-learning call uses the fixed comparator bond $D$ and error $\gamma=\theta/16=\Theta(\varepsilon^2)$. Thus every learner output has bond dimension at most
\begin{equation*}
B_0:=dn^2p(D,16/\theta) =\poly(n,d,D,1/\varepsilon).
\end{equation*}
There are $m=O(\varepsilon^{-2})$ accepted outputs $b_0,\ldots,b_{m-1}$, and $\operatorname{span}\{s_1,\ldots,s_m\}= \operatorname{span}\{b_0,\ldots,b_{m-1}\}.$ Store every orthonormal direction as a linear combination of these original outputs. This avoids repeatedly expanding linear combinations of previously expanded directions. Lemma~\ref{lem:mps-linear-combinations} then gives bond at most $mB_0$ for every $s_a$ and every pure target used in the coefficient search. The residual estimates and the improper-learning calls use polynomially many copies. Estimating the $m\times m$ support matrix also has polynomial cost since 
entrywise accuracy $O(\varepsilon/m)$ suffices for operator-norm accuracy $O(\varepsilon)$. Together with Lemma~\ref{lem:mps-explicit-operations} and the bounded postselection overhead, this proves the stated copy bound. All subsequent optimization is classical. The spectral cutoff leaves $O(\varepsilon^{-1})$ eigenvectors, so there are $(C/\varepsilon)^{O(1/\varepsilon)}$ pure-target calls, each with error $\Theta(\varepsilon)$ and input bond at most $mB_0$. Applying Theorem~\ref{thm:mps-arbitrary-target-optimizer} and absorbing the coefficient-net factor into its running-time bound proves the theorem.
\end{proof}
\begin{corollary}
\label{cor:mps-optimal-scores}
Given copies of $\rho$, a maximum bond dimension $D_{\max}\geq1$, and $\varepsilon,\delta\in(0,1/4)$, one can use $\poly\!\left(n,d,D_{\max},1/\varepsilon,\log(1/\delta)\right)$
copies to construct a single classical data set with the following property. With probability at least $1-\delta$, for every $1\leq D\leq D_{\max}$, classical processing of this same data set returns $\widehat\psi_D\in\mps_D$ and $\widehat v_D\in[0,1]$ such that
\begin{equation*}
\langle\widehat\psi_D|\rho|\widehat\psi_D\rangle
\geq\opt_D(\rho)-\varepsilon,\qquad|\widehat v_D-\opt_D(\rho)|\leq\varepsilon.
\end{equation*}
\end{corollary}

\begin{proof}
Construct the finite core using the fixed comparator bond $D_{\max}$, and use the parameters in \eqref{eq:general-main-parameters}. The resulting truncated estimated core obeys
\begin{equation*}
\sup_{\psi\in\mps_{D_{\max}}}\left|\langle\psi|\rho|\psi\rangle
-\langle\psi|\widehat A_\tau|\psi\rangle\right|\leq a,
\qquad a:=2\sqrt\theta+\xi+\tau.
\end{equation*}
This same bound holds for every smaller bond dimension. For each $D$, optimize the explicit objective $\widehat A_\tau$ over $\mps_D$ to additive error $\zeta$. Then
\begin{equation*}
\langle\widehat\psi_D|\widehat A_\tau|\widehat\psi_D\rangle-a \leq\opt_D(\rho) \leq \langle\widehat\psi_D|\widehat A_\tau|\widehat\psi_D\rangle+a+\zeta.
\end{equation*}
Let $\widehat v_D$ be $\langle\widehat\psi_D|\widehat A_\tau|\widehat\psi_D\rangle+\zeta/2$ clipped to $[0,1]$. Its error is at most $a+\zeta/2=13\varepsilon/32$. The returned state's score is within $2a+\zeta=13\varepsilon/16$ of $\opt_D(\rho)$. All statements hold on the same successful event.
\end{proof}

\noindent The score estimate also gives a tolerant test for the desired thresholds $t_-<t_+$: choosing the estimation error below $(t_+-t_-)/2$ distinguishes $\opt_D(\rho)\geq t_+$ from $\opt_D(\rho)\leq t_-$ under this promise.

\begin{corollary}[Improper learning of periodic-boundary MPSs]
\label{cor:periodic-mps-improper-learning}
Let $\mps_D^{\mathrm{per}}$ denote the normalized periodic-boundary MPSs of bond dimension at most $D$. Given copies of an arbitrary density operator $\rho$ and $\varepsilon,\delta\in(0,1/4)$, one can output a normalized open-boundary MPS $\widehat\psi\in\mps_{D^2}$ such that,
with probability at least $1-\delta$,
\begin{equation*}
\langle\widehat\psi|\rho|\widehat\psi\rangle\geq\sup_{\phi\in\mps_D^{\mathrm{per}}}\langle\phi|\rho|\phi\rangle-\varepsilon.
\end{equation*}
The complexity is bounded by Theorem~\ref{thm:proper-agnostic-mps} with $D$ replaced by $D^2$.
\end{corollary}

\begin{proof}
Separating a contiguous prefix from a periodic MPS cuts at most two virtual bonds. Its Schmidt rank across every prefix cut is therefore at most $D^2$. Lemma~\ref{lem:mps-cut-rank-characterization} gives $\mps_D^{\mathrm{per}}\subseteq\mps_{D^2}$. Apply Theorem~\ref{thm:proper-agnostic-mps} to the open-boundary comparator class $\mps_{D^2}$.
\end{proof}
\section{Matrix Product Operators}
\label{sec:matrix-product-operators}

Matrix product operators are matrix product states on the local operator space. Fix orthonormal bases of $\mathcal H_i\cong\mathbb C^d$, set $\mathcal K_i:=\mathcal H_i\otimes\overline{\mathcal H_i}$ and $\mathcal K:=\bigotimes_{i=1}^n\mathcal K_i$, and define
\begin{equation} \operatorname{vec}\left(|s_1\cdots s_n\rangle\langle t_1\cdots t_n|\right):=\bigotimes_{i=1}^n\left(|s_i\rangle\otimes|\overline{t_i}\rangle\right). \label{eq:mpo-vectorization} \end{equation}
Then
\begin{equation} \langle\operatorname{vec}(X),\operatorname{vec}(Y)\rangle=\operatorname{tr}(X^\dagger Y),\qquad \|\operatorname{vec}(X)\|_2=\|X\|_{\mathrm{HS}}, \label{eq:mpo-vectorization-isometry} \end{equation}
and $X$ is an MPO of bond dimension at most $D$ if and only if $\operatorname{vec}(X)$ is an MPS of bond dimension at most $D$ on the chain $\mathcal K_1\otimes\cdots\otimes\mathcal K_n$. Thus the Hilbert--Schmidt-normalized MPO class $\mpo_D$ is identified with $\mps_D$ at local dimension $d^2$. For $X\in\mathrm L(\mathcal H)$, write
\begin{equation} \|X\|_{\mpo(D)}:=\sup_{Y\in\mpo_D}|\operatorname{tr}(Y^\dagger X)|. \label{eq:mpo-comparator-dual-norm} \end{equation}

\begin{corollary}[Proper agnostic learning of MPOs]
\label{cor:mpo-consequences}
The results of Appendix~\ref{sec:matrix-product-states} hold for MPOs with $d\mapsto d^2$. In particular,

\emph{(i) Comparator-dual compression.} For every $X\in\mathrm L(\mathcal H)$ and $D,K\geq1$, there is an MPO $\widetilde X_K$ of bond dimension at most $K$ such that
\begin{equation} \|X-\widetilde X_K\|_{\mpo(D)}\leq\|X\|_{\mathrm{HS}}\sqrt{\frac{D}{K+1}},\qquad \|\widetilde X_K\|_{\mathrm{HS}}\leq\|X\|_{\mathrm{HS}}. \label{eq:mpo-compression-guarantee} \end{equation}
If $X$ is an explicit MPO of bond $B$, the construction uses $\poly(n,d^2,B,K)$ arithmetic operations.

\emph{(ii) Proper pure-target optimization.} If $X$ is a Hilbert--Schmidt-normalized explicit MPO of bond $B$, then for every $D\geq1$ and $\eta\in(0,1)$ one can output $\widehat Y\in\mpo_D$ such that
\begin{equation} |\operatorname{tr}(X^\dagger\widehat Y)|^2\geq\max_{Y\in\mpo_D}|\operatorname{tr}(X^\dagger Y)|^2-\eta, \label{eq:mpo-arbitrary-target-guarantee} \end{equation}
in time
\begin{equation} n\left(\frac{Cn\sqrt D}{\eta}\right)^{O(D^2/\eta^2+d^2D^2)}\poly(d,B,D,1/\eta). \label{eq:mpo-arbitrary-target-runtime} \end{equation}

\emph{(iii) Proper agnostic learning.} For a density operator $\Omega$ on $\mathcal K$, define
\begin{equation} \opt_D^{\mpo}(\Omega):=\max_{X\in\mpo_D}\langle\operatorname{vec}(X),\Omega\operatorname{vec}(X)\rangle. \label{eq:mpo-agnostic-optimum} \end{equation}
Given copies of $\Omega$ and $\varepsilon,\delta\in(0,1/4)$, there is a randomized quantum algorithm which outputs $\widehat X\in\mpo_D$ such that, with probability at least $1-\delta$,
\begin{equation} \langle\operatorname{vec}(\widehat X),\Omega\operatorname{vec}(\widehat X)\rangle\geq\opt_D^{\mpo}(\Omega)-\varepsilon. \label{eq:mpo-proper-agnostic-guarantee} \end{equation}
The copy complexity is $\poly(n,d^2,D,1/\varepsilon,\log(1/\delta))$.
The runtime is bounded by \\$\poly\!\left(n,d^2,D,1/\varepsilon,\log(1/\delta)\right)\left(\frac{Cn\sqrt D}{\varepsilon}\right)^{O(D^2/\varepsilon^2+d^2D^2)}.$
\end{corollary}

\begin{proof}
Apply Theorems~\ref{thm:mps-comparator-dual-compression}, \ref{thm:mps-arbitrary-target-optimizer}, and \ref{thm:proper-agnostic-mps} to the vectorized chain whose local dimension is $d^2$.
\end{proof}
\noindent Corollary~\ref{cor:mps-relative-error-compression} also transfers under vectorization. In particular, for $K\geq D$,
\begin{equation*}
\|X-\widetilde X_K\|_{\mathrm{HS}}^2\leq\frac{K+1}{K+1-D}\min_{\substack{a\in\mathbb C\\Y\in\mpo_D}}\|X-aY\|_{\mathrm{HS}}^2.
\end{equation*}
These statements concern Hilbert--Schmidt approximation. They do not impose positivity, trace normalization, or channel constraints on the output.
\section{Tree Tensor Networks}
\label{sec:tree-tensor-networks}

Let $T=(V,E)$ be a tree with $|V|=n$. Fix a root $r\in V$, let $\operatorname{ch}(v)$ denote the children of $v$, and define
\begin{equation}
\Delta_T:=\max_{v\in V}\deg_T(v).
\label{eq:ttn-maximum-degree}
\end{equation}
At each vertex $v$, let $\mathcal H_v\cong\mathbb C^d$, and write
\begin{equation*}
\mathcal H_T:=\bigotimes_{v\in V}\mathcal H_v.
\end{equation*}
A tree tensor network on $T$ assigns a virtual space $\mathcal B_e$ to each edge $e\in E$ and a local tensor to each vertex; contracting all virtual indices produces a vector in $\mathcal H_T$. Let $\ttn_{T,D}$ denote the normalized vectors admitting such a representation with $\dim\mathcal B_e\leq D$ for every $e\in E$. For $u\in\mathcal H_T$ and $e\in E$, let $\sr_e(u)$ denote the Schmidt rank of $u$ across the bipartition obtained by deleting $e$.

\begin{lemma}
\label{lem:ttn-edge-rank-characterization}
A nonzero vector $u\in\mathcal H_T$ admits a TTN representation of bond dimension at most $D$ if and only if $\sr_e(u)\leq D$ for all $e\in E$.
\end{lemma}

\begin{proof}
Standard.
\end{proof}

\noindent In particular, $\ttn_{T,D}$ is compact.

\subsection{An MPS ordering adapted to the tree}

The ambient improper learner is the MPS learner from Appendix~\ref{sec:matrix-product-states}. We first choose an ordering in which bounded-degree TTNs have polynomial MPS bond.

\begin{lemma}
\label{lem:ttn-heavy-last-order}
There is an ordering $\pi=(v_1,\ldots,v_n)$ of the vertices such that every rooted subtree is a contiguous interval and every prefix cut crosses at most $w_T:=\Delta_T\left(1+\lceil\log_2 n\rceil\right)$ tree edges.
\end{lemma}

\begin{proof}
For every vertex $v$, let $n_v:=|T_v|$ and choose a child of maximum subtree size as the ``heavy child.'' Traverse the tree by visiting $v$ first, then each non-heavy child subtree, and finally the heavy-child subtree. This is a preorder traversal, so every rooted subtree is contiguous.

Fix a prefix. A crossing edge originates at an active vertex on the current recursion stack. The current vertex contributes at most $\Delta_T$ edges. An active proper ancestor can contribute an unvisited child only when the recursion descended through a non-heavy child. If $u$ is a non-heavy child of $v$, then
\begin{equation} n_u\leq\frac{n_v-1}{2}<\frac{n_v}{2}. \label{eq:ttn-light-child-halving} \end{equation}
Since each subtree must have size at least $1$, there are at most $\lceil\log_2 n\rceil$ contributing proper ancestors, each contributing at most $\Delta_T$ crossing edges. The current vertex $v$ also contributes $\Delta_T$ itself.
\end{proof}

\begin{definition}
    Let $\mps_R^\pi$ denote open-boundary MPSs whose sites are ordered by $\pi$.
\end{definition}

\begin{lemma}
\label{lem:ttn-mps-comparison}
Let $\pi$ be the ordering of Lemma~\ref{lem:ttn-heavy-last-order}.
\begin{enumerate}
    \item Every state in $\ttn_{T,D}$ belongs to $\mps_{D^{w_T}}^\pi$.
    \item Every state in $\mps_B^\pi$ has a TTN representation on $T$ of bond dimension at most $B^2$.
    \item Both conversions can be carried out in time polynomial in the input and output representation size.
\end{enumerate}
\end{lemma}

\begin{proof}
Across a prefix cut of $\pi$, at most $w_T$ tree edges are cut. A bond-$D$ TTN therefore has Schmidt rank at most $D^{w_T}$ across every prefix cut, and so Lemma~\ref{lem:mps-cut-rank-characterization} gives the first claim. Conversely, every rooted subtree is an interval in $\pi$. Separating an interval from an open-boundary MPS cuts at most two MPS virtual bonds, so every tree-edge Schmidt rank is at most $B^2$. The second claim follows from Lemma~\ref{lem:ttn-edge-rank-characterization}. The conversions use tensor contractions, QR factorizations, and Schmidt decompositions with dimensions controlled by the input and output virtual spaces.
\end{proof}

\subsection{Comparator-dual compression on a tree}

For any tree $\mathcal T$ whose vertices may include auxiliary boundary spaces, define
\begin{equation} \|x\|_{\mathcal T,D}:=\sup\left\{|\langle z,x\rangle|:\|z\|_2\leq 1,\ \sr_e(z)\leq D\ \forall e\in E(\mathcal T)\right\}. \label{eq:ttn-comparator-dual-norm} \end{equation}
For the physical tree $T$, Lemma~\ref{lem:ttn-edge-rank-characterization} gives
\begin{equation} \|x\|_{T,D}=\sup_{\psi\in\ttn_{T,D}}|\langle\psi,x\rangle|. \label{eq:ttn-dual-norm-physical-interpretation} \end{equation}

\begin{lemma}
\label{lem:ttn-comparator-norm-properties}
The tree comparator norm satisfies the following properties.
\begin{enumerate}[label=(\roman*)]
    \item Let $\mathcal U$ be a connected pendant part of $\mathcal T$, attached through one edge, and let $\mathcal T'$ be obtained by replacing $\mathcal U$ by a boundary space $\mathcal B$. If $C:\mathcal B\to\mathcal H_{\mathcal U}$ is a contraction, then
    \begin{equation}
    \|(C\otimes I)x\|_{\mathcal T,D}\leq\|x\|_{\mathcal T',D}.
    \label{eq:ttn-pendant-contraction-bound}
    \end{equation}

    \item Suppose a boundary space decomposes as $\mathcal B=\mathcal B_1\oplus\mathcal B_2$, with $x_a\in\mathcal B_a\otimes\mathcal R$. Then
    \begin{equation}
    \|x_1\oplus x_2\|_{\mathcal T,D}^2
    \leq
    \|x_1\|_{\mathcal T_1,D}^2
    +
    \|x_2\|_{\mathcal T_2,D}^2,
    \label{eq:ttn-boundary-direct-sum-bound}
    \end{equation}
    where $\mathcal T_a$ denotes the corresponding tree with boundary space $\mathcal B_a$.

    \item Let $\mathcal B$ be a leaf boundary space of $\mathcal T$ joined to the rest of the tree by a single edge, and let $\mathcal R$ be the Hilbert space of the remaining vertices. Regard $L\in\mathcal B\otimes\mathcal R$ as a matrix from $\mathcal R$ to $\mathcal B$. If every singular value of $L$ is at most $\sqrt{\Delta}$, then
    \begin{equation}
    \|L\|_{\mathcal T,D}^2\leq D\Delta.
    \label{eq:ttn-flat-loss-bound}
    \end{equation}
\end{enumerate}
\end{lemma}

\begin{proof}
For the first claim, pull an admissible comparator $z$ back to $z':=(C^\dagger\otimes I)z$. Its norm and every edge Schmidt rank can only decrease, so
\begin{equation} |\langle z,(C\otimes I)x\rangle|=|\langle z',x\rangle|\leq\|x\|_{\mathcal T',D}. \label{eq:ttn-pendant-contraction-proof} \end{equation}

For the second claim, write an admissible comparator as $z=z_1\oplus z_2$. Projection onto either sector cannot increase edge ranks and
\begin{equation} \|z_1\|_2^2+\|z_2\|_2^2\leq 1. \label{eq:ttn-projected-comparator-norms} \end{equation}
Cauchy--Schwarz then gives
\begin{equation} |\langle z,x_1\oplus x_2\rangle|\leq\left(\|x_1\|_{\mathcal T_1,D}^2+\|x_2\|_{\mathcal T_2,D}^2\right)^{1/2}. \label{eq:ttn-boundary-direct-sum-proof} \end{equation}

For the third claim, an admissible comparator is a matrix $Z$ of rank at most $D$ and Frobenius norm at most one across the boundary edge. Von Neumann's trace inequality and Cauchy--Schwarz give
\begin{equation} |\operatorname{tr}(Z^\dagger L)|\leq\left(\sum_{j=1}^D s_j(Z)^2\right)^{1/2}\left(\sum_{j=1}^D s_j(L)^2\right)^{1/2}\leq\sqrt{D\Delta}. \label{eq:ttn-flat-loss-proof} \end{equation}
Taking the appropriate suprema proves all three claims.
\end{proof}

\begin{theorem}
\label{thm:ttn-comparator-dual-compression}
For every $u\in\mathcal H_T$ and every $D,K\geq 1$, one can construct a TTN $\widetilde u_K$ on $T$ of bond dimension at most $K$ satisfying
\begin{equation} \|u-\widetilde u_K\|_{T,D}\leq\|u\|_2\sqrt{\frac{D}{K+1}}, \label{eq:ttn-compression-guarantee} \end{equation}
and $\|\widetilde u_K\|_2\leq\|u\|_2.$ If $u$ is an explicit TTN of bond dimension $B$, the construction uses $\poly(n,d,B^{\Delta_T},K^{\Delta_T})$ arithmetic operations.
\end{theorem}

\begin{proof}
Order the vertices from the leaves toward the root, so that every child appears before its parent, and write this order as $v_1,\ldots,v_n=r$. During the construction, each processed child subtree is replaced by one boundary space of dimension at most $K$. Immediately before processing $v_i$, let
\begin{equation} \mathcal U_i:=\mathcal H_{v_i}\otimes\bigotimes_{w\in\operatorname{ch}(v_i)}\mathcal B_w,\qquad T_i\in\mathcal U_i\otimes\mathcal R_i, \label{eq:ttn-compression-current-remainder} \end{equation}
where $\mathcal R_i$ is the rest of the current quotient tree and $T_1:=u$.

For $i<n$, regard $T_i$ as a matrix from $\mathcal R_i$ to $\mathcal U_i$ and apply the one-step SVD shrinkage from the proof of Theorem~\ref{thm:mps-comparator-dual-compression}. With the same notation, it produces a boundary space $\mathcal B_i$ of dimension at most $K$, a remainder $T_{i+1}$, a loss tensor $L_i$, maps $C_i^{\mathrm{keep}}$ and $C_i^{\mathrm{loss}}$, and a number $\Delta_i\geq 0$ such that
\begin{equation} T_i=C_i^{\mathrm{keep}}T_{i+1}+C_i^{\mathrm{loss}}L_i,\qquad C_iC_i^\dagger\preceq I, \label{eq:ttn-compression-one-step-properties} \end{equation}
where $C_i=[C_i^{\mathrm{keep}}\ C_i^{\mathrm{loss}}]$, every singular value of $L_i$ is at most $\sqrt{\Delta_i}$, and
\begin{equation} (K+1)\Delta_i\leq\|T_i\|_2^2-\|T_{i+1}\|_2^2. \label{eq:ttn-compression-energy-drop} \end{equation}
These are exactly the matrix identities used in \eqref{eq:mps-compression-one-step-decomposition}, \eqref{eq:mps-compression-map-contraction}, and \eqref{eq:mps-compression-energy-drop}; they do not depend on the underlying graph.

Keep the final root tensor $T_n$ and define the retained contraction recursively by
\begin{equation} \widetilde T_n:=T_n,\qquad \widetilde T_i:=(C_i^{\mathrm{keep}}\otimes I_{\mathcal R_i})\widetilde T_{i+1},\qquad \widetilde u_K:=\widetilde T_1. \label{eq:ttn-compression-retained-tensors} \end{equation}
Every retained boundary has dimension at most $K$, so $\widetilde u_K$ is a bond-$K$ TTN. Set
\begin{equation} E_i:=T_i-\widetilde T_i. \label{eq:ttn-compression-accumulated-error} \end{equation}
Then $E_n=0$ and
\begin{equation} E_i=(C_i\otimes I_{\mathcal R_i})(E_{i+1}\oplus L_i). \label{eq:ttn-compression-error-recursion} \end{equation}
Lemma~\ref{lem:ttn-comparator-norm-properties} gives
\begin{equation} \|E_i\|_{\mathcal T_i,D}^2\leq\|E_{i+1}\|_{\mathcal T_{i+1},D}^2+D\Delta_i. \label{eq:ttn-compression-error-step} \end{equation}
Therefore
\begin{equation} \|u-\widetilde u_K\|_{T,D}^2\leq D\sum_{i=1}^{n-1}\Delta_i. \label{eq:ttn-compression-error-telescope} \end{equation}
On the other hand, summing \eqref{eq:ttn-compression-energy-drop} yields
\begin{equation} (K+1)\sum_{i=1}^{n-1}\Delta_i\leq\|T_1\|_2^2-\|T_n\|_2^2\leq\|u\|_2^2. \label{eq:ttn-compression-delta-sum} \end{equation}
Combining the last two inequalities proves \eqref{eq:ttn-compression-guarantee}. The one-step construction also has $\|T_{i+1}\|_2\leq\|T_i\|_2$ and $\|C_i^{\mathrm{keep}}\|_\infty\leq 1$, which gives $\|\widetilde u_K\|_2\le \|u\|_2$. If $u$ has TTN bond $B$, its supplied local tensors contain at most $ndB^{\Delta_T}$ entries. During the sweep, the remaining virtual indices have dimensions at most $B$, while the retained boundary spaces have dimensions at most $K$. The required Gram matrices can be computed by contracting the corresponding doubled tree networks, without expanding the vector in the physical basis. These contractions and the resulting SVDs use $\poly(n,d,B^{\Delta_T},K^{\Delta_T})$ arithmetic operations.
\end{proof}
Define
\begin{equation*}
e_{T,D}(u)^2:=\min_{\substack{a\in\mathbb C\\\psi\in\ttn_{T,D}}}
\|u-a\psi\|_2^2=\|u\|_2^2-\|u\|_{T,D}^2.
\end{equation*}
\begin{corollary}
\label{cor:ttn-relative-error-compression}
For $K\geq D$, the construction of
Theorem~\ref{thm:ttn-comparator-dual-compression} satisfies
\begin{equation*}
\|u-\widetilde u_K\|_2^2\leq\frac{K+1}{K+1-D}\,e_{T,D}(u)^2,\qquad\|u-\widetilde u_K\|_{T,D}^2\leq\frac{D}{K+1-D}\,e_{T,D}(u)^2.
\end{equation*}
For $\alpha\in(0,1)$, choosing $K=D+\lceil D/\alpha\rceil$ gives squared Euclidean error at most $(1+\alpha)e_{T,D}(u)^2$. If $\|u\|_2=1$, return $\widehat u_K:=\widetilde u_K/\|\widetilde u_K\|_2$ when $\widetilde u_K\neq0$, and any normalized product state otherwise. Then
\begin{equation}
1-|\langle u|\widehat u_K\rangle|^2\leq(1+\alpha)\left(1-\max_{\psi\in\ttn_{T,D}}|\langle u|\psi\rangle|^2\right).
\label{eq:ttn-relative-infidelity}
\end{equation}
\end{corollary}

\begin{proof}
The proof is the same as that of Corollary \ref{cor:mps-relative-error-compression} and applies Lemma~\ref{lem:ttn-comparator-norm-properties}.
\end{proof}

\subsection{Proper optimization against a bounded-bond target}

For every non-root vertex $v$, let $T_v$ be its rooted subtree and set
\begin{equation} h_v:=\dim\mathcal H_{T_v},\qquad r_v:=\min\{D,h_v\},\qquad k_v:=\min\{K,h_v\}. \label{eq:ttn-padded-edge-dimensions} \end{equation}

\begin{lemma}
\label{lem:ttn-rooted-isometric-form}
Every normalized $\psi\in\ttn_{T,D}$ can be represented by a unit root tensor and, for every non-root vertex $v$, an isometry
\begin{equation} A_v:\mathbb C^{r_v}\longrightarrow\mathcal H_v\otimes\bigotimes_{w\in\operatorname{ch}(v)}\mathbb C^{r_w}. \label{eq:ttn-candidate-isometric-tensor} \end{equation}
Every vector $u$ of TTN bond at most $K$ and norm at most one has the same form with $K$ in place of $D$, non-root isometries $V_v$, and a root tensor of norm at most one.
\end{lemma}

\begin{proof}
For each non-root vertex $v$, let $q_v$ be the Schmidt rank across the edge separating the rooted subtree $T_v$ from its complement, and choose an orthonormal basis for the corresponding Schmidt support in $\mathcal H_{T_v}$. Proceeding from the leaves toward the root, the Schmidt-support vectors for $T_v$ lie in
\begin{equation}
\mathcal H_v\otimes\bigotimes_{w\in\operatorname{ch}(v)}\mathbb C^{q_w},
\end{equation}
where $\mathbb C^{q_w}$ represents the Schmidt support of the child subtree $T_w$. Expanding an orthonormal basis of the Schmidt support of $T_v$ in this tensor-product space therefore defines an isometry
\begin{equation}
\widetilde A_v:\mathbb C^{q_v}\longrightarrow
\mathcal H_v\otimes\bigotimes_{w\in\operatorname{ch}(v)}\mathbb C^{q_w}.
\end{equation}
Since $q_v\leq r_v$, embed each $\mathbb C^{q_v}$ into $\mathbb C^{r_v}$. The isometry $\widetilde A_v$ can then be extended to an isometry
\begin{equation}
A_v:\mathbb C^{r_v}\longrightarrow
\mathcal H_v\otimes\bigotimes_{w\in\operatorname{ch}(v)}\mathbb C^{r_w},
\end{equation}
because
\begin{equation}
r_v\leq d\prod_{w\in\operatorname{ch}(v)}r_w.
\label{eq:ttn-isometry-dimension-inequality}
\end{equation}
The added dimensions do not change the represented state, since the contraction remains within the embedded Schmidt supports. At the root, the remaining tensor is the expansion of the full state in $\mathcal H_r$ and the Schmidt-support bases of its children, and hence has unit norm.

The same construction applies to a target of TTN bond dimension at most $K$ and norm at most one, with $K$ in place of $D$. In this case the root tensor has norm equal to the norm of the target, and therefore at most one.
\end{proof}

Let $v$ be a target of bond at most $K$ and let $\psi$ be a candidate of bond at most $D$, both in rooted isometric form. Contracting the two subtrees rooted at a non-root vertex $x$ and leaving the parent virtual legs open defines $X_x:\mathbb C^{r_x}\longrightarrow\mathbb C^{k_x}.$ These environments obey
\begin{equation} X_x=V_x^\dagger\left(I_{\mathcal H_x}\otimes\bigotimes_{w\in\operatorname{ch}(x)}X_w\right)A_x. \label{eq:ttn-cross-environment-update} \end{equation}
At the root, the same contraction is $\langle v|\psi\rangle$.

\begin{lemma}
\label{lem:ttn-environment-stability}
If $V,A,A'$ are contractions, then
\begin{align}
\|V^\dagger(I\otimes Z)A-V^\dagger(I\otimes Z')A\|_\infty &\leq\|Z-Z'\|_\infty, \label{eq:ttn-environment-contraction} \\
\|V^\dagger(I\otimes Z)A-V^\dagger(I\otimes Z)A'\|_\infty &\leq\|Z\|_\infty\|A-A'\|_\infty. \label{eq:ttn-environment-local-lipschitz}
\end{align}
Moreover, if $\|X_j\|_\infty,\|Y_j\|_\infty\leq 1$, then
\begin{equation} \left\|\bigotimes_{j=1}^t X_j-\bigotimes_{j=1}^t Y_j\right\|_\infty\leq\sum_{j=1}^t\|X_j-Y_j\|_\infty. \label{eq:ttn-tensor-product-lipschitz} \end{equation}
\end{lemma}

\begin{proof}
The first two inequalities are immediate from submultiplicativity. For the last inequality, telescope the difference one tensor factor at a time
\begin{equation*}
\bigotimes_{j=1}^t X_j-\bigotimes_{j=1}^t Y_j=\sum_{\ell=1}^t
\left(\bigotimes_{j<\ell}Y_j\right)
\otimes (X_\ell-Y_\ell)\otimes
\left(\bigotimes_{j>\ell}X_j\right).
\end{equation*}
Indeed, consecutive terms in this sum cancel, leaving only
$\bigotimes_{j=1}^t X_j-\bigotimes_{j=1}^t Y_j$.
Using the triangle inequality and multiplicativity of the operator norm
under tensor products,
\begin{align*}
\left\| \bigotimes_{j=1}^t X_j-\bigotimes_{j=1}^t Y_j\right\|_\infty &\leq \sum_{\ell=1}^t \left(\prod_{j<\ell}\|Y_j\|_\infty\right) \|X_\ell-Y_\ell\|_\infty
\left(\prod_{j>\ell}\|X_j\|_\infty\right)\\ &\leq \sum_{\ell=1}^t\|X_\ell-Y_\ell\|_\infty,
\end{align*}
where the last inequality uses
$\|X_j\|_\infty,\|Y_j\|_\infty\leq 1$.
\end{proof}

\begin{theorem}
\label{thm:ttn-bounded-target-optimizer}
Let $v$ be an explicitly given TTN on $T$ of bond dimension at most $K$ and norm at most one. For every $D\geq 1$ and every $\eta\in(0,1)$, there is a deterministic algorithm which outputs a normalized $\widehat\psi\in\ttn_{T,D}$ such that
\begin{equation} |\langle v|\widehat\psi\rangle|\geq\max_{\psi\in\ttn_{T,D}}|\langle v|\psi\rangle|-\eta. \label{eq:ttn-bounded-target-guarantee} \end{equation}
Its runtime is bounded by
\begin{equation} n\left(\frac{Cn\sqrt{KD}}{\eta}\right)^{O(\Delta_TKD+dD^{\Delta_T+1})}\poly(K,D,d) \label{eq:ttn-bounded-target-runtime} \end{equation}
for a universal constant $C$.
\end{theorem}

\begin{proof}
Put the target in rooted isometric form. By Lemma~\ref{lem:ttn-environment-stability} and induction from the leaves, every exact cross environment has operator norm at most one. For every non-root vertex $x$, choose an operator-norm $q$-net $\mathcal N_x^{\mathrm{env}}$ for the operator-norm unit ball in $\mathbb C^{k_x\times r_x}$ and an operator-norm $h$-net $\mathcal N_x^{\mathrm{loc}}$ for the candidate isometries. The nets satisfy
\begin{equation} |\mathcal N_x^{\mathrm{env}}|\leq\left(\frac{C\sqrt{KD}}{q}\right)^{2KD},\qquad |\mathcal N_x^{\mathrm{loc}}|\leq\left(\frac{C\sqrt D}{h}\right)^{O(dD^{\Delta_T+1})}. \label{eq:ttn-net-sizes} \end{equation}
At the root, use an $h$-net for the unit root tensors.

The dynamic program proceeds from the leaves toward the root. For each choice of one retained subtree from every child and each local tensor in $\mathcal N_x^{\mathrm{loc}}$, compute the exact environment using \eqref{eq:ttn-cross-environment-update}. Assign it to a nearest point of $\mathcal N_x^{\mathrm{env}}$ and retain one exact subtree in every occupied cell. At the root, enumerate the retained child choices and root tensors and return the candidate with largest exact absolute overlap.

Let $\psi^\star$ be an optimal bond-$D$ TTN. Choose local net tensors within operator norm $h$ of its tensors and follow this path through the dynamic program. If $e_x$ is the operator-norm error of the surviving environment at $x$, then Lemma~\ref{lem:ttn-environment-stability} and the cell-merging error give
\begin{equation} e_x\leq\sum_{w\in\operatorname{ch}(x)}e_w+h+2q. \label{eq:ttn-environment-error-recursion} \end{equation}
If $N_x:=|T_x|$, induction gives $e_x\leq N_x(h+2q).$ There is no merging at the root, so the final overlap error is at most $nh+2(n-1)q.$ Choose
\begin{equation} h:=\frac{\eta}{4n},\qquad q:=\frac{\eta}{8n}. \label{eq:ttn-net-parameters} \end{equation}
The surviving path then has absolute overlap at least the optimum minus $\eta/2$, and the returned candidate satisfies \eqref{eq:ttn-bounded-target-guarantee}.

At each vertex one enumerates at most $\Delta_T$ child environments and one local tensor. Combining the net sizes in \eqref{eq:ttn-net-sizes} with \eqref{eq:ttn-net-parameters} proves \eqref{eq:ttn-bounded-target-runtime}.
\end{proof}

\subsection{Proper optimization against arbitrary explicit targets}

\begin{theorem}
\label{thm:ttn-arbitrary-target-optimizer}
Let $u$ be a normalized explicit TTN on $T$ of arbitrary bond dimension $B$. For every $D\geq 1$ and every $\eta\in(0,1)$, there is a deterministic algorithm which outputs a normalized $\widehat\psi\in\ttn_{T,D}$ such that
\begin{equation} |\langle u|\widehat\psi\rangle|^2\geq\max_{\psi\in\ttn_{T,D}}|\langle u|\psi\rangle|^2-\eta. \label{eq:ttn-arbitrary-target-guarantee} \end{equation}
For fixed $d,D,\Delta_T,\eta$, its runtime is polynomial in $n$ and $B$.
\end{theorem}

\begin{proof}
Choose $K+1\geq64D/\eta^2$. By Theorem~\ref{thm:ttn-comparator-dual-compression}, construct a TTN $\widetilde u_K$ of bond at most $K$ satisfying
\begin{equation} \sup_{\psi\in\ttn_{T,D}}|\langle\psi|u-\widetilde u_K\rangle|\leq\frac{\eta}{8},\qquad \|\widetilde u_K\|_2\leq 1. \label{eq:ttn-arbitrary-target-compression-error} \end{equation}
Apply Theorem~\ref{thm:ttn-bounded-target-optimizer} to $\widetilde u_K$ with amplitude error $\eta/4$. The same triangle inequality and difference-of-squares argument as in the proof of Theorem~\ref{thm:mps-arbitrary-target-optimizer} gives \eqref{eq:ttn-arbitrary-target-guarantee}.
\end{proof}

\begin{corollary}[Proper TTN optimization against an ordered MPS target]
\label{cor:ttn-mps-target-optimizer}
Let $u$ be a normalized explicit MPS of bond $B$ in the ordering $\pi$. For every $D\geq 1$ and every $\eta\in(0,1)$, one can output a normalized $\widehat\psi\in\ttn_{T,D}$ satisfying
\begin{equation} |\langle u|\widehat\psi\rangle|^2\geq\max_{\psi\in\ttn_{T,D}}|\langle u|\psi\rangle|^2-\eta. \label{eq:ttn-mps-target-guarantee} \end{equation}
For fixed $d,D,\Delta_T,\eta$, the runtime is polynomial in $n$ and $B$.
\end{corollary}

\begin{proof}
Convert $u$ to a TTN of bond at most $B^2$ using Lemma~\ref{lem:ttn-mps-comparison}, and apply Theorem~\ref{thm:ttn-arbitrary-target-optimizer}.
\end{proof}

\subsection{Proper agnostic learning of bounded-degree TTNs}

\begin{theorem}[Proper agnostic learning of bounded-degree TTNs]
\label{thm:proper-agnostic-ttn}
For every rooted tree $T$ with $n$ vertices and maximum degree at most $\Delta_T$, every $d,D\in\mathbb N$, and every $\varepsilon,\delta\in(0,1/4)$, there is a randomized quantum algorithm which, given copies of an arbitrary density operator $\rho$ on $\mathcal H_T$, outputs a classical description of a normalized state $\widehat\psi\in\ttn_{T,D} $ such that, with probability at least $1-\delta$,
\begin{equation} \langle\widehat\psi|\rho|\widehat\psi\rangle\geq\max_{\psi\in\ttn_{T,D}}\langle\psi|\rho|\psi\rangle-\varepsilon. \label{eq:ttn-proper-agnostic-guarantee} \end{equation}
With $w_T:=\Delta_T(1+\lceil\log_2n\rceil)$, the copy complexity is
bounded by $\poly\!\left(n,d,D^{w_T},1/\varepsilon,\log(1/\delta)\right).$
For fixed $d,D,\Delta_T,\varepsilon$, the runtime is
bounded by $n^{F(d,D,\Delta_T,1/\varepsilon)}\poly(\log(1/\delta))$
for an explicit computable function $F$ independent of $n$.
\end{theorem}

\begin{proof}
Fix the ordering $\pi=(v_1,\ldots,v_n)$ from Lemma~\ref{lem:ttn-heavy-last-order}. For $1\leq i\leq n-1$, write
$\mathcal H_{\pi,1:i}:=\bigotimes_{j=1}^i\mathcal H_{v_j}$ and
$\mathcal H_{\pi,i+1:n}:=\bigotimes_{j=i+1}^n\mathcal H_{v_j}$.
Apply Definition~\ref{def:properization-data} with
\begin{equation}
\mathcal C:=\ttn_{T,D},\qquad\mathcal V_R:= \left\{u\in\mathcal H_T:\operatorname{rank}_{\mathcal H_{\pi,1:i}\,|\,\mathcal H_{\pi,i+1:n}}(u)\leq R \text{ for every }1\leq i\leq n-1 \right\}.
\label{eq:ttn-properization-pair}
\end{equation}
By Lemma~\ref{lem:mps-cut-rank-characterization}, $\mathcal V_R$ is precisely the class of possibly unnormalized open-boundary MPSs of bond dimension at most $R$ in the ordering $\pi$, and hence $\mathcal A_R=\mps_R^\pi$. Lemma~\ref{lem:ttn-mps-comparison} therefore gives $\mathcal C\subseteq\mathcal A_{R'}$ with $R':=D^{w_T}$. Lemma~\ref{lem:mps-linear-combinations} supplies ambient linear combination closure with
\begin{equation}
\Lambda(R_1,\ldots,R_q):=\sum_{a=1}^q R_a.
\label{eq:ttn-ambient-lambda}
\end{equation}
Theorem~\ref{thm:mps-improper-learner} supplies the improper ambient learner, Corollary~\ref{cor:ttn-mps-target-optimizer} supplies the proper pure-target optimizer, and Lemma~\ref{lem:mps-explicit-operations} supplies the remaining explicit operations. A fixed product state belongs to $\mathcal C$. Therefore Theorem~\ref{thm:general-improper-to-proper} gives \eqref{eq:ttn-proper-agnostic-guarantee}.

The comparator parameter remains $R':=D^{w_T} =D^{\Delta_T(1+\lceil\log_2n\rceil)}$ in every improper-learning call.
Each learner output therefore has MPS bond $B_0=\poly(n,d,R',1/\varepsilon)$. There are $m=O(\varepsilon^{-2})$ accepted outputs. Representing all orthonormal directions and pure targets directly as linear combinations of these original outputs
keeps their MPS bonds at most $mB_0$. The quantum stage uses only the MPS explicit operations and the fixed-comparator improper learner. Its copy complexity is consequently $\poly(n,d,R',1/\varepsilon,\log(1/\delta))$. For fixed $D$ and $\Delta_T$, the parameter $R'$ is polynomial in $n$.

The subsequent computation is classical. Each pure target
is converted to a TTN of bond at most $(mB_0)^2$ and
optimized using Theorem~\ref{thm:ttn-arbitrary-target-optimizer}.
For fixed tree degree, this conversion has size polynomial in $n$ and $mB_0$. The coefficient search uses $(C/\varepsilon)^{O(1/\varepsilon)}$ targets.
\end{proof}
\section{Learning with Branch Structure}
\label{sec:learning-with-branching-structure}

We consider coherent superpositions of $r$ normalized bond-$D$ MPS branches. A single MPS representation need not reveal this decomposition, even when the components represent distinct physical alternatives, as in the two product branches of a GHZ state. Our goal is to return the branches and their coefficients. We require small interference under operators supported on at most $k$ sites so that few-site observables have little sensitivity to coherence between the branches. In the zero-interference case, such observables have identical expectations in the superposition and the corresponding incoherent mixture, although the branches themselves may be locally distinguishable. The algorithm preserves the branch count and bond bounds, while allowing the interference to increase by at most a prescribed amount.

\subsection{Locally noninterfering superpositions}

For $A\subseteq[n]$, write $\overline A:=[n]\setminus A$. Given normalized states $\psi_a,\psi_b$, define
\begin{equation}
\rho_{a,\overline A}:=\tr_A(|\psi_a\rangle\langle\psi_a|),\qquad\Gamma_{ab}^{A}:=\tr_{\overline A}(|\psi_b\rangle\langle\psi_a|).
\label{eq:branching-reduced-and-transition-operators}
\end{equation}
Expanding in product bases gives
\begin{equation}
\|\Gamma_{ab}^{A}\|_2^2=\tr(\rho_{a,\overline A}\rho_{b,\overline A}),\qquad\langle\psi_a|O_A|\psi_b\rangle=\tr(O_A\Gamma_{ab}^{A}).
\label{eq:branching-local-interference-identities}
\end{equation}
Thus $\Gamma_{ab}^{A}=0$ precisely when no operator on $A$ connects the branches. Equivalently, their reduced states on $\overline A$ have orthogonal supports, so the branches remain perfectly distinguishable after discarding $A$. For example, $|0\rangle^{\otimes n}$ and $|1\rangle^{\otimes n}$ have this property whenever $A\neq[n]$. Approximately, Hilbert--Schmidt duality gives
\begin{equation}
|\langle\psi_a|O_A|\psi_b\rangle|\leq\|O_A\|_2\|\Gamma_{ab}^{A}\|_2\leq d^{|A|/2}\|O_A\|_\infty\|\Gamma_{ab}^{A}\|_2. 
\label{eq:branching-approximate-local-interference}
\end{equation}

For $k\in\{0,\ldots,n\}$ and an ordered branch tuple $\boldsymbol\psi=(\psi_1,\ldots,\psi_r)$, define the interference
\begin{equation}
\mathcal I_k(\boldsymbol\psi):=\sum_{a\neq b}\sum_{\substack{A\subseteq[n]\\|A|\leq k}}\|\Gamma_{ab}^{A}\|_2^2,\qquad G(\boldsymbol\psi)_{ab}:=\langle\psi_a|\psi_b\rangle.
\label{eq:branching-total-interference}
\end{equation}
The sum is unnormalized and includes $A=\varnothing$. Hence,
\begin{equation}
\|G-I_r\|_\infty\leq\|G-I_r\|_{\mathrm F}\leq\sqrt{\mathcal I_k(\boldsymbol\psi)}.
\label{eq:branching-gram-perturbation}
\end{equation}
Fix $\mu_0:=1/64$. Whenever $\mathcal I_k(\boldsymbol\psi)\leq\mu_0$, we therefore have
\begin{equation}
\frac78I_r\preceq G\preceq\frac98I_r.
\label{eq:branching-gram-conditioning}
\end{equation}

\begin{definition}[Locally noninterfering superpositions]
\label{def:branching-superposition-class}
For $0\leq\mu\leq\mu_0$, let $\mathcal L_{r,D,k}(\mu)$ consist of states
\begin{equation}
|\Phi\rangle=\sum_{a=1}^r c_a|\psi_a\rangle,\qquad \psi_a\in\mps_D,\qquad \mathcal I_k(\boldsymbol\psi)\leq\mu,\qquad c^\dagger Gc=1.
\label{eq:branching-class-constraints}
\end{equation}
For $X\succeq0$, write
\begin{equation}
\opt_{\mathcal L_{r,D,k}(\mu)}(X):=
\max_{\Phi\in\mathcal L_{r,D,k}(\mu)}
\langle\Phi|X|\Phi\rangle.
\label{eq:branching-optimum}
\end{equation}
\end{definition}

This class is compact since $\mps_D$ is compact, the constraints are closed, and \eqref{eq:branching-gram-conditioning} gives $\|c\|_2^2\leq8/7$. By Lemma~\ref{lem:mps-linear-combinations}, it is contained in $\mps_{rD}$.

\subsection{Recursive evaluation of the interference}

The quantities $\|\Gamma_{ab}^{A}\|_2^2$ can be evaluated recursively as the branches are contracted from left to right. Since the total interference sums these quantities over all subsets $A$ of size at most $k$, we group subsets by their cardinality and track one averaged contraction for each $|A|=0,\ldots,k$.

Put each branch in a padded form as in Lemma~\ref{lem:mps-sequential-isometric-form} with local isometries
\begin{equation}
A_{a,i}:\mathbb C^{\chi_{i-1}}\longrightarrow\mathbb C^d\otimes\mathbb C^{\chi_i},\qquad \chi_i:=\min\{D,d^i,d^{n-i}\}.
\label{eq:branching-padded-branch-dimensions}
\end{equation}
Let $S_i$ swap the two physical copies of site $i$ and let $B_{ab,i}:=A_{a,i}\otimes A_{b,i}$ with the physical factors placed first. For $O\in\{I,S_i\}$, define
\begin{equation}
\mathcal F_{ab,i}^{O}(X):=\tr_{\mathrm{phys}}\!\left[(O\otimes I)B_{ab,i}X B_{ab,i}^\dagger\right].
\label{eq:branching-swap-transfer}
\end{equation}
This map updates the contraction of the two branches while leaving the outgoing virtual indices open. In the reduced-state overlap, a site in $A$ is discarded and contributes an identity while a site outside $A$ contributes a swap. Let $H_{i,\ell}^{ab}$ be the average partial contraction over all $\ell$-element subsets $A\subseteq[i]$, with identities on $A$ and swaps elsewhere. Initialize $H_{0,0}^{ab}:=1$ and set invalid-index messages to zero.

To update this average, split the subsets according to whether they contain site $i$. If $i\notin A$, the first $i-1$ sites contain all $\ell$ discarded sites, and we apply $\mathcal F_{ab,i}^{S_i}$ to $H_{i-1,\ell}^{ab}$. If $i\in A$, the first $i-1$ sites contain $\ell-1$ discarded sites and we apply $\mathcal F_{ab,i}^{I}$ to $H_{i-1,\ell-1}^{ab}$. Among uniformly chosen $\ell$-element subsets, these cases occur with probabilities $(i-\ell)/i$ and $\ell/i$. Thus, for $1\leq i\leq n$ and $0\leq\ell\leq\min\{k,i\}$,
\begin{equation}
H_{i,\ell}^{ab}:=\frac{i-\ell}{i}\mathcal F_{ab,i}^{S_i}(H_{i-1,\ell}^{ab})+\frac{\ell}{i}\mathcal F_{ab,i}^{I}(H_{i-1,\ell-1}^{ab}).
\label{eq:branching-h-recursion}
\end{equation}
Averaging rather than summing makes the update a convex combination of contractions, keeping every message bounded independently of $i$. The factors $\binom n\ell$ are restored only when recovering the original interference sum at the final cut.

\begin{lemma}
\label{lem:branching-finite-memory}
For $O\in\{I,S_i\}$, $\mathcal{F}_{ab,i}^{O}(X)$ is trace-norm contractive. If each local branch isometry changes by at most $h$ in operator norm, then
\begin{equation}
\|\mathcal F_{ab,i}^{O}-\widetilde{\mathcal F}_{ab,i}^{O}\|_{1\to1}\leq4h.
\label{eq:branching-swap-transfer-stability}
\end{equation}
Moreover, $\|H_{i,\ell}^{ab}\|_1\leq1$, and the final scalar messages satisfy
\begin{equation}
\begin{aligned}
H_{n,\ell}^{ab}&=\frac{1}{\binom n\ell}\sum_{\substack{A\subseteq[n]\\|A|=\ell}}\|\Gamma_{ab}^{A}\|_2^2,\\
\mathcal I_k(\boldsymbol\psi)&=\sum_{a\neq b}\sum_{\ell=0}^k\binom n\ell H_{n,\ell}^{ab}.
\end{aligned}
\label{eq:branching-interference-from-h}
\end{equation}
\end{lemma}

\begin{proof}
Since $\|B_{ab,i}\|_\infty=\|O\|_\infty=1$, Schatten duality gives
\begin{equation}
\begin{aligned}
\|\mathcal F_{ab,i}^{O}(X)\|_1&=\sup_{\|W\|_\infty\leq1}\left|\tr\!\left[B_{ab,i}^\dagger(O\otimes W^\dagger)B_{ab,i}X\right]\right|\\
&\leq\sup_{\|W\|_\infty\leq1}\left\|B_{ab,i}^\dagger(O\otimes W^\dagger)B_{ab,i}\right\|_\infty\|X\|_1\\
&\leq\|X\|_1.
\end{aligned}
\label{eq:branching-swap-duality}
\end{equation}
For perturbed branch isometries,
\begin{equation}
\|B_{ab,i}-\widetilde B_{ab,i}\|_\infty\leq\|A_{a,i}-\widetilde A_{a,i}\|_\infty+\|A_{b,i}-\widetilde A_{b,i}\|_\infty\leq2h.
\label{eq:branching-paired-isometry-perturbation}
\end{equation}
Applying the same duality argument,
\begin{equation}
\begin{aligned}
\|\mathcal F_{ab,i}^{O}-\widetilde{\mathcal F}_{ab,i}^{O}\|_{1\to1}
&\leq\sup_{\|W\|_\infty\leq1}\Big(
\left\|(B_{ab,i}-\widetilde B_{ab,i})^\dagger(O\otimes W^\dagger)B_{ab,i}\right\|_\infty\\
&\qquad+\left\|\widetilde B_{ab,i}^\dagger(O\otimes W^\dagger)(B_{ab,i}-\widetilde B_{ab,i})\right\|_\infty\Big)\\
&\leq2\|B_{ab,i}-\widetilde B_{ab,i}\|_\infty\leq4h.
\end{aligned}
\label{eq:branching-swap-perturbation-proof}
\end{equation}

For $A\subseteq[i]$, let
\begin{equation}
M_i^{ab}(A):=\left(\mathcal F_{ab,i}^{O_i(A)}\circ\cdots\circ\mathcal F_{ab,1}^{O_1(A)}\right)(1),\qquad O_j(A):=\begin{cases}I,&j\in A,\\S_j,&j\notin A.\end{cases}
\label{eq:branching-subset-contraction}
\end{equation}
Contraction gives $\|M_i^{ab}(A)\|_1\leq1$. For $0\leq\ell\leq\min\{k,i\}$, induction in \eqref{eq:branching-h-recursion} gives
\begin{equation}
\begin{aligned}
H_{i,\ell}^{ab}&=\frac{1}{\binom i\ell}\sum_{\substack{A\subseteq[i]\\|A|=\ell}}M_i^{ab}(A),\\
\|H_{i,\ell}^{ab}\|_1&\leq\frac{1}{\binom i\ell}\sum_{\substack{A\subseteq[i]\\|A|=\ell}}\|M_i^{ab}(A)\|_1\leq1.
\end{aligned}
\label{eq:branching-subset-average}
\end{equation}
At $i=n$, using $\tr(S(R\otimes T))=\tr(RT)$,
\begin{equation}
\begin{aligned}
M_n^{ab}(A)&=\langle\psi_a\otimes\psi_b|\prod_{j\notin A}S_j|\psi_a\otimes\psi_b\rangle\\
&=\tr(\rho_{a,\overline A}\rho_{b,\overline A})=\|\Gamma_{ab}^{A}\|_2^2.
\end{aligned}
\label{eq:branching-swap-identity}
\end{equation}
Substituting into \eqref{eq:branching-subset-average} and summing over $a\neq b$ and $\ell\leq k$ proves \eqref{eq:branching-interference-from-h}.
\end{proof}

\subsection{Optimization with a relaxed interference bound}

For fixed feasible branches and a target $v$, optimizing the coefficients reduces to projection onto the branch span. With $z_a:=\langle\psi_a|v\rangle$, we have
\begin{equation}
\max_{c^\dagger Gc=1}\left|\left\langle v\middle|\sum_a c_a\psi_a\right\rangle\right|^2=z^\dagger G^{-1}z.
\label{eq:branching-generalized-rayleigh}
\end{equation}
Indeed, substitute $y=G^{1/2}c$. If $z\neq0$, a maximizing choice of coefficients is
\begin{equation}
c=\frac{G^{-1}z}{\sqrt{z^\dagger G^{-1}z}};
\label{eq:branching-optimal-coefficients}
\end{equation}
if $z=0$, take $c=e_1$.

The comparison class has interference at most $\mu$; the returned decomposition may have interference up to $\mu+\sigma$.

\begin{theorem}
\label{thm:branching-pure-target-optimizer}
Let $\mathcal L_{r,D,k}(\mu)\neq\varnothing$, $\eta\in(0,1)$, and $0<\sigma\leq\mu_0-\mu$. Given a normalized explicit MPS $u$ of bond dimension $B$, a deterministic algorithm returns normalized branches $\widehat\psi_a\in\mps_D$ and coefficients $\widehat c$ such that
\begin{equation}
\widehat\Phi:=\sum_{a=1}^r\widehat c_a\widehat\psi_a\in\mathcal L_{r,D,k}(\mu+\sigma)
\label{eq:branching-pure-target-output}
\end{equation}
and
\begin{equation}
|\langle u|\widehat\Phi\rangle|^2\geq\max_{\Phi\in\mathcal L_{r,D,k}(\mu)}|\langle u|\Phi\rangle|^2-\eta.
\label{eq:branching-pure-target-guarantee}
\end{equation}
For fixed $d,r,D,k,\eta,\sigma$, its arithmetic complexity is polynomial in $n$ and $B$.
\end{theorem}

\begin{proof}
First consider a target $v$ of bond at most $K$ and norm at most one. Represent it by sequential tensors $V_i$ that are isometries except for a boundary tensor absorbing its norm. For $v=0$, set the tensor to zero. The proof of Lemma~\ref{lem:mps-environment-stability} also applies to these contractive target tensors.

Search jointly over the $r$ branches as in Theorem~\ref{thm:mps-bounded-target-optimizer}. Besides the interference messages $H_{i,\ell}^{ab}$, retain the cross environments
\begin{equation}
\begin{aligned}
X_i^a&:=\sum_s V_i^sX_{i-1}^a(A_{a,i}^s)^\dagger,\qquad X_0^a:=1,\\
Y_i^{ab}&:=\sum_s A_{b,i}^sY_{i-1}^{ab}(A_{a,i}^s)^\dagger,\qquad Y_0^{ab}:=1.
\end{aligned}
\label{eq:branching-cross-environments}
\end{equation}
At the final cut these give $z_a=X_n^a$ and $G_{ab}=Y_n^{ab}$. All messages have trace norm at most one, and their total real dimension is
\begin{equation}
M=O\!\left(rKD+r^2(k+1)D^4\right).
\label{eq:branching-joint-environment-dimension}
\end{equation}

Use operator-norm $h$-nets of the local branch isometries and a product grid of trace-norm radius $q$ for the messages. These are constructive finite-dimensional discretizations: rounding a local isometry to Frobenius accuracy $h/2$ and taking the polar factor gives an isometry within operator norm $h$. For the environments, Frobenius grid radius $q/D$ suffices since every matrix difference has rank at most $D^2$. At each nonfinal cut, extend every retained prefix tuple by every local tensor tuple, compute its messages, and keep one exact prefix tuple per occupied grid cell. At the final cut, discard tuples with $\mathcal I_k>\mu+\sigma$ and return one maximizing $z^\dagger G^{-1}z$ with coefficients from \eqref{eq:branching-optimal-coefficients}.

Follow an optimal feasible branch tuple through this search. Let $e_i$ be the maximum trace-norm error among its messages and those of the surviving representative at cut $i$. The local perturbation costs are at most $h$ for $X$, $2h$ for $Y$, and $4h$ for $H$. Since \eqref{eq:branching-h-recursion} is a convex combination of contractions, Lemmas~\ref{lem:mps-environment-stability} and \ref{lem:branching-finite-memory} give
\begin{equation}
e_i\leq e_{i-1}+4h+2q,\qquad e_n\leq n(4h+2q).
\label{eq:branching-environment-error}
\end{equation}
Set
\begin{equation}
N_{n,k}:=\sum_{\ell=0}^k\binom n\ell,\qquad t:=\min\left\{\frac{\eta}{12r},\frac{\sigma}{2r^2N_{n,k}}\right\},\qquad h=q:=\frac{t}{6n}.
\label{eq:branching-net-parameters}
\end{equation}
Then $e_n\leq t$. Writing stars for the optimal tuple and hats for the surviving tuple, \eqref{eq:branching-interference-from-h} yields
\begin{equation}
\mathcal I_k(\widehat{\boldsymbol\psi})\leq\mu+r(r-1)N_{n,k}e_n\leq\mu+\frac{\sigma}{2}.
\label{eq:branching-output-interference-bound}
\end{equation}
In particular, this tuple is not discarded. Moreover, $\|\widehat z-z^\star\|_2\leq\sqrt r\,e_n$ and $\|\widehat G-G^\star\|_\infty\leq r e_n$. Both Gram matrices satisfy \eqref{eq:branching-gram-conditioning}. Using $\|z\|_2\leq\sqrt{\|G\|_\infty}$ and $G^{-1}-(G')^{-1}=G^{-1}(G'-G)(G')^{-1}$ gives
\begin{equation}
\begin{aligned}
\left|\widehat z^\dagger\widehat G^{-1}\widehat z-(z^\star)^\dagger(G^\star)^{-1}z^\star\right|&\leq3\left(\|\widehat z-z^\star\|_2+\|\widehat G-G^\star\|_\infty\right)\\
&\leq6r e_n\leq\frac{\eta}{2}.
\end{aligned}
\label{eq:branching-objective-stability}
\end{equation}
Thus the returned tuple has squared overlap within $\eta$ of optimal. For fixed $d,r,D,k,K,\eta,\sigma$, the message and local-tensor dimensions are fixed, while $N_{n,k}=O_k(n^k)$ and $h^{-1},q^{-1}$ are polynomial in $n$. So the search has polynomial size.

For an arbitrary target $u$, apply Theorem~\ref{thm:mps-comparator-dual-compression} with comparator bond $rD$ and $K=\lceil64rD/\eta^2\rceil$. The compressed target $\widetilde u_K$ has norm at most one and satisfies
\begin{equation}
\sup_{\Phi\in\mps_{rD}}|\langle\Phi|u-\widetilde u_K\rangle|\leq\frac{\eta}{8}.
\label{eq:branching-compression-error}
\end{equation}
Hence replacing $u$ by $\widetilde u_K$ changes each squared overlap by at most $\eta/4$. Apply the bounded-target search with error $\eta/2$, allowing the interference bound to increase by $\sigma$. Both the comparator and returned state belong to $\mps_{rD}$, so transferring their scores back to $u$ costs at most another $\eta/2$. Compression costs $\poly(n,d,B,K)$ arithmetic operations, and $K$ is independent of $n$ and $B$, proving the theorem.
\end{proof}

\subsection{Agnostic learning}

\begin{theorem}[Agnostic learning with a relaxed interference bound]
\label{thm:proper-agnostic-branching}
Let $\mathcal L_{r,D,k}(\mu)\neq\varnothing$ and $0<\sigma\leq\mu_0-\mu$. Given copies of an arbitrary density operator $\rho$ and $\varepsilon,\delta\in(0,1/4)$, a randomized quantum algorithm returns normalized branches $\widehat\psi_a\in\mps_D$ and coefficients $\widehat c$ defining $\widehat\Phi\in\mathcal L_{r,D,k}(\mu+\sigma)$ such that, with probability at least $1-\delta$,
\begin{equation}
\langle\widehat\Phi|\rho|\widehat\Phi\rangle\geq\opt_{\mathcal L_{r,D,k}(\mu)}(\rho)-\varepsilon.
\label{eq:branching-proper-agnostic-guarantee}
\end{equation}
The copy complexity is $\poly(n,d,rD,1/\varepsilon,\log(1/\delta))$. For fixed $d,r,D,k,\varepsilon,\sigma$, the runtime is polynomial in $n$ and $\log(1/\delta)$.
\end{theorem}

\begin{proof}
First obtain a fixed feasible decomposition $\Phi_0\in\mathcal L_{r,D,k}(\mu+\sigma)$ by applying Theorem~\ref{thm:branching-pure-target-optimizer} to a product target with accuracy $1/2$. Use this output on any failure.

Construct the finite core for the full class $\mps_{rD}$, not merely $\mathcal L_{r,D,k}(\mu)$. Using the parameters in \eqref{eq:general-main-parameters}, Theorem~\ref{thm:general-finite-relevant-subspace} and Lemma~\ref{lem:general-spectral-cutoff} give a truncated estimated core $\widehat A_\tau$ of rank $\ell=O(\varepsilon^{-1})$ satisfying
\begin{equation}
\sup_{\Phi\in\mps_{rD}}\left|\langle\Phi|\rho|\Phi\rangle-\langle\Phi|\widehat A_\tau|\Phi\rangle\right|\leq2\sqrt\theta+\xi+\tau.
\label{eq:branching-core-approximation}
\end{equation}
If $\ell=0$, return $\Phi_0$. Otherwise, repeat the coefficient-net construction in Lemma~\ref{lem:general-spectral-cutoff}, calling Theorem~\ref{thm:branching-pure-target-optimizer} on each normalized pure target with error $\zeta/4$, allowing the interference bound to increase by $\sigma$. The targets have squared norm at most two, so rescaling costs at most $\zeta/2$, and the coefficient net costs at most $\zeta/4$. Choosing the output with largest $\widehat A_\tau$-score therefore gives
\begin{equation}
\langle\widehat\Phi|\widehat A_\tau|\widehat\Phi\rangle
\geq
\opt_{\mathcal L_{r,D,k}(\mu)}(\widehat A_\tau)-\zeta.
\label{eq:branching-estimated-core-guarantee}
\end{equation}
Both the comparator class and the relaxed output class lie in $\mps_{rD}$. Applying \eqref{eq:branching-core-approximation} to both gives total error
\begin{equation}
4\sqrt\theta+2(\xi+\tau)+\zeta=\frac{13\varepsilon}{16}<\varepsilon.
\label{eq:branching-total-learning-error}
\end{equation}
The quantum stage is the MPS core construction at fixed comparator bond $rD$. The capped postselection and failure allocation from Theorem~\ref{thm:general-improper-to-proper} give success probability at least $1-\delta$ and the stated copy complexity. As in the proof of Theorem~\ref{thm:proper-agnostic-mps}, expressing all core directions and pure targets directly in the original learner outputs keeps their MPS bonds polynomial in $n,d,rD,1/\varepsilon$. The remaining $(C/\varepsilon)^{O(1/\varepsilon)}$ pure-target calls are classical and have polynomial runtime in $n$ for fixed $d,r,D,k,\varepsilon,\sigma$ by Theorem~\ref{thm:branching-pure-target-optimizer}.
\end{proof}

\subsection{Recovery under identifiability}

The algorithm finds a good decomposition, but need not recover a particular original decomposition. To state a conditional recovery guarantee, write $\mathcal D=((c_a,\psi_a))_{a=1}^r$, retaining the normalized branches even when a coefficient vanishes, and set $\Phi(\mathcal D):=\sum_a c_a\psi_a$ and $\mathcal I_k(\mathcal D):=\mathcal I_k(\boldsymbol\psi)$. Define
\begin{equation}
d_{\mathrm{dec}}(\mathcal D,\mathcal D'):=\min_{\substack{\pi\in S_r\\|\omega|=1}}\left(\sum_{a=1}^r\|c_a\psi_a-\omega c'_{\pi(a)}\psi'_{\pi(a)}\|_2^2\right)^{1/2}.
\label{eq:branching-decomposition-distance}
\end{equation}
This compares weighted components so zero-weight branches are not distinguished. A decomposition $\mathcal D_\star$ of a normalized state $\Phi_\star$ is identifiable with parameters $\tau,\gamma,\nu$ if every decomposition $\mathcal D$ into $r$ normalized bond-$D$ branches with normalized $\Phi(\mathcal D)$ satisfies
\begin{equation}
\mathcal I_k(\mathcal D)\leq\nu,\quad |\langle\Phi(\mathcal D)|\Phi_\star\rangle|^2\geq1-\gamma\quad\Longrightarrow\quad d_{\mathrm{dec}}(\mathcal D,\mathcal D_\star)\leq\tau.
\label{eq:branching-identifiability}
\end{equation}
Identifiability is an additional assumption, not a consequence of the interference bound.

\begin{corollary}[Recovery under identifiability]
\label{cor:branching-decomposition-recovery}
Suppose $\mathcal D_\star=((c_a,\psi_a))_{a=1}^r$ represents a state $\Phi_\star\in\mathcal L_{r,D,k}(\mu)$ and is identifiable with parameters $\tau,\gamma,\nu$. Apply Theorem~\ref{thm:proper-agnostic-branching} to $\rho=|\Phi_\star\rangle\langle\Phi_\star|$ with $\varepsilon\leq\gamma$ and $\mu+\sigma\leq\nu$. With probability at least $1-\delta$, its output satisfies $d_{\mathrm{dec}}(\widehat{\mathcal D},\mathcal D_\star)\leq\tau$. If also $|c_a|^2\geq p_\star>0$ for every $a$ and $\tau\leq\sqrt{p_\star}/2$, then, after one permutation,
\begin{equation}
\big||\widehat c_a|-|c_a|\big|\leq\tau,\qquad \frac12\big\||\widehat\psi_a\rangle\langle\widehat\psi_a|-|\psi_a\rangle\langle\psi_a|\big\|_1\leq\frac{2\tau}{\sqrt{p_\star}}.
\label{eq:branching-component-recovery}
\end{equation}
\end{corollary}

\begin{proof}
The comparator optimum is one so the learning guarantee and identifiability imply the bound on $d_{\mathrm{dec}}$. After matching labels and a common phase, the weighted components $x=\widehat c_a\widehat\psi_a$ and $y=\omega c_a\psi_a$ satisfy $\|x-y\|_2\leq\tau$. The reverse triangle inequality gives the coefficient bound, and $\|y\|_2\geq\sqrt{p_\star}$ ensures $x\neq0$. Finally,
\begin{equation}
\left\|\frac{x}{\|x\|_2}-\frac{y}{\|y\|_2}\right\|_2\leq\frac{2\|x-y\|_2}{\|y\|_2}\leq\frac{2\tau}{\sqrt{p_\star}}.
\label{eq:branching-normalized-component-error}
\end{equation}
The trace distance between pure states is at most the Euclidean distance between any phase choices of their unit vectors, proving \eqref{eq:branching-component-recovery}.
\end{proof}

\bibliographystyle{alpha}
\bibliography{references}
\end{document}